\documentclass[final,onefignum,onetabnum]{siamart171218}

\usepackage{lipsum}
\usepackage{amsfonts}
\usepackage{graphicx}
\usepackage{epstopdf}
\usepackage{algorithmic}
\ifpdf
  \DeclareGraphicsExtensions{.eps,.pdf,.png,.jpg}
\else
  \DeclareGraphicsExtensions{.eps}
\fi

\newsiamremark{remark}{Remark}
\newsiamremark{hypothesis}{Hypothesis}
\crefname{hypothesis}{Hypothesis}{Hypotheses}
\newsiamthm{claim}{Claim}

\title{Series expansions and direct inversion for the Heston model\thanks{Submitted to the editors DATE.
\funding{James Watt Scholarship.}}}

\author{Simon J. A. Malham\thanks{Maxwell Institute for Mathematical Sciences and School of Mathematical and Computer Sciences, Heriot-Watt University, Edinburgh, EH14 4AS, UK
  (\email{S.J.A.Malham@hw.ac.uk}, \email{js46@hw.ac.uk}, \email{A.Wiese@hw.ac.uk}).}
\and Jiaqi Shen\footnotemark[2]
\and Anke Wiese\footnotemark[2]}

\usepackage{amsopn}

\usepackage{braket,amsfonts}
\usepackage{amssymb}

\usepackage{xcolor}

\usepackage{array}

\usepackage[caption=false]{subfig}

\usepackage{pgfplots}
\usepackage{multirow}
\usepackage{booktabs}

\usepackage{algorithmic}

\usepackage{graphicx,epstopdf}

\Crefname{ALC@unique}{Line}{Lines}

\usepackage{amsopn}
\DeclareMathOperator{\Range}{Range}

\usepackage{subfig}
\usepackage{extarrows} 
\usepackage{mathrsfs}   
\usepackage{mathtools}
\usepackage[section]{placeins}  

\usepackage{xspace}
\usepackage{bold-extra}
\usepackage[most]{tcolorbox}

\colorlet{texcscolor}{blue!50!black}
\colorlet{texemcolor}{red!70!black}
\colorlet{texpreamble}{red!70!black}
\colorlet{codebackground}{black!25!white!25}

\newcommand\bs{\symbol{'134}} 
\newcommand{\preamble}[2][\small]{\textcolor{texpreamble}{#1\texttt{#2 \emph{\% <- Preamble}}}}

\lstdefinestyle{siamlatex}{%
  style=tcblatex,
  texcsstyle=*\color{texcscolor},
  texcsstyle=[2]\color{texemcolor},
  keywordstyle=[2]\color{texemcolor},
  moretexcs={eq:rho_taylor,eq:rho_taylor,maketitle,mathcal,text,headers,email,url},
}

\tcbset{%
  colframe=black!75!white!75,
  coltitle=white,
  colback=codebackground, 
  colbacklower=white, 
  fonttitle=\bfseries,
  arc=0pt,outer arc=0pt,
  top=1pt,bottom=1pt,left=1mm,right=1mm,middle=1mm,boxsep=1mm,
  leftrule=0.3mm,rightrule=0.3mm,toprule=0.3mm,bottomrule=0.3mm,
  listing options={style=siamlatex}
}

\newtcblisting[use counter=example]{example}[2][]{%
  title={Example~\thetcbcounter: #2},#1}

\newtcbinputlisting[use counter=example]{\examplefile}[3][]{%
  title={Example~\thetcbcounter: #2},listing file={#3},#1}

\DeclareTotalTCBox{\code}{ v O{} }
{ 
  fontupper=\ttfamily\color{black},
  nobeforeafter,
  tcbox raise base,
  colback=codebackground,colframe=white,
  top=0pt,bottom=0pt,left=0mm,right=0mm,
  leftrule=0pt,rightrule=0pt,toprule=0mm,bottomrule=0mm,
  boxsep=0.5mm,
  #2}{#1}

\patchcmd\newpage{\vfil}{}{}{}
\begin{tcbverbatimwrite}{tmp_\jobname_header.tex}
\title{Series Expansions and direct inversion for the Heston model\thanks{Submitted to the editors DATE.
\funding{James Watt Scholarship.}}}

\author{Simon J. A. Malham\thanks{Maxwell Institute for Mathematical Sciences and School of Mathematical and Computer Sciences, Heriot-Watt University, Edinburgh, EH14 4AS, UK (\email{S.J.A.Malham@hw.ac.uk}, \email{js46@hw.ac.uk}, \email{A.Wiese@hw.ac.uk}).}
\and Jiaqi Shen\footnotemark[2] 
\and Anke Wiese\footnotemark[2]}

\end{tcbverbatimwrite}
\input{tmp_\jobname_header.tex}

\ifpdf
\hypersetup{ pdftitle={Series Expansions and direct inversion for the Heston model} }
\fi

\begin{document}
\maketitle

\begin{tcbverbatimwrite}{tmp_\jobname_abstract.tex}
\begin{abstract}
Efficient sampling for the conditional time integrated variance process in the Heston stochastic volatility model is key to the simulation of the stock price based on its exact distribution. We construct a new series expansion for this integral in terms of double infinite weighted sums of particular independent random variables through a change of measure and the decomposition of squared Bessel bridges. When approximated by series truncations, this representation has exponentially decaying truncation errors. We propose feasible strategies to largely reduce the implementation of the new series to simulations of simple random variables that are independent of any model parameters. We further develop direct inversion algorithms to generate samples for such random variables based on Chebyshev polynomial approximations for their inverse distribution functions. These approximations can be used under any market conditions. Thus, we establish a strong, efficient and almost exact sampling scheme for the Heston model.

\end{abstract}


\end{tcbverbatimwrite}
\input{tmp_\jobname_abstract.tex}

\section{Introduction}
\label{sec:intro}
Stochastic volatility models involving a pair of stochastic differential equations, with the diffusion term of the first one governed by the evolution of the second equation, are immensely popular in the pricing of derivatives. Among the existing stochastic volatility models, the Heston model plays an important role and is used widely. It can be expressed in the form of a two-dimensional system
\begin{align}
    \frac{d S_t }{S_t} & = \mu \, dt + \sqrt{V_t} \left( \rho \, dW_t^1 + \sqrt{1-\rho^2} \, dW_t^2 \right), \label{eq:asset_pro} \\
    d V_t & = \kappa \left( \theta - V_t \right) \, dt + \sigma \sqrt{V_t} \, dW_t^1, \label{eq:variance_pro}
\end{align}
where $W^1$ and $W^2$ are two independent standard Brownian motions, and $\kappa$, $\theta$, $\sigma$ and typically also $\mu$ are positive constants with $\rho \in [-1,1]$. The component $S$ characterises the dynamics of the stock price while the component $V$ specifies the variances of its returns. The introduction of randomness to the volatility has been used to explain the long-observed features of the implied volatility surface in a self-consistent way. The variance process follows a mean-reverting square-root or Cox-Ingersoll-Ross (CIR) process (Cox, Ingersoll and Ross \cite{cox2005theory}).

Closed form solutions for standard vanilla option prices under the Heston model are available; see Heston \cite{heston1993closed} and Kahl and J{\"a}ckel \cite{kahl2005not}. However for exotic options, especially path-dependent options, such closed form solutions are not known in general and Monte Carlo simulation is often employed. Typically, continuous stochastic processes are approximated by paths simulated on discrete time grids. It is normally natural to consider the Euler-Maruyama scheme which converges weakly with convergence rate one under certain regularity conditions; see Section $14.5$ in Kloeden and Platen \cite{kloeden1999numerical}, or other standard higher-order discretization approaches such as the Milstein \cite{milstein1994numerical} and It\^{o}-Taylor schemes introduced in Chapter $14$ and $15$ in Kloeden and Platen \cite{kloeden1999numerical}; see Section $6.2$ in Glasserman \cite{glasserman2003monte} as well. However, these conditions do not hold in the Heston model, which will be discussed in detail below.

Discretization schemes such as those have several drawbacks for the Heston model. The first issue is that the probability of the discretised variance process becoming negative is nonzero, which will bring considerable biases to the simulation estimators. Correction techniques such as absorption and reflection are designed to overcome this problem, see Gatheral \cite{gatheral2006volatility}, Bossy and Diop \cite{bossy2006efficient} and Higham and Mao \cite{higham2005convergence}. Lord, Koekkoek and Van Dijk \cite{lord2006comparison} unify a large number of traditional correction techniques and design a new scheme, the full truncation method, which seems to perform well in many situations. Taking advantage of the qualitative properties of the true distributions, Andersen \cite{andersen2007efficient} proposes two new time-discretization algorithms based on moment-matching strategies, namely the truncated Gaussian scheme and the quadratic-exponential scheme. These positivity-preserving schemes are reported to have substantial improvements in efficiency and robustness over other existing methods; see Andersen \cite{andersen2007efficient}, Lord, Koekkoek and Van Dijk \cite{lord2006comparison} and Haastrecht and Pelsser \cite{van2010efficient}.

The second issue is related to convergence, which requires the drift and diffusion coefficients to be globally Lipschitz, see Kloeden and Platen \cite{kloeden1999numerical}. However, the square root functions embedded in the Heston model are not Lipschitz. Thus, convergence of these discretization schemes is difficult to establish; see Glasserman \cite{glasserman2003monte} and Andersen \cite{andersen2007efficient}.  Recently, Altmayer and Neuenkirch \cite{altmayer2017discretising} have studied the weak convergence rate for a numerical scheme under the Heston model, which typically reaches order one with mild assumptions. Hefter and Jentzen \cite{Hefter2019arbitrarily} consider the one-dimensional CIR process and show that equidistant discretization methods may have an arbitrary slow convergence rate in the strong sense. See Alfonsi \cite{alfonsi2010high} and Berkaoui, Bossy and Diop \cite{berkaoui2008euler} for more discussions on the convergence of the discretised univariate variance process.
 

Apart from discretization schemes, there are also (almost) exact simulation methods based on the exact distributions of the stock price and variance processes. The respective transition laws of these follow a conditional lognormal distribution and a conditional scaled noncentral chi-square distribution; see Cox, Ingersoll and Ross \cite{cox2005theory}. Broadie and Kaya \cite{broadie2006exact} take this approach to generate sample variance and stock price. They apply an acceptance-rejection method to the noncentral chi-square sampling for the variance process. Malham and Wiese \cite{malham2013chi} propose an exact acceptance-rejection method  and a high-accuracy direct inversion method for the simulation of the generalised Gaussian distribution, which are then applied to the noncentral chi-squared sampling. Haastrecht and Pelsser \cite{van2010efficient} focus on the efficient approximation of the variance. They explore the features of the distribution for the variance process and suggest a cache for its inverse distribution functions, leading to an almost exact simulation scheme.    

To realise the stock price, the key task of Broadie and Kaya \cite{broadie2006exact} is to sample from the time integrated variance conditional at the endpoints, i.e. $\left( \int_0^t V_s \, ds \middle\vert V_0, V_t  \right)$. They build on the results $($2.m$)$ and $($6.d$)$ in Pitman and Yor \cite{pitman1982decomposition} to derive the explicit form for the corresponding characteristic function. Fourier inversion techniques in conjunction with the trapezoidal rule are applied to numerically evaluate the probability distribution function. This is followed by inverse transform sampling to simulate the value of the above integral. Their numerical results imply that the proposed method has a faster convergence rate compared to the Euler scheme with bias-free simulation. 

Because of the dependence on $V_0$ and $V_t$ , Broadie and Kaya \cite{broadie2006exact} compute the characteristic function for each step and path in the Monte Carlo simulation. At the expense of a small bias, Smith \cite{smith2007almost} presents an approximation to the characteristic function, which makes it possible to precalculate and store the values of the characteristic function for all the points required in advance. Glasserman and Kim \cite{glasserman2011gamma} provide another sampling method 
for the time integrated conditional variance, which relies on an explicit representation as infinite sums and mixtures of gamma random variables. When combined with the exact simulation method suggested by Broadie and Kaya \cite{broadie2006exact}, their method is highly effective in terms of both accuracy and computational speed for pricing non-path-dependent options across a full range of model parameter values. 

Motivated by the decomposition in Glasserman and Kim \cite{glasserman2011gamma} $($Theorem $2.2)$, we simplify the variance process to a squared Bessel process by a measure transformation and construct a new series expansion for its conditional integral under the new measure with exponentially decaying truncation errors. We demonstrate that the task of sampling the new series can be largely reduced to simulations of simple random variables, which are independent of any model parameters. We provide highly accurate Chebyshev polynomial approximations to the inverse distribution functions of such random variables and design direct inversion algorithms to generate their samples. Thus, we establish a flexible, efficient and almost exact simulation scheme for the Heston model. To summarise, the advantages of our method are that truncation errors decay exponentially, high-accuracy samples can be generated efficiently by direct inversions and approximations of inverse distribution functions can be used under any market conditions. 

The paper is organised as follows. In \cref{sec:main_results}, we present our new representation under the new probability measure and the acceptance-rejection algorithm for changing back to the original measure. In \cref{sec:simulation}, we detail the simulation methods for each individual part of the representation. We include the derivation of asymptotic expansions for the corresponding distribution functions and the construction of Chebyshev polynomial approximations for their inverse. We apply our method for pricing purposes and we compare its efficiency and accuracy to Glasserman and Kim \cite{glasserman2011gamma} with numerical results reported in \cref{sec:numerical_analysis}. Conclusions are drawn in \cref{sec:conclusion}.

\section{Main results}
\label{sec:main_results}
The method we propose closely follows the lead of Broadie and Kaya \cite{broadie2006exact} and Glasserman and Kim \cite{glasserman2011gamma} with the key difference for the simulation of the conditional integral of the variance process. To complete the understanding of the motivation for sampling from the conditional integral, we first quote some properties with regard to the Heston model.

We start with the variance process governed by \cref{eq:variance_pro}, which is a CIR process (Cox, Ingersoll and Ross \cite{cox2005theory}) with transition probability given explicitly as a scaled noncentral chi-squared distribution. With the degrees of freedom for this process defined to be $\delta \coloneqq 4 \kappa \theta / \sigma^2$, we have  
\begin{align}
 V_t \sim  \frac{\sigma^2 \left( 1-\exp{ \left( -\kappa t \right)} \right) }{4 \kappa} \chi_{\delta}^{ 2} \left(  \frac{4 \kappa \exp{\left( - \kappa t \right)}}{\sigma^2 \left( 1-\exp{ \left(- \kappa t \right)} \right) } V_0 \right),  \qquad t>0, \label{eq:transition_var} 
\end{align}
where $V_0 > 0$ is the initial value and $\chi_{\delta}^{ 2} \left( \lambda \right)$ denotes a noncentral chi-squared random variable with degrees of freedom $\delta$ and noncentrality parameter $\lambda$. This means that conditional on $V_0$, $V_t$ is distributed as $\sigma^2 \left( 1-\exp{ \left( -\kappa t \right)} \right) / \left( 4 \kappa \right)$ multiplied by a noncentral chi-squared distribution with degrees of freedom $\delta$ and noncentrality parameter 
\begin{align*}
    \lambda \coloneqq  \frac{4 \kappa \exp{\left( - \kappa t \right)}}{\sigma^2 \left( 1-\exp{ \left(- \kappa t \right)} \right) } V_0.
\end{align*}
The above law provides a way of exactly simulating $V_t$ from $V_0$, see Broadie and Kaya \cite{broadie2006exact}, Scott \cite{scott1996simulating} and Malham and Wiese \cite{malham2013chi} for details.

By employing the explicit solution of the stock price process \cref{eq:asset_pro} and It\^{o}'s formula, we obtain 
\begin{align*}
    \log{S_t} = \log{S_0}+\mu t -\frac{1}{2} \int_0^t V_s \, ds + \rho \int_0^t \sqrt{V_s} \, dW_s^1 +\sqrt{1-\rho^2} \int_0^t \sqrt{V_s} \, dW_s^2.
\end{align*} 
Integrating the variance process \cref{eq:variance_pro} also gives
\begin{align*}
    \int_0^t \sqrt{V_s} \, dW_s^1 = \int_0^t \frac{1}{\sigma} \left( \, dV_s - \kappa \left( \theta - V_s \right) \,  ds \right)
     = \frac{V_t -V_0 -\kappa \theta t}{\sigma} + \frac{\kappa}{\sigma} \int_0^t V_s \, ds. 
\end{align*}
Combining these two results, Broadie and Kaya \cite{broadie2006exact} observe that given $V_0$, $V_t$ and $\int_0^t V_s \, ds$, the distribution of $\log{\left( S_t/S_0 \right)}$ is Gaussian with known moments since the process $V$ is independent of the Brownian motion $W^2$, i.e. 
\begin{align*}
  \log{\frac{S_t}{S_0}} \sim \text{N} \left(  \mu t +\frac{\rho}{\sigma} \left( V_t-V_0-\kappa \theta t \right)  +\left( \frac{\rho \kappa}{\sigma} -\frac{1}{2} \right)  \int_0^t V_s \, ds ,\left( 1-\rho^2 \right) \int_0^t V_s \, ds       \right). 
\end{align*}

Hence, an exact simulation for the stock price $S_t$ given the initial conditions $S_0$ and $V_0$ is now reduced to sampling a conditional normal random variable given above provided there is a way to sampling from the joint distribution $ \left( V_t, \int_0^t V_s \, ds \right) $. As $V_t$ can be simulated using the transition law in \cref{eq:transition_var}, the main challenge is now to develop a tractable method for sampling from the time integral of the variance process $V_s$ over $\left[ 0,t \right]$ given its endpoints $V_0$ and $V_t$, i.e. 
\begin{align*}
\left( \int_0^t V_s \, ds \middle\vert V_0, V_t \right).    
\end{align*}
In what follows we focus on developing a new representation for the above integral building upon
the decomposition suggested by Glasserman and Kim \cite{glasserman2011gamma}, which applies the decomposition of the squared Bessel bridges from Pitman and Yor \cite{pitman1982decomposition}.

Before proceeding, we reduce the model to a special case by time-rescaling and a measure transformation. First, define $\tilde{A}_t = V_{4t/\sigma^2}$. Then, $\tilde{A}_t$ satisfies the following stochastic differential equation (Glasserman and Kim \cite{glasserman2011gamma})
\begin{align*}
d\tilde{A}_t =  \big( \delta - 2 q \tilde{A}_t \big)  \, dt+ 2 \sqrt{\tilde{A}_t} \, d \tilde{W}_t^1,
\end{align*} 
where $q \coloneqq 2 \kappa/\sigma^2$ and $\tilde{W}_t^1 \coloneqq \sigma W_{4t/{\sigma^2}}^1 / 2$ becomes a standard Brownian motion. In order to consider a new probability measure, we suppose that the original model is established under measure $\mathbb{Q}$, meaning that our target now is to simulate 
\begin{align*}
\left[ \int_0^t V_s \, ds \middle\vert V_0=v_0, V_t=v_t \right] \xlongequal[]{d} \left[ \frac{4}{\sigma^2} \int_{0}^{\tau} \tilde{A}_s \, ds \middle\vert \tilde{A}_0 = a_0, \tilde{A}_{\tau} = a_{\tau} \right]
\end{align*}
under $\mathbb{Q}$ with $\tau=\sigma^2 t/4$, $a_0 = v_0$ and $a_{\tau} = v_t$.

Second, we further simplify the model by introducing a new probability measure $\mathbb{P}$ (Pitman and Yor \cite{pitman1982decomposition} and Glasserman and Kim \cite{glasserman2011gamma}) such that 
\begin{align}
\frac{d \mathbb{P}}{d \mathbb{Q}} = \exp \left( q \int_0^\tau \sqrt{\tilde{A}_s} \, d \tilde{W}_s^1 - \frac{q^2}{2} \int_0^\tau \tilde{A}_s \, ds   \right). \label{eq:radon-nikodym}
\end{align}
 By the Girsanov theorem, $W_\tau^\mathbb{P} \coloneqq \tilde{W}_\tau^1 - \int_0^\tau q \sqrt{\tilde{A}_s} \, ds$ is a standard Brownian motion under $\mathbb{P}$. With this replacement, the rescaled process $\tilde{A}$ satisfies
\begin{align}
d\tilde{A}_t = \delta \, dt + 2 \sqrt{\tilde{A}_t} \, dW_t^\mathbb{P}, \label{eq:squared_Bessel}
\end{align}
which is a $\delta$-dimensional squared Bessel process under $\mathbb{P}$. Hence, our objective is to sample from the time integral of a squared Bessel process $\tilde{A}$ given its values at the endpoints, denoted by $I$, under the new probability measure $\mathbb{P}$, i.e.  
\begin{align}
   I = \left( \int_0^\tau \tilde{A}_s \, ds \middle\vert \tilde{A}_0=a_0, \tilde{A}_\tau =a_\tau \right), \label{eq:integral_P_tildeA}
\end{align}
 and to find a connection between the distributions for the conditional integral under $\mathbb{P}$ and $\mathbb{Q}$. 
 
Next, we state our main result which relies on decomposing the conditional integral as introduced by Glasserman and Kim \cite{glasserman2011gamma}.
\begin{theorem} \label{thm:Series_Rep} 
Under the new probability measure $\mathbb{P}$, the conditional integral 
of the rescaled variance process $\tilde{A}$ is equivalent in distribution to the sum of three infinite series of random variables
\begin{align*}
I = \left( \int_0^\tau \tilde{A}_s \, ds \middle\vert \tilde{A}_0 = a_0, \tilde{A}_\tau = a_\tau \right) \xlongequal[]{d} X_1 + X_2 + \sum_{j=1}^\eta Z_j,
\end{align*}
where $X_1$, $X_2$, $\eta$, $Z_1$, $Z_2$, $\ldots$ are mutually independent, and $\eta$ is a Bessel random variable with parameters $\nu = \delta/2-1$ and $z=\sqrt{a_0 a_\tau}/\tau$, i.e. $\eta \sim \text{Bessel} \left( \nu, z \right)$. Moreover, $X_1$, $X_2$, $Z_1$, $Z_2$, $\ldots$ admit the following representations: \\
$($a$)$ We have
\begin{align*}
    X_1 \xlongequal[]{d} \sum_{n=0}^{\infty} \frac{\tau^2}{4^n} \sum_{k=1}^{P_n} S_{n,k},
\end{align*}
where for $n=0,1, \ldots$, the $P_n$ are independent Poisson random variables with mean $ \left( a_0 + a_\tau \right) 2^{n-1}/ \tau$ and for $k=1,2, \ldots, P_n$, the $S_{n,k}$ are independent copies of the random variable $S \coloneqq \left( 2/{\pi^2} \right) \sum_{l=1}^{\infty} \epsilon_l/{l^2}$ and $\epsilon_l \sim \text{Exp}\left( 1 \right)$ are independent exponential random variables for $l=1,2,\ldots$; \\
$($b$)$ Further we have
\begin{align*}
     X_2 \xlongequal[]{d} \sum_{n=1}^{\infty} \frac{\tau^2}{4^n} C_{n}^{\delta/2},
\end{align*}
where for $n=1, 2,\ldots$, the $C_n^{\delta/2}$ are independent copies of the random variable $C^{\delta/2} \coloneqq \left( 2/{\pi^2} \right) \sum_{l=1}^{\infty} \Gamma_{\delta/2,l}/{ \left( l- 1/2 \right)^2}$ and $\Gamma_{\delta/2, l} \sim \text{Gamma}\left( \delta/{2},1 \right)$ are independent gamma random variables with shape $\delta/2$ and rate $1$ for $l=1,2,\ldots$; \\
$($c$)$ And also we have the $Z_j$, $j=1,2,\ldots, \eta$, which are independent copies of the random variable $Z$ such that
\begin{align*}
       Z \xlongequal[]{d} \sum_{n=1}^{\infty} \frac{\tau^2}{4^n} C_n^{'}, 
\end{align*}
where for $n=1,2, \ldots$, the $C_n^{'}$ are independent copies of the random variable $C^2 \coloneqq \left( 2/{\pi^2} \right) \sum_{l=1}^{\infty} \Gamma_{2,l}/{ \left( l-1/2 \right)^2}$ and $\Gamma_{2,l} \sim \text{Gamma}\left( 2,1 \right)$ are independent gamma random variables  with shape $2$ and rate $1$ for $l=1,2,\ldots$.
\end{theorem}
\begin{proof}
We work on the probability measure $\mathbb{P}$ throughout this proof. We note that for a fixed $\tau>0$,
\begin{align}
\left( \int_0^\tau \tilde{A}_s \, ds \middle\vert \tilde{A}_0 = a_0, \tilde{A}_\tau = a_\tau \right) = \left( \tau^2 \int_0^1 A_s \, ds \middle\vert A_0 =x, A_1 =y \right), \label{eq:integral_P_A}
\end{align}
where $A_s$ is defined by setting $A_s =  \tilde{A}_{s \tau} / \tau$ for $ 0 \leq s \leq 1$ and $x= a_0/ \tau, y = a_\tau/\tau$. Then using equation \cref{eq:squared_Bessel}, the process $A$ satisfies 
\begin{align*}
dA_s = \delta \, ds + 2 \sqrt{A_s} \, dW_s,
\end{align*}
where $W_s \coloneqq  W_{s \tau}^\mathbb{P} / \sqrt{\tau}$ is a standard Brownian motion. We observe that
$ \left\lbrace  A_s \right\rbrace _{0 \leq s \leq 1}$ is a $\delta$-dimensional squared Bessel process. Conditional on the end points, the process
$ \big( A_s, 0 \leq s \leq 1 \big\vert  A_0 = x, A_1 =y  \big) $ is then a squared Bessel bridge, denoted
by $A_{x,y}^{\delta,1} = \left\lbrace A_{x,y}^{\delta,1} \left( s \right) \right\rbrace_{0 \leq s \leq 1}$. Then the right hand side of \cref{eq:integral_P_A} has the same distribution as 
\begin{align}
     \left( \tau^2 \int_0^1 A_{x,y}^{\delta,1} \left( s \right) \, ds \right). \label{eq:integral_P_Bessel}
\end{align}
We prove the result in three steps.

First, the integral \cref{eq:integral_P_Bessel} can be decomposed into the sum of three independent parts as follows:
    \begin{align*}
     \tau^2 \int_0^1 A_{x,y}^{\delta,1} \left( s \right) \, ds \xlongequal[]{d}   X_1' + X_2' + \sum_{j=1}^\eta Z_j', 
    \end{align*}
where 
   \begin{align*}
   X_1' & = \tau^2 \int_0^1 A_{x+y,0}^{0,1} \left( s \right) \, ds, \\
   X_2' & = \tau^2 \int_0^1 A_{0,0}^{\delta, 1} \left( s \right) \, ds, 
   \end{align*}
and for $j=1,2,\ldots, \eta$, $Z_j'$ are independent copies of
\begin{align*}
    Z' = \tau^2 \int_0^1 {A_{0,0}^{4,1}} \left( s \right) \, ds,
\end{align*}
and $\eta$ is an independent Bessel random variable with parameters $\nu = {\delta}/{2} -1$ and $ z = \sqrt{xy} = {\sqrt{a_0 a_\tau}}/{\tau}$, i.e. $\eta \sim \text{Bessel} \left( \nu, z \right)$. This is a direct result from Glasserman and Kim \cite{glasserman2011gamma}, who apply the decomposition of squared Bessel bridges proposed by Pitman and Yor \cite{pitman1982decomposition} to the transformed variance process.

Second, it follows from Glasserman and Kim \cite{glasserman2011gamma} that the Laplace transforms of $X_1'$, $X_2'$ and $Z'$ for $b \geq 0$ are given by  
 \begin{align}
       \Phi'_1 \left(b\right) & =  \exp{ \left( \frac{a_0+a_\tau}{2 \tau} \left( 1-\sqrt{2b} \tau \coth{ \big( \sqrt{2b} \tau  \big) } \right)   \right)},  \label{eq:Laplace1} \\
       \Phi'_2\left(b\right) & =\left( \frac{\sqrt{2b} \, \tau}{\sinh  \big( \sqrt{2b} \, \tau \big) }  \right)^{\delta/2}, \label{eq:Laplace2} \\
       \Phi'_3\left(b\right) & =\left(  \frac{\sqrt{2b} \, \tau}{\sinh  \big(  \sqrt{2b} \, \tau \big)  }  \right)^{2}. \label{eq:Laplace3}
   \end{align}

  Third, to verify the random variables $X'_1$, $X_2'$ and $Z'$ have the same distribution as the series expansions which define $X_1$, $X_2$ and $Z$ respectively, it is sufficient to show that they have identical Laplace transforms. To do this, let us first rewrite $\Phi'_i$, $i=1,2,3$ using some important identities regarding the hyperbolic functions $\coth$ and $\sinh$; see
  Malham and Wiese \cite{malham2014efficient}. Specifically, we observe
   \begin{align*}
       \coth{\zeta} & \equiv \coth{\frac{\zeta}{2}} -\frac{1}{\sinh{\zeta}} , \\
       \sinh{\zeta} & \equiv 2 \sinh{\frac{\zeta}{2}} \cosh{\frac{\zeta}{2}}.
   \end{align*}
   Iterating $N$ times gives us
   \begin{align}
       \coth{\zeta} & \equiv \coth{\frac{\zeta}{2^{N+1}}} - \sum_{n=0}^N \frac{1}{\sinh{\frac{\zeta}{2^n}}}, \label{eq:identity_coth} \\
       \sinh{\zeta} & \equiv 2^N \sinh{\frac{\zeta}{2^N}} \prod_{n=1}^N \cosh{\frac{\zeta}{2^n}}. \label{eq:identity_sinh}
   \end{align}
   Substituting \cref{eq:identity_coth} into \cref{eq:Laplace1} and rearranging the terms, we get
   \begin{align*}
    \exp{ \left(   \frac{a_0 +a_\tau}{2 \tau} \left(  1- \zeta \coth \zeta    \right) \right)   }  = \prod_{n=0}^N \exp{ \left( \frac{a_0+a_\tau}{2 \tau} 2^n  \frac{ \frac{\zeta}{2^n}  }{ \sinh \frac{\zeta}{2^n}}  \right)}  I_{N} \left( \zeta, \tau, a_0, a_\tau \right),
   \end{align*}
    where $I_N \left( \zeta, \tau,a_0,a_\tau \right) \coloneqq \exp{\left( - \left( {a_0+a_\tau} \right) \left( \zeta \coth ({\zeta}/{2^{N+1}}) -1 \right) / \left( {2 \tau}  \right)   \right)} $. On the other hand, it follows from $\sum_{n=0}^N 2^n = 2^{N+1}-1$ that 
    \begin{align*}
         I_N \left( \zeta, \tau, a_0,a_\tau \right) = \prod_{n=0}^N \exp{ \left( -\frac{a_0+a_\tau}{2 \tau} 2^n  \right)  } \exp{ \left( \varepsilon_{1,N} \left( \zeta, \tau, a_0, a_\tau \right) \right)} ,
    \end{align*}
     where $\varepsilon_{1,N} \left( \zeta, \tau, a_0, a_\tau \right) \coloneqq - \left( {a_0+ a_\tau} \right)  \left( \zeta \coth \left( {\zeta}/{2^{N+1}} \right) - 2^{N+1} \right) /\left( {2 \tau} \right) \to 0$ as $ N\to \infty$. Thus, for $\zeta=\sqrt{2b} \, \tau$, we have an alternative form for $\Phi'_1$ given by
    \begin{align*}
    \Phi'_1\left(b\right) = \prod_{n=0}^{\infty} \exp{\left(  \frac{a_0+a_\tau}{2\tau} \, 2^n \left( \frac{\frac{\sqrt{2b} \, \tau}{2^n}}{\sinh \frac{\sqrt{2b} \, \tau}{2^n}} -1  \right) \right)}.
    \end{align*}
    Similarly, for $\Phi'_2$ after substitution and rearrangement, we have 
    \begin{align*}
       \left( \frac{\zeta}{\sinh \zeta} \right)^{\delta/2} = \prod_{n=1}^N \left( \cosh \frac{\zeta}{2^n} \right)^{- \delta/2}  \varepsilon_{2,N} \left( \zeta,\delta \right),
    \end{align*}
    where $\varepsilon_{2,N} (\zeta, \delta) \coloneqq \left( \left( {\zeta/2^N} \right)/{\sinh{ \left( {\zeta/2^N} \right)} }\right)^{\delta/2} \to 1$ as $N \to \infty$. As a result, plugging $\zeta= \sqrt{2b} \, \tau $ into this expression yields
    \begin{align*}
    \Phi'_2\left(b\right) = \prod_{n=1}^{\infty} \left( \cosh \frac{\sqrt{2b} \, \tau}{2^n} \right)^{- \delta/2}.
    \end{align*} 
    Next, we derive the Laplace transforms of $X_1$, $X_2$ and $Z$, denoted by $\Phi_1$, $\Phi_2$ and $\Phi_3$ respectively. For any $b \geq 0$, we have
   \begingroup
    \allowdisplaybreaks
    \begin{align*}
        \Phi_1 \left(b\right) & = \mathbb{E} \left[ \exp{ \left( -b X_1 \right)} \right] \\
                   & = \prod_{n=0}^\infty \mathbb{E} \left[  \exp{ \left( -b \frac{\tau^2}{4^n} \sum_{k=1}^{P_n} S_{n,k} \right) } \right]  \\
                   & =\prod_{n=0}^\infty \mathbb{E} \left[ \prod_{k=1}^{P_n} \mathbb{E} \left( \exp \left( -b \frac{\tau^2}{4^n} S_{n,k} \right) \right) \right] \\ 
                   & = \prod_{n=0}^\infty \mathbb{E} \left[ \left(   \frac{\sqrt{ \frac{2 b {\tau^2}}{4^n} }}{\sinh{\sqrt{ \frac{2b \tau^2}{4^n}}} } \right)^{P_n}  \right] \\ 
                   & = \prod_{n=0}^\infty \exp{ \left(  \frac{a_0+a_{\tau}}{2 \tau} 2^n \left( \frac{\frac{\sqrt{2b} \, \tau}{2^n}}{\sinh{\frac{\sqrt{2b}  \, \tau}{2^n}}} -1 \right) \right)},
    \end{align*}
    \endgroup
where the second equality comes from the interchange of expectation and limit by the Bounded Convergence Theorem and the fourth equality holds due to the fact that $ \mathbb{E} \left[ \exp \left( -b S_{n,k} \right) \right] =\sqrt{2b} / \sinh{\sqrt{2b}} $ for all $ n \geq 0$ and $k \geq 1$ (see Biane, Pitman and Yor \cite[formula $(1.8)$]{biane2001probability}).
    
Following similar arguments, we now determine the Laplace transform $\Phi_2$ for $X_2$. Indeed, from $\mathbb{E} \left[ \exp{ \left( -b C_n^{\delta/2} \right) } \right] = \left( \cosh{\sqrt{2b}} \right)^{-\delta/2}$ for any $ n \geq 1$ (see Biane, Pitman and Yor \cite[formula $(1.8)$]{biane2001probability}, we conclude that
     \begin{align*}
        \Phi_2\left(b\right) & = \mathbb{E} \left[ \exp{ \left( -b X_2 \right) } \right] \\
                  & = \prod_{n=1}^{\infty} \mathbb{E} \left[ \exp{\left( -b \frac{\tau^2}{4^n} C_n^{\delta/2} \right)} \right] \\
                  & = \prod_{n=1}^{\infty} \left( {\cosh{ \frac{\sqrt{2b} \, \tau}{2^n}}} \right)^{-\delta/2}.
    \end{align*}

Hence, we can now deduce that $X_i' \xlongequal[]{d} X_i $ as $\Phi_i' = \Phi_i$ for $i=1,2$. In line with the steps explained above, $Z' \xlongequal[]{d} Z$ follows since this is a special case when $\delta=4$, completing the proof.
\end{proof}

\begin{remark}
We notice that after separating the time parameter $\tau$, the dependence of $X_1$ on model parameters is only though the Poisson random variable $P_n$ and $X_2$ depends only on one parameter $\delta$. This feature provides us with a possibility that the task of sampling the conditional integral $I$ can be largely reduced to the simulations of simple random variables, whose distributions remain unchanged as we change the values for the model parameters; see \textcolor{siaminlinkcolor}{section} \ref{sec:simulation of X_1} and
\textcolor{siaminlinkcolor}{section}
\ref{sec:simulation of X_2}.
\end{remark}

We have represented the conditional time integral $I$ by 
double infinite weighted sums and mixtures of simple independent random variables under the new probability measure $\mathbb{P}$, which serves as a theoretical basis for the exact simulation from the distribution of \cref{eq:integral_P_tildeA} under $\mathbb{P}$. However, our goal is set up under the probability measure $\mathbb{Q}$. We now focus on the task of generating a sample from the distribution of the conditional integral $I$ under $\mathbb{Q}$ 
once we have generated a sample under $\mathbb{P}$. In particular, we explore the relationship between the probability density functions of the integral under these two probability measures. We specify the details in the following theorem.
\begin{theorem} \label{thm:Laplace_connection}
Suppose that $f_P$ and $f_Q$ are the probability density functions of $I$ under the probability measures $\mathbb{P}$ and $\mathbb{Q}$, respectively. Then, we have
\begin{align*}
    f_Q \left( x \right) = L \left( q, \nu, \tau, a_0, a_\tau \right) \exp{\left( - \frac{q^2}{2} x \right)} f_P\left( x \right),
\end{align*}
where 
\begin{align*}
    L\left( q, \nu,\tau, a_0, a_\tau \right) = \frac{\sinh{\left(q \tau \right)}}{ q \tau} \exp{  \left( \frac{a_0+a_\tau}{2 \tau} \left(q \tau    \coth{\left( q \tau \right)   }          -1     \right)   \right)} \frac{ \text{I}_\nu \left(  \frac{\sqrt{a_0 a_\tau}}{\tau}       \right)  }{ \text{I}_\nu \left(   \frac{q \sqrt{a_0 a_\tau}}{\sinh{\left( q \tau \right)}} \right)} 
\end{align*}
with $\text{I}_\nu\left( \cdot \right)$ denoting the modified Bessel function of the first kind.
\end{theorem}
\begin{proof}
We will make use of the shift property of the Laplace transform to justify the theorem. We first establish a connection between their respective Laplace transforms. For any $b \geq 0$, consider the Laplace transform $\mathcal{L}\left\{ f_Q \right\} \left( b \right)$ of $f_Q$ at $b$, which is the $\mathbb{Q}$-expectation of $\exp{ \left( -b I \right)}$. Thus, we get
 \begin{align*}
    \mathcal{L}  \left\{    f_Q  \right\} \left(  b  \right)
    & = \mathbb{E}^\mathbb{Q} \left[ \exp{\left( -b \int_0^\tau \tilde{A}_s \, ds \right)} \middle\vert   \tilde{A}_0 = a_0, \tilde{A}_\tau = a_\tau  \right]  \\
    & = \frac{    \mathbb{E}^\mathbb{P} \left[   \exp{\left( - \left( b +\frac{q^2}{2} \right) \int_0^\tau \tilde{A}_s \, ds \right)}  \middle\vert   \tilde{A}_0 = a_0, \tilde{A}_\tau = a_\tau  \right]}{  \mathbb{E}^\mathbb{P} \left[   \exp{\left( -\frac{q^2}{2} \int_0^\tau \tilde{A}_s \, ds \right)}   \middle\vert   \tilde{A}_0 = a_0, \tilde{A}_\tau = a_\tau  \right]}  \\
    & = \frac{\mathcal{L} \left\{ f_P \right\} \left( b+\frac{q^2}{2} \right)}{\mathcal{L} \left\{ f_P \right\} \left(\frac{q^2}{2} \right)} \\
    & = \mathcal{L}   \left\{ \frac{f_P}{ \mathcal{L} \left\{ f_P \right\} \left( \frac{q^2}{2} \right)  } \right\}  \left( b+\frac{q^2}{2} \right).
\end{align*}
The second equality is a result of the change of law formula $(6.\text{d})$ from Pitman and Yor \cite{pitman1982decomposition}, see the Appendix of Broadie and Kaya \cite{broadie2006exact} as well. Now by the application of the shift property, we can write 
\begin{align*}
    f_Q \left( x\right) = \frac{f_P \left( x \right)}{  \mathcal{L}   \left\{ f_P \right\}  \left( \frac{q^2}{2} \right) } \exp{\left( - \frac{q^2}{2} x \right)}.
\end{align*}
Using the formula $(2.\text{m})$ in Pitman and Yor \cite{pitman1982decomposition} for the Laplace transform $\mathcal{L} \left\{ f_P \right\}$ of $f_P$ at $q^2/2$ given by 
\begin{align*}
    \mathcal{L} \left\{ f_P \right\} \left( \frac{q^2}{2} \right) = \frac{q \tau}{ \sinh{\left( q \tau \right)}} \exp{ \left( \frac{a_0+a_\tau}{2 \tau}  \left( 1-q \tau \coth{\left( q \tau \right)} \right) \right)} \frac{\text{I}_\nu\left( \frac{ q \sqrt{a_0 a_\tau}}{\sinh{\left(  q \tau \right)}} \right)}{\text{I}_\nu\left( \frac{\sqrt{a_0 a_\tau}}{\tau} \right)}
\end{align*}
and setting $L\left( q,\nu, \tau, a_0, a_\tau \right) \coloneqq \left( \mathcal{L} \left\{ f_P \right\} \left( {q^2}/{2} \right) \right)^{-1} $ establishes the stated result.
\end{proof}

The above theorem relates the density $f_P$ of the distribution explicitly given by \cref{thm:Series_Rep} in terms of infinite sums
to the density $f_Q$ of the distribution we are interested in. This means we can simulate
the random variable $I$ under the measure $\mathbb{Q}$ provided we have an observation from its distribution under the measure $\mathbb{P}$. In general, we construct the acceptance-rejection algorithm outlined in \cref{alg:accpet-reject} to generate samples from $f_Q$. 

On average, the probability of accepting a proposed sample is 
\begin{align*}
    \mathbb{P} \left( U \leq \exp{\left( - \frac{q^2 Y}{2} \right)} \right) = \frac{1}{L\left( q,\nu, \tau, a_0, a_\tau \right)},
\end{align*}
where $U \sim \text{Unif} \left( 0, 1 \right)$ and independently $Y$ follows the distribution of $I$ under $\mathbb{P}$. Consequently, we require $L \left( q, \nu, \tau, a_0, a_\tau \right) \geq 1$ due to the fact that a probability only takes values between zero and one. In practice, we prefer a value of $L$ closer to one as it indicates higher acceptance probability on average, and thus fewer iteration steps needed.
\begin{algorithm}[H]
\caption{Acceptance-Rejection}
\label{alg:accpet-reject}
\begin{algorithmic}[1]
\STATE{Simulate a realisation $Y$ of the random variable $I$ under $\mathbb{P}$ using \cref{thm:Series_Rep}.}
\STATE{Obtain a sample $U$ independently from the uniform distribution $\text{Unif} \left( 0,1 \right)$ over unit interval.}
\STATE{If $U \leq \exp{ \left( -q^2 Y/2 \right)} $, accept $Y$ as a sample drawn from the distribution of $I$ under $\mathbb{Q}$; otherwise reject the value of $Y$ and return to the first step.}
\end{algorithmic}
\end{algorithm}

\begin{remark}
The requirement $L \left( q, \nu, \tau, a_0, a_\tau \right) \geq 1$ is fulfilled under any market conditions. In fact, because the support for the random variable $I$ is $\left( 0, + \infty \right)$, we have
  $\int_0^\infty  L \left( q, \nu, \tau, a_0, a_\tau \right) \exp{\left( - {q^2 x}/{2}  \right) } f_P \left( x \right) \, dx =  \int_0^\infty  f_Q \left( x \right) \, dx = 1 $.
Noticing that $ \exp{\left( - {q^2 x}/{2}  \right) } \leq 1 $ for $x \geq 0$ and $0 < \int_0^\infty \exp{\left( - {q^2 x}/{2}  \right) } f_P \left( x \right) \, dx \leq \int_0^\infty f_P \left( x \right) \, dx = 1$ then leads to $ L \left( q, \nu, \tau, a_0, a_\tau \right) \geq 1$.
\end{remark}

So far, we have described the theories behind the simulation of the time integral conditional variance. It is a matter of sampling infinite series, combined with the acceptance-rejection method. For the next stage, we will address some issues concerning the practical implementation of the theory. In particular, strategies are required to deal with infinite summations of random variables. We propose direct inversion algorithms based on approximating the corresponding inverse distribution functions to tackle this problem. Further, the dependence between the series and the model parameters implies that the recomputation of the inverse distribution functions when using a new set of model coefficients is inevitable. We show that the series can be decomposed as the sum of simple random variables, whose distributions do not depend on any market conditions. We specify the details in the next section.
\section{Simulation}
\label{sec:simulation}
In this section, we outline how to generate an exact sample for $I$ under $\mathbb{P}$ by \cref{thm:Series_Rep} introduced earlier. In particular, we discuss the sampling techniques corresponding to $X_1$ and $X_2$. We note that $Z$ is a special case of $X_2$ with $\delta=4$. In order to apply the decomposition theorem to sample the conditional integral, we need to determine a point at which the infinite summation is terminated. We consider the truncation for the outer summation now, leaving the inner one to be discussed further in the following contexts. Let us denote the truncation level by $K$ and the resulting \emph{remainder} random variables of $X_1$ and $X_2$ by $R_1^K$ and $R_2^K$ respectively, i.e.
\begin{align*}
    R_1^K & \coloneqq \sum_{n=K+1}^{\infty} \frac{\tau^2}{4^n} \sum_{k=1}^{P_n} S_{n,k}, \\
    R_2^K & \coloneqq \sum_{n=K+1}^{\infty} \frac{\tau^2}{4^n} C_{n}^{\delta/2}.
\end{align*}
We evaluate the effect of truncation by summarising the means and variances of the remainder terms in the next lemma; see \cref{sec:appendix_proof_3.1} for a detailed proof.
\begin{lemma} 
\label{lem:moments_of_remainder}
Given the truncation level $K > 0$, we have
\begin{align*}
    \mathbb{E}\left[ R_1^K \right] & = \frac{ \left( a_0+a_\tau \right) \tau}{6} \frac{1}{2^K}, & \mathrm{Var}\left[ R_1^K  \right] & = \frac{ \left( a_0+a_\tau \right) \tau^3}{90} \frac{1}{8^K} , \\
    \mathbb{E} \left[ R_2^K \right] & = \frac{\delta \tau^2}{6} \frac{1}{4^K}, & \mathrm{Var} \left[ R_2^K \right] & = \frac{\delta \tau^4}{45} \frac{1}{16^K}.
\end{align*}
\end{lemma}
\begin{remark}
The above lemma implies that the truncation errors decay exponentially. This is an appealing property of the new series in \cref{thm:Series_Rep} as the truncation error will decrease so quickly that the Monte Carlo error will dominate the total error even for small truncation level $K$. Hence, including the terms at lower levels will be enough to produce an accurate approximation. This is supported by our numerical simulations in \cref{sec:numerical_analysis}.
\end{remark}
\subsection{Simulation of \boldmath$X_1$}
\label{sec:simulation of X_1}
Recall that by dropping the remainders, we approximate $X_1$ by ${X}_1^K$ where
\begin{align*}
    {X}_1^K = \sum_{n=0}^{K} \frac{\tau^2}{4^n} \sum_{k=1}^{P_n} S_{n,k}.
\end{align*}
Notice that the $S_{n,k}$ are independently and identically distributed as $S = \left( 2/{\pi^2} \right) \sum_{l=1}^{\infty}$ $\epsilon_l/{l^2} $.
To reduce the truncation error further, we simulate the tail sum $R_1^K$ as well. Glasserman and Kim \cite{glasserman2011gamma} use the central limit theorem to show the validity of a normal approximation for the remainders. They also point out that a gamma approximation is feasible and better in the sense that its cumulant generating function is closer to that of the remainder random variable compared with that of a normal approximation. Therefore, inspired by this, the approximation to $X_1$ with tail simulation for a given truncation level $K$ is 
\begin{align*}
    X_1 \approx X_1^K +\Gamma_1^K,
\end{align*}
where $\Gamma_1^K$ is a gamma random variable such that its first two moments match those of the remainder $R_1^K$ from \cref{lem:moments_of_remainder}.

We now detail our sampling strategy for $X_1^K$. The series which defines $X_1^K$ suggests two potential problems. First, the random variables $S_{n,k} \xlongequal[]{d} S$ are represented by an infinite weighted sum of independent exponential random variables, which requires an efficient simulation method. Second, given a Poisson sample $P_n = P$ for a fixed level $n=0, \ldots, K$, sampling the sum of $P$ independent random variables $S$ becomes increasingly computationally demanding when the sample $P$ tends to be larger. Thus, we now incorporate these two tasks with each other and consider simulating the sum of $P$ independent random variables $S$ directly, denoted by $S^P$, i.e.
\begin{align*}
S^P=\sum_{k=1}^P S_k,
\end{align*}
where $S_k$ are independent copies of $S$. Using the Laplace transform for $S$ given in Biane, Pitman and Yor \cite{biane2001probability}, $S^P$ has the following Laplace transform:
\begin{align}
\label{eq:Laplace_S^P}
    \Phi_{S^P} \left(b\right) =  \mathbb{E} \left[ \exp{\left( -b S^P \right)} \right] = \left( \frac{\sqrt{2b}}{\sinh{\sqrt{2b}}} \right)^P,
\end{align}
for $b > 0$.

We observe that any positive integer $P$ can be expressed in the form 
\begin{align*}
    P = p_1 + 10 p_{10} + 50 p_{50} + 5000 p_{5000} + 10^4 p_{10^4} + 10^5 p_{10^5} + 10^6 p_{10^6}.
\end{align*}
Here $p_{10^6}$ is the multiples of $10^6$ present in the integer $P$, i.e. $p_{10^6} = \lfloor P / 10^6 \rfloor$, $p_{10^5}$ is the multiples of $10^5$ present in the remainder of the division of $P$ by $10^6$, i.e. $p_{10^5} = \lfloor \left( P - 10^6 p_{10^6} \right)/10^5 \rfloor$, and so forth. 
As the law of $S^P$ is infinitely divisible for any $P > 0$ (see Section $3.2$ of Biane, Pitman and Yor \cite{biane2001probability}), the sum $S^P$ admits the representation
\begin{align*}
    S^P  \xlongequal[]{d}  \sum_{k \in \mathbb{S} } \sum_{i=1}^{p_k} S^k_i,
\end{align*}
where $\mathbb{S} = \left\{ 1, 10, 50, 5000, 10^4, 10^5, 10^6\right\}$ and for $i=1, \ldots, p_k$, $S_i^k$  are independent copies of $S^k$, i.e. the sum of $k$ independent random variables $S$, with $k \in \mathbb{S}$. Then, the above representation can be intended as a basis for an efficient sampling scheme for $S^P$ for all $ P > 0$ if we can realise $S^k$ effectively for $k \in \mathbb{S}$.
Indeed, we apply the direct inversion method to simulate $S^k$ with their inverse distribution functions approximated by predetermined Chebyshev polynomials for each $k \in \mathbb{S}$.
In general, the direct inversion algorithm for generating the samples of $S^P$ for any $P > 0$ is described as follows.
\begin{algorithm}
\caption{Direct inversion for $S^P$}
\label{alg:dir-inv-S^P}
\begin{algorithmic}[1]
\STATE{For each $k \in \mathbb{S}$, 
sample $p_k$ independent random variables $S^k_i $, $i=1,\ldots, p_k$  from the distribution of $S^k$ using the inverse distribution functions based on the corresponding Chebyshev polynomial approximations.    }
\STATE{Compute the accumulated sum, i.e. $\sum_{k} \sum_{i=1}^{p_k} S_i^k    \sim S^P$.}
\end{algorithmic}
\end{algorithm}

The advantage of this algorithm is that we only need to construct the Chebyshev polynomial approximations for the inverse distribution function of $S^k$ for $k \in \mathbb{S}$.
With this replacement, the complicated inverse distribution function becomes very easy to compute at arbitrary points. Moreover, since $S^k$ does not depend on any model parameters, the coefficients of the polynomials can be computed and tabulated in advance. As such, when a sample for $X_1$ is needed, we truncate the series representation to include the terms at $n \leq K$ with the tail approximated by a gamma distribution. For each $n= 0, \ldots, K$, we generate Poisson samples $P_n$ and simulate the sums $S^{P_n}$ directly by \cref{alg:dir-inv-S^P}, which requires evaluating some prescribed polynomials with coefficients drawn directly from the cached table; see the supplementary materials. To make the above process fast for implementation, we take advantage of the direct inversion to obtain Poisson samples when the mean is less than $10$. For larger means, the PTRD transformed rejection method suggested by H{\"o}rmann \cite{hormann1993transformed} will be used. 

To obtain the Chebyshev coefficients, it is crucial to determine the values of the inverse distribution functions at several points efficiently and accurately. For large $P$, we derive an asymptotic series expansion for the distribution function of $S^P$ when $P \to +\infty$ through the inverse Fourier transform of its characteristic function. While for small $P$, we utilise the explicit expression for the density function given by Biane and Yor \cite{biane1987valeurs}, which involves the parabolic cylinder functions. To derive the representation for the distribution function, we use a routine consisting of the power series and asymptotic expansions for the parabolic cylinder functions to evaluate the density function followed by term-wise integration. With these expansions computed, we apply root-finding algorithms to calculate the required values.    

\subsection{Asymptotic expansion for the distribution function of \boldmath$S^P$ for large \boldmath$P$}
\label{sec:aymp_largeP}
Before proceeding, it is worth noticing that in the limit $P \to \infty$, the expectation $\mathbb{E} \left[ S^P \right]   =   P/3$ and variance $ \mathrm{Var}\left[S^P\right]  =   2P/45$ of $S^P$ will diverge. Thus, we standardise the random variable $S^P$ by $ Z^P= \left( {S^P- P/3} \right) / {\sqrt{{2 P}/{45}  }}$, so that the new random variable $Z^P$ has mean zero and variance one. As $S^P$ is non-negative, the support of $Z^P$ is $\left[ - \sqrt{5P} / \sqrt{2}, + \infty \right)$. Then, taking the inverse Fourier transform of the characteristic function, the probability density function $f_{Z^P}$ of $Z^P$ has the form
\begin{align*}
    f_{Z^P} \left( x \right) & = \sqrt{\frac{2P}{45}} f_{S^P} \left( \frac{P}{3} + x \sqrt{\frac{2P}{45}} \right) \\
    & = \sqrt{\frac{2P}{45}} \frac{1}{2 \pi} \int_{-\infty}^{+\infty} \exp{\left( - \mathrm{i} z \left( \frac{P}{3} + x \sqrt{\frac{2P}{45}}  \right) \right)} \left( \frac{\sqrt{-2 z \mathrm{i}}}{\sinh{\sqrt{-2 z \mathrm{i}}}} \right)^P \, d z,
\end{align*}
where $f_{S^P}$ denotes the probability density function of $S^P$ and the first equality follows from the classical theorem on transforming density functions. By introducing $\beta=x \sqrt{2/45}  / \sqrt{P}$, the above equation can be written as 
\begin{align}
    f_{Z^P}\left(x\right)  = \frac{1}{4 \pi} \sqrt{\frac{2 P}{45}}  \int_{- \infty}^{+ \infty} \exp{ \left( P  \rho \left( z; \beta \right) \right) } \, dz, \label{eq:f__invFour_beta}
\end{align}
where $\rho \left( z; \beta \right)$ satisfies 
\begin{align}
    \rho \left( z; \beta \right) = \log{ \left( \frac{\sqrt{z \mathrm{i}}}{\sinh{\sqrt{z \mathrm{i}}}} \right)} + z \mathrm{i} \left( \frac{1}{6} + \frac{1}{2} \beta \right).
    \label{eq:rho}
\end{align}

We apply the standard technique of the steepest descent method to develop the asymptotic approximation for $f_{Z^P}$, where all the higher order terms are given in reciprocal powers of $P$, see Bender and Orszag \cite{bender1999advanced}, Bleistein and Handelsman \cite{bleistein1986asymptotic} and Ablowitz and Fokas \cite{ablowitz2003complex}. The expansion is then integrated term-wise to generate the asymptotic representation for the distribution function. The general procedure is given below. We first identify the critical points including saddle points $z_0$ of $\rho\left( z;\beta \right)$ such that $\rho^{\prime} \left( z_0;\beta \right)=0$. Note that since $\rho \left( z;\beta \right)$ depends on the parameter $\beta$, the saddle point $z_0$ will also depend on $\beta$. Due to the fact that $\beta$ is quite small as $P \to +\infty$, we can establish a useful expression for $z_0$ as a Taylor series in $\beta$. Afterwards, we demonstrate that the original contour of integration, i.e. the real line, can be deformed onto the steepest descent paths, obtained by considering the contour defined by $\mathrm{Im} \left( \rho \left( z;\beta \right) \right) = \mathrm{Im} \left( \rho \left( z_0;\beta \right) \right)$ and $\mathrm{Re} \left( \rho \left( z; \beta \right) \right) < \mathrm{Re} \left( \rho \left( z_0; \beta \right) \right)$, in the domain where the integrand is analytic. In this way, the rapid oscillations of the integrand can be removed when $P$ is large, whence the asymptotic behaviour of the integral can be determined locally depending only on a small neighbourhood of the critical points. 
We present the results in the next theorem, which is proved in \cref{sec:appendix_proof_3.2}.
\begin{theorem}
\label{thm:Asymp_Exp_pdf_Z^P}
As $P \to + \infty$ for fixed $x$ with $x > - \sqrt{5P/2} \left( 1+ 3 / \pi^2 -3 \coth{\pi} / \pi \right)$ and $\left\vert x \right\vert$ sufficiently small, we have 
\begin{align}
    f_{Z^P} \left( x \right) & \sim \frac{1}{4 \pi} \sqrt{\frac{2 }{45}} \exp{ \left( P \sum_{l=2}^{\infty} \hat{\rho}_l \beta^l \right)} \sum_{j=0}^{\infty}  \sum_{l=0}^{\infty} \sum_{n=0}^{\lfloor \frac{2}{3} j \rfloor} \hat{\alpha}_{n,l,j} \Gamma \left( j +\frac{1}{2} \right)  \beta^l   P^{n- j},
    \label{eq:asymp_exp_pdf_Z^P}
\end{align}
where $\hat{\rho}_l$ and $\hat{\alpha}_{n,l,j}$ are constants with explicit form derived in the proof and $\Gamma \left( c \right) $ is the gamma function.
\end{theorem}

\begin{remark}
The constant $\hat{\rho}_l$ is defined in the proof in equation \cref{eq:coeff_rho} using constants $\hat{r}_k$, $\hat{\xi}_k$ and $\hat{\upsilon}_{l,j}$ defined in equations \cref{eq:coeff_r}, \cref{eq:coeff_xi} and \cref{eq:coeff_upsilon}, respectively. The constant $\hat{\alpha}_{n,l,j}$ is defined in equation \cref{eq:coeff_alpha}, which depends on constants $\hat{\omega}_{n,j}$, $\hat{\mathcal{E}}_{l,k,n}$, $\hat{K}_k$, $\hat{\gamma}_{l,k}$, $\hat{\mathcal{C}}_{l, k_1, k_2, \cdots, k_n}$, $\hat{\varpi}_l$, $\hat{\phi}_{l,n}$, $\hat{\varphi}_{k,n}$, $\hat{\upsilon}_{l,j}$, $\hat{r}_k$ and $\hat{\xi}_k$ defined in equations \cref{eq:coeff_omega1}-\cref{eq:coeff_omega2}, \cref{eq:coeff_E1}-\cref{eq:coeff_E4}, \cref{eq:coeff_K}, \cref{eq:coeff_gamma}, \cref{eq:coeff_C}, \cref{eq:coeff_varpi}, \cref{eq:coeff_phi}, \cref{eq:coeff_varphi}, \cref{eq:coeff_upsilon}, \cref{eq:coeff_r} and \cref{eq:coeff_xi}, respectively.
\end{remark}

Having developed the large $P$ asymptotic approximation for the probability density function $f_{Z^P}$ with all the higher order terms given in reciprocal powers of $P$, the next stage is to derive an asymptotic representation for the corresponding distribution function. 
Before that, we first consider the asymptotic expansion for the probability $\mathbb{P} \left( z_1 < Z^P \leq z_2 \right)$ for some $z_1$, $z_2 > - \sqrt{5P/2} \left( 1+ 3 / \pi^2 -3 \coth{\pi} / \pi \right) $ with $\left\vert z_1 \right\vert$ and $\left\vert z_2 \right\vert$ sufficiently small, which can be accomplished by taking the integration of \cref{eq:asymp_exp_pdf_Z^P} on the finite interval $\left( z_1, z_2 \right]$. We then explain how this expression can be used to approximate the distribution function. The results are summarised in the next theorem with the proof given in \cref{sec:appendix_proof_3.3}.
\begin{theorem} \label{thm:Asymp_Exp_CDF_Z^P}
For $z_1$, $z_2 > - \sqrt{5P/2} \left( 1+ 3 / \pi^2 -3 \coth{\pi} / \pi \right)$ and $\left\vert z_1 \right\vert$, $\left\vert z_2 \right\vert$ sufficiently small, the following asymptotic series expansion holds as $P \to +\infty$. For $z_1 < z_2 < 0$, we have
\begingroup
\allowdisplaybreaks
\begin{align*}
   \int_{z_1}^{z_2} f_{Z^P} \left( x \right) \, dx  \sim &  \frac{1}{4 \pi} \sqrt{\frac{2}{45}} \sum_{j=0}^\infty P^{-\frac{j}{2}}  \sum_{r=0}^j  \sum_{n=0}^r \sum_{l=0}^{j-r} \hat{\eta}_{n,r} \hat{\lambda}_{l,j-r}  \left( -1 \right)^{r+l} \left( \sqrt{2} \right)^{2n+r+l-1} \\
        & \cdot \left(  \gamma \left( \frac{2n+r+l+1}{2}, \frac{\left(z_1\right)^2}{2}  \right) -  \gamma \left( \frac{2n+r+l+1}{2}, \frac{\left( z_2 \right)^2}{2} \right) \right). 
\end{align*}
\endgroup
For $z_1 < 0 \leq z_2$, we have
\begingroup
\allowdisplaybreaks
\begin{align*}
     \int_{z_1}^{z_2} f_{Z^P} \left( x \right) \, dx \sim &  \frac{1}{4 \pi} \sqrt{\frac{2}{45}} \sum_{j=0}^\infty P^{-\frac{j}{2}}  \sum_{r=0}^j  \sum_{n=0}^r \sum_{l=0}^{j-r} \hat{\eta}_{n,r} \hat{\lambda}_{l,j-r}  \left( \sqrt{2} \right)^{2n+r+l-1} \\
     & \cdot   \left( \left(-1 \right)^{r+l} \gamma \left( \frac{2n+r+l+1}{2}, \frac{\left(z_1 \right)^2}{2} \right) +  \gamma \left( \frac{2n+r+l+1}{2}, \frac{\left( z_2 \right)^2}{2} \right) \right).
\end{align*}
\endgroup
Here,  $\hat{\eta}_{n,r}$ and $\hat{\lambda}_{l,j-r}$ are constants explicitly given in \cref{eq:coeff_vartheta}-\cref{eq:coeff_eta} and \cref{eq:coeff_lambda1}-\cref{eq:coeff_lambda2} in the proof and $\gamma \left( s, z \right)$ is the lower incomplete gamma function.
\end{theorem}

\begin{remark}
We have thus established a large $P$ asymptotic series expansion for the probability that the random variable $Z^P$ takes values in $\left( z_1, z_2 \right]$. Notice that this representation is valid when $z_i > - \sqrt{5P/2} \left( 1+ 3 / \pi^2 -3 \coth{\pi} / \pi \right)$ with $\left\vert z_i \right\vert$ sufficiently small for $i=1,2$. This restriction can be traced back to \cref{thm:Asymp_Exp_pdf_Z^P}, where the saddle point is given in a Taylor series in $\beta$ for small $\left\vert \beta \right\vert$. Hence for practical applications, we truncate the Taylor series to generate an accurate approximation for the saddle point when $\left\vert \beta \right\vert$ is sufficiently small. More precisely, there is a region centred around zero with width $\tilde{\beta}$, throughout which the error of the approximation is below a given threshold. The range of validity can be determined by numerical comparisons using for example Maple in practice. This range of validity turns out to be large enough so that the error in computing the distribution function for large $P$ is negligible. More precisely, we take the error below $10^{-12}$ for all the cases considered here and the corresponding values for $\tilde{\beta}$ are $0.0596$, $0.0421$, $0.0133$ and $0.0042$ for $P = 5000, 10^4, 10^5$ and $10^6$, respectively.  
\end{remark}

In summary, we have so far developed a tractable method to evaluate the distribution function $F_{Z^P} \left( z \right)$. This is approximated by integration of the corresponding density function on some restricted interval $\left( -\tilde{z}, z \right]$ with $\tilde{z}$ carefully chosen for each $P$. We derive an asymptotic expansion for the integral in reciprocal powers of $P$ for all orders following the steepest descent method. In practice, we compute enough terms for the expansion to achieve the desirable accuracy in Maple with $50$ digit accuracy for $P=5000, 10^4, 10^5$ and $10^6$, along with the root-finding for $F_{Z^P}^{-1}$ values at nodal points required by Chebyshev polynomial approximations.

\subsection{Series expansion for the distribution function of \boldmath$S^P$ for small \boldmath$P$}
\label{sec:exp_samllP}
In this section, we turn to the specifics of the series expansion for the distribution function of $S^P$ for small $P$. Recall from \cref{eq:Laplace_S^P} that $S^P$ has the Laplace transform $\Phi_{S^P} \left(b\right)  =  \left( {\sqrt{2 b}}/{\sinh  \sqrt{2 b}} \right)^P $ for $b > 0$. Biane and Yor \cite[formula $(3\text{x})$]{biane1987valeurs} have given an explicit expression for the probability density function with such a Laplace transform. Namely, for arbitrary $P>0$, the probability density function $f_{S^P}$ is of the form
\begin{align*}
    f_{S^P} \left( y \right)  =\frac{1}{\sqrt{2 \pi}} \frac{2^P}{\Gamma \left(P\right)} y^{-\frac{1}{2}\left( P+2 \right)} \sum_{n=0}^\infty \frac{\Gamma \left( n+P \right)}{\Gamma \left( n+1 \right)} \exp{\left( -\frac{ \left( 2n+P \right)^2}{4y} \right)} D_{P+1} \left( \frac{2n+P}{\sqrt{y}} \right),
\end{align*}
where $D_{P+1} \left( z \right)$ is the parabolic cylinder function with order $P+1$. We use different strategies to calculate these functions according to different ranges of $z$. For small $z$, the power series is preferable while for large $z$ an asymptotic expansion will be applied. We summarise these properties here. First, the series expansion for the parabolic cylinder function can be written as
 \begin{align}
 \label{eq:Cylinder_power}
        D_{P+1} \left(z \right) = \sum_{k=0}^\infty \hat{d}_k \left( P \right) z^k,
    \end{align}
with the coefficients $\hat{d}_k \left( P\right)$ satisfying some recurrence relations, that can be found in Gil, Segura and Temme \cite{gil2006computing} or Abramowitz and Stegun \cite{abramowitz1964handbook}. We use Maple for their practical implementation. Second, in the limit  $z \to +\infty$, $D_{P+1} (z)$ has the following asymptotic behaviour (Gil, Segura and Temme \cite[formula $(23), (24), (25)$]{gil2006computing}):
\begin{align}
   \label{eq:Cylinder_asymp}
        D_{P+1} \left( z \right) \sim \exp{ \left( -\frac{1}{4} z^2 \right) } z^{P+1} \sum_{k=0}^\infty \left( -1 \right)^k \frac{\left(-\left(P+1\right) \right)_{2k}}{k! \left(2z^2\right)^k},
\end{align}
where $ \left( a \right)_{k}$ denotes the Pochhammer symbol $ \left( a \right)_k = \Gamma\left( a+k \right)/\Gamma \left( a \right)$. For comparisons of different computational methods, see Temme \cite{temme2000numerical} and Gil, Segura and Temme \cite{gil2006computing}.  Finally, integrating the density function $f_{S^P}$ term-wise yields the series representation for the distribution function $F_{S^P}$ of $S^P$ stated below. See \cref{sec:appendix_proof_3.5} for a detailed proof. 
\begin{theorem} \label{thm:Series_Exp_CDF_S^P}
For any $0 \leq x < \infty$ and $P \in \left( 0, 1 \right) \cup \mathbb{N}$, the distribution function $F_{S^P} \left( x \right)$ can be written as the following convergent series
\begin{align}
    F_{S^P} \left( x \right) =  & \frac{1}{\sqrt{2 \pi}} \frac{2^{P+1}}{\Gamma \left( P \right)} \sum_{n=0}^\infty \frac{\Gamma \left( n+P \right)}{\Gamma \left( n+1 \right)} \left( 2n+P \right)^{-P} G \left( \frac{2 n+P}{\sqrt{x}} \right), \label{eq:series_Exp_CDF_S^P_smally}
\end{align}
where the function $G \left( y \right)$ for $y > 0$ is given by 
\begin{align*}
    G \left( y \right) = \int_y^{+\infty} z^{P-1} \exp{\left( - \frac{1}{4} z^2 \right)} D_{P+1} \left( z \right) \, dz. 
\end{align*}
Further, $G$ satisfies 
\begin{align*}
     G \left( y \right) = G_1 \left( y, y^* \right) + G_2 \left( y^* \right)
\end{align*}
for all $y^* \geq y$ that are sufficiently large, where $G_1$ can be expressed as the convergent series
\begin{align*}
    G_1 \left( y , y^* \right)  =  \sum_{k=0}^\infty \hat{d}_k \left( P \right) 2^{P+k-1}  \left(  \Gamma \left( \frac{P+k}{2}, \frac{y^2}{4}  \right) - \Gamma \left( \frac{P+k}{2},  \frac{\left( y^* \right)^2}{4}  \right)  \right),  
\end{align*}
and $G_2$ has the asymptotic expansion
\begin{align*}
G_2 \left( y^* \right) \sim \sum_{k=0}^\infty \left( -1 \right)^k
        \frac{ \left(- \left( P+1 \right) \right)_{2k} }{  k!}  
           2^{P-2k-\frac{1}{2}}  \Gamma \left( P-k+\frac{1}{2}, \frac{\left( y^* \right)^2}{2} \right), \quad \text{as} \quad y^* \to +\infty. 
\end{align*}
Here, $\Gamma\left( s,z \right)$ is the upper incomplete gamma function.
\end{theorem}

The previous theorem provides an effective approach to calculate the distribution function $F_{S^P}$ for small $P$ across its support with high precision. In practice, we choose to use the asymptotic expansion for the parabolic cylinder function $D_{P+1} \left( z \right)$ whenever $z \geq \Delta \left( P + 3/2 \right)$ for some positive constant $\Delta \gg 1$, suggested by Gil, Segura and Temme \cite[Section $5$]{gil2006computing}. Accordingly, we set $y^* = \text{max} \left\{  \Delta \left( P + {3}/{2} \right), y \right\}$ when computing the function $G \left( y \right)$ for fixed $y > 0$. This means only asymptotic series is involved in the computation of $G \left( y \right) = G_2 \left( y \right)$ for sufficiently large $y$ such that $y \geq \Delta \left( P + 3/2 \right)$. The constant $\Delta$, which may vary depending on the value of $P$, can be determined by numerical trials of comparing the accuracy and efficiency of evaluating both the power series and asymptotic representations at particular points. As in the case for large $P$, we compute the above series representation for $F_{S^P}$ and perform the root-finding for $F_{S^P}^{-1}$ in Maple for $P =1,10$ and $50$. Notice that the series expansion developed here is valid for any $P \in \left(0 , 1 \right)$, not only for integer $P$. This will be useful in the simulation of $X_2$ later on.

\subsection{Chebyshev polynomial approximation for the inverse distribution function of \boldmath$S^P$} \label{sec:Cheby_S^P}
As presented above, for any positive integer $P$, the simulation of $S^P$ is based on generating a series of random variables $S^k$ for $k \in \mathbb{S}$
by direct inversion. This method takes a uniform sample $u \sim \text{Unif} \left( 0, 1 \right)$
and returns the quantile function evaluated at $u$ as a sample for the associated distribution, which requires computing the inverse of the distribution function. However, it is often the case that the inverting process is computationally inefficient due to many factors such as poor initial guess and the lack of an analytical expression for the corresponding quantile function. Since a large number of samples is needed for the Monte Carlo simulation when the same number of inversions of the distribution will be performed, we now look for a more tractable technique to complete this step.

Indeed, we employ the method of Chebyshev polynomials explained below to approximate the inverse distribution function $F_{S^P}^{-1}$ for $P \in \mathbb{S}$.
Despite the fact that the polynomial is just an approximation, we can still obtain highly accurate results by restricting the error, which is controlled by the degrees of the polynomials we construct. In practice, we require the uniform error to be far smaller than the Monte Carlo error, e.g. of order $10^{-12}$.

Recall that a degree $n$ Chebyshev polynomial approximation has the form 
\begin{align*}
    c_0 T_0 \left( z \right) +c_1 T_1 \left( z \right) +\cdots+c_n T_n \left( z \right) -\frac{1}{2} c_0,
\end{align*}
where $T_k \left( z \right) = \cos{ \left( k \arccos{z} \right)}$ for $k=1,\ldots,n$ are the Chebyshev polynomials of degree $k$ defined on $ \left[ -1,1 \right]$ and $c_k$ for $k=1,\ldots,n$ are the Chebyshev coefficients computed in the standard way following Press \textit{et al.} \cite{press1992numerical}. Since polynomials often exhibit more rapid changes than the distribution functions, approximations by polynomials might not be able to fully capture the behaviour of the inverse function $F_{S^P}^{-1} \left( u \right)$. Hence, identifying appropriate scaling schemes of the argument $u$ is of great importance to allow the application of the Chebyshev polynomial approximation.  
The choices of the scales are mainly characterised by the behaviour of the function depending on the range of $P$. We briefly state the scaling and its rationale behind for large $P$ and small $P$ separately.
\paragraph{Large $P$}
Instead of the sum $S^P$, we take the normalised random variable $Z^P$ with zero mean and unit variance into consideration. For the approximation of the inverse distribution function $F_{Z^P}^{-1}$, we focus on the sub-interval $ \left[ F_{Z^P}\left( 0 \right),1 \right)$ of its support $\left[ 0, 1 \right]$ first, corresponding to the region where the random variable $Z^P$ takes positive values. In the limit of large $P$, the distribution function of $Z^P$ resembles that of a standard normal distribution. Thus, we generalise and apply the ideas underlying the Beasley-Springer-Moro direct inversion method for standard normal random variables; see Moro \cite{moro1995full}, Joy, Boyle and Tan \cite{joy1996quasi} and Malham and Wiese \cite{malham2014efficient}. The normal distribution function has three regions exhibiting different characteristic behaviours on the positive real line. Accordingly, we roughly split the interval $\left[ F_{Z^P} \left( 0 \right),1 \right)$ into three regimes: the central $\left[ F_{Z^P} \left( 0 \right),u_1 \right]$, the middle $ \left( u_1,u_2 \right]$ and the tail $ \left( u_2,1-10^{-12} \right]$ regimes. In general, the central regime roughly represents the area where the decreasing density function has a increasing slope while the middle regime represents the area where the decreasing density function has a declining slope with the tail regime representing the region where  the density function is flat taking values close to zero. We neglect the regimes from $1-10^{-12}$ to $1$. 

\begin{remark}
It should be pointed out here that the above rule is just for reference only. In
reality, we can choose optimal values for the boundaries $u_1$ and $u_2$ by a small number of trials in Maple to ensure that the resulting Chebyshev polynomial approximations have moderate degrees while retaining the accuracy for all three regimes. We may come across the circumstance that the approximations which achieves the desired accuracy have degrees of say $15$ for both the central and middle regimes but a higher degree of say $50$ for the tail regime for some given $u_1$ and $u_2$. Such a case should be avoided from the perspective of efficiency as higher degree often comes with higher computational cost. Hence, it is necessary to set the values $u_1$ and $u_2$ again through further investigations so that the degrees of the approximation for all regions are balanced with each other. If both of the degrees of the Chebyshev polynomials constructed for two neighbouring regions are at relative lower level, we may combine those two regimes to one and produce a unified approximation.
\end{remark}

In the central regime, we follow Malham and Wiese \cite{malham2014efficient} to scale and shift the variable. Define $U \left( u \right) \coloneqq \sqrt{2 \pi} \left( u-F_{Z^P} \left( 0 \right) \right)$ and $z \left( u \right) \coloneqq k_1 U \left( u \right)+k_2$, where the parameters $k_1$ and $k_2$ are chosen to make sure $z\left( F_{Z^P} \left( 0 \right) \right) = -1$ and $z \left( u_1 \right) =1$. Then, we approximate the inverse distribution function by 
\begin{align*}
    F_{Z^P}^{-1} \left( u \right) \approx U \left( u \right) \cdot \left( c_0 T_0 \left(z \left( u \right) \right) +c_1 T_1\left(z \left( u \right) \right) +\cdots+c_n T_n\left(z \left( u \right) \right) -\frac{1}{2} c_0 \right).
\end{align*}
  
In the middle and tail regimes, we approximate 
\begin{align*}
        F_{Z^P}^{-1} \left( u \right) \approx  c_0 T_0 \left( z \left( u \right) \right) +c_1 T_1\left(z \left( u \right) \right) +\cdots+c_n T_n\left(z \left( u \right)\right) -\frac{1}{2} c_0,
\end{align*}
where $U\left( u \right) \coloneqq\log{ \left( -\log{ \left( 1-u \right)} \right)}$ and $z \left( u \right) \coloneqq k_1 U\left( u \right)+k_2$ with the parameters $k_1$ and $k_2$ chosen such that $z \left( u \right) =-1$ at the left endpoint and $z \left( u \right)=1$ at the right endpoint. The ansatz for $U$ follows from inverting the asymptotic tail approximation for the standard normal, which is equivalent in distribution to $Z^P$ when $P \to +\infty$ by the central limit theorem; see Moro \cite{moro1995full}.
 
The above serves as a general discussion for choosing the scaled variables and approximations in the region of $\left[ F_{Z^P} \left( 0 \right),1-10^{-12} \right]$ for large $P$. We apply this procedure to the cases $P=10, 50, 5000, 10^4, 10^5$ and $10^6$, the inverse distribution functions of which are roughly anti-symmetric.
For the remaining half sub-interval $\left[10^{-12},  F_{Z^P} \left( 0 \right) \right)$ of its support, we can apply similar results to the scaling and approximation following the arguments mentioned above. 
\paragraph{Small $P$}
In the Chebyshev polynomial approximation for small $P$, the idea remains the same as above. Notice that the random variable $S^P$ takes positive values only. Since the distribution has a heavy right tail, we break down the support of $F_{S^P}^{-1}$ into four regimes: the left $ \left[ 10^{-12},u_1 \right]$, the central $ \left( u_1,u_2 \right]$, the middle $\left( u_2,u_3 \right]$ and the right tail $ \left( u_3,1-10^{-12} \right]$ regimes. We neglect the regimes at a distance of $10^{-12}$ from its endpoints. In theory, these boundary points are determined in accordance with the behaviour of the distribution function, but again it is better 
to set them via empirical studies in practice.

The central limit theorem tells us the asymptotic distribution of the sum $S^P$ when P is large. However, for small $P$ we have to analyse the limiting behaviour of the distribution function $F_{S^P}$ and its inverse $F_{S^P}^{-1}$ in order to help us find the proper scaled variables when we construct Chebyshev polynomial approximations. We build on the series representation for $F_{S^P}$ given in \cref{thm:Series_Exp_CDF_S^P} to derive the results below; see \cref{sec:appendix_proof_3.7} for a detailed proof.  
\begin{corollary}
\label{thm:leading_CDF_S^P}
For any $P \in \left( 0, 1 \right) \cup \mathbb{N}$, the distribution function $F_{S^P}$ has the following asymptotic expansion 
\begin{align}
    F_{S^P} \left( x \right) \sim \frac{1}{\sqrt{\pi}} 2^{P+\frac{1}{2}} P^{P-1} x^{-P+\frac{1}{2}} \exp{\left( -\frac{P^2}{2x} \right)}, \quad \text{as} \quad x \to 0^+. \label{eq:cdf_Y2_leading_left}
\end{align}
\end{corollary}

The above expression describes the behaviour of the distribution function $F_{S^P} \left( x \right)$ as $x \to 0^+$. Now, our goal is to invert this relation to obtain the asymptotic approximation for the inverse distribution function $F_{S^P}^{-1} \left( u \right)$ as $u \to 0^+$. Let $u = F_{S^P}(y)$, then from \cref{eq:cdf_Y2_leading_left} it is clear 
\begin{align}
    \frac{u \sqrt{\pi}}{2^{P+\frac{1}{2}} P^{P-1}} \sim y^{-P+\frac{1}{2}} \exp{\left( -\frac{P^2}{2y} \right)}, \quad \text{as} \quad y \to 0^+. \label{eq:cdf_Y2_leading_left_u}
\end{align}
Introducing the new variable $v\coloneqq {u \sqrt{\pi}}/ \left( {2^{P+\frac{1}{2}} P^{P-1}} \right)$ and taking logarithm on both sides, we can rewrite \cref{eq:cdf_Y2_leading_left_u} as
\begin{align*}
    \log{v} \sim \left( P-\frac{1}{2} \right) \log{\frac{1}{y}} -\frac{P^2}{2y},  \quad \text{as} \quad y \to 0^+. 
\end{align*}
After rearrangement, the above expression becomes
\begin{align*}
  \frac{1}{y} \sim  \frac{2}{P^2} \left( P-\frac{1}{2} \right) \log{\frac{1}{y}} -\frac{2}{P^2} \log{v}, \quad \text{as} \quad y \to 0^+. 
\end{align*}
Since $\log{y^{-1}} / y^{-1} \to 0$ as $y \to 0^+$, we have $-2 \log{v} / P^2 \sim y^{-1}$ in the limit $y \to 0^+$. This relation further gives  
\begin{align*}
    \log{ \left(  - \frac{2}{P^2} \log{v} \right)} = \mathrm{o} \left( \frac{1}{y} \right),  \quad \text{as} \quad y \to 0^+. 
\end{align*}
Hence, taking advantage of order relations we can write $y$ in terms of $v$ as 
\begin{align} 
    \frac{1}{y} & \sim -\frac{2}{P^2} \log{v} +\frac{2}{P^2} \left( P-\frac{1}{2} \right) \log{
    \left(  -\frac{2}{P^2} \log{v} \right)}     \label{eq:asymp_inverse_0} \\
    & =  -\frac{2}{P^2} \log{v} + \frac{2}{P^2} \left( P-\frac{1}{2} \right) \log{\left( \frac{2}{P^2} \right)} + \frac{2}{P^2} \left( P-\frac{1}{2} \right) \log{\left( - \log{v} \right)} 
\nonumber 
\end{align}
when $y \to 0^+$. In particular by the fact that $\log{\left( - \log{v} \right)} = \mathrm{o} \left( y^{-1} \right)$ as $y \to 0^+$, its leading order behaviour yields
\begin{align*}
    y & \sim \left( \frac{2}{P^2} \left( P-\frac{1}{2} \right) \log{\frac{2}{P^2}} -\frac{2}{P^2} \log{v}    \right)^{-1} \\
    & = \left(   \frac{2}{P^2} \left( P-\frac{1}{2} \right) \log{\frac{2}{P^2}} -\frac{2}{P^2} \log \left( \frac{u \sqrt{\pi}}{2^{P+\frac{1}{2}} P^{P-1}} \right) \right)^{-1},
\end{align*}
as $y \to0^+$, i.e. $u \to 0^+$, where the last equation comes from the transformation $v = {u \sqrt{\pi}}/ \left( {2^{P+\frac{1}{2}} P^{P-1}} \right)$.

As $u \to 1$, i.e. $y \to + \infty$, we adopt a gamma approximation for the tail. This is implied by empirical tests which show that the distribution is positively skewed with a longer right tail. Hence, by matching the mean and variance of $S^P$ with those of a gamma random variable, the shape and rate parameters take the form $s=5P/2$ and $r=15/2$. Then, the distribution function $F_{S^P}$ is approximated by the distribution function $F_X$ of a gamma random variable $X$ with parameters $s$ and $r$ given as follows:
\begin{align*}
   F_X \left( y \right) = 1- \frac{1}{\Gamma \left( \frac{5}{2} P \right)} \Gamma \left( \frac{5}{2} P,\frac{15}{2} y \right),
\end{align*}
Making use of the asymptotic relationship given below for the incomplete gamma function (Abramowitz and Stegun \cite[formula $(6.5.32)$]{abramowitz1964handbook})
\begin{align*}
 \Gamma \left( s,z \right) \sim z^{s-1} \exp{ \left( -z \right)} \sum_{k=0}^\infty \frac{\Gamma \left( s \right)}{\Gamma \left( s-k \right)} z^{-k}, \quad \text{as} \quad z \to +\infty  
\end{align*}
establishes as $y \to +\infty$,
\begin{align}
F_X \left( y \right) \sim  1-\frac{1}{\Gamma \left( \frac{5}{2} P \right)} \left( \frac{15}{2} y \right)^{\frac{5}{2} P-1} \exp{ \left( -\frac{15}{2} y \right)}. \label{eq:cdf_Y2_leading_right}
\end{align}
Set $u = F_{X} \left( y \right)$. 
After rewriting \cref{eq:cdf_Y2_leading_right}, we obtain
\begin{align*}
\left( 1 - u \right) \Gamma \left( \frac{5}{2} P \right) \sim  \left( \frac{15}{2} y \right)^{\frac{5}{2} P-1} \exp{ \left( -\frac{15}{2} y \right)}, \quad \text{as} \quad y \to + \infty.
\end{align*}
To generate an asymptotic expression for $y$, we start by taking logarithm and defining the new variable $v\coloneqq\left( 1-u \right) \Gamma \left( 5 P/2 \right)$, which gives
\begin{align}
    \log{v} \sim \left( \frac{5}{2} P-1 \right) \log{\left( \frac{15}{2} y \right)} -\frac{15}{2} y,  \quad \text{as} \quad y \to + \infty. \label{eq:cdf_Y2_leading_right_log}
\end{align}
Rearranging \cref{eq:cdf_Y2_leading_right_log} leads to
\begin{align*}
    y \sim -\frac{2}{15} \log{v} +\frac{2}{15} \left( \frac{5}{2} P-1 \right) \log{\left( \frac{15}{2} y \right)},  \quad \text{as} \quad y \to + \infty.
\end{align*}
By a short calculation analogous to \cref{eq:asymp_inverse_0}, 
we conclude
\begin{align*}
 y \sim  -\frac{2}{15} \log{v} +\frac{2}{15} \left( \frac{5}{2} P-1 \right) \log{\left( -\log{v} \right)},  \quad \text{as} \quad y \to + \infty.             
\end{align*}
On substituting $v=\left( 1-u \right) \Gamma \left( 5 P/2 \right)$, as $y \to + \infty$, i.e. $u \to 1$, its leading order is of the form
\begin{align*}
    y \sim -\frac{2}{15} \log{\left(  \left( 1-u \right) \Gamma \left( \frac{5}{2} P \right) \right)}.
\end{align*}

The analysis above outlines the asymptotic approximation for $F_{S^P}^{-1} \left( u \right)$ as $u \to 0$ and $u \to 1$, and provides the ansatz behind the choices of reasonable scaling variables for Chebyshev polynomial approximations for small $P$, i.e. $P=1$. Accordingly, we report the routines for approximations of the inverse distribution function $F_{S^P}^{-1} \left( u \right)$ through Chebyshev polynomials for the four regimes identified in detail. Note again the parameters $k_1$ and $k_2$ given below restrict the ranges of the transformed variable $z$ to the interval $\left[-1, 1 \right]$.  

For all the four regimes, we approximate $F_{S^P}^{-1} \left( u \right)$ by a degree $n$ Chebyshev polynomial approximation of a scaled and shifted variable $z \left( u \right) \coloneqq k_1 U\left( u \right)+k_2$ as below:
    \begin{align*}
        F_{S^P}^{-1} \left( u \right) \approx c_0 T_0 \left( z \left( u \right) \right) +c_1 T_1\left( z \left( u \right) \right) +\cdots+c_n T_n \left( z \left( u \right) \right) -\frac{1}{2} c_0,
    \end{align*}
where $U \left( u \right) \coloneqq\left( \left( 2/P^2 \right) \left( P-1/2 \right) \log{ \left( 2/P^2 \right)} -\left( 2/P^2 \right) \log{ \left( u \sqrt{\pi}/ \left( 2^{P+1/2} P^{P-1} \right) \right)} \right)^{-1}  $ in the left regime $\left[ 10^{-12},u_1 \right]$, $U \left( u \right) \coloneqq \left( u_2-u \right) \sqrt{2P/45}$ in the central regime $ \left( u_1,u_2 \right]$ and $U \left( u \right) \coloneqq \left( - 2/15 \right) \log{ \left( \left( 1-u \right) \Gamma\left(5P/2 \right) \right)} $ in the middle regime $ \left( u_2,u_3 \right]$ and right tail regime $\left( u_3,1-10^{-12} \right]$. 
\begin{figure}[tbhp!]
    \centering
        \subfloat{\label{fig:P10^6}\includegraphics[width=0.95\textwidth]{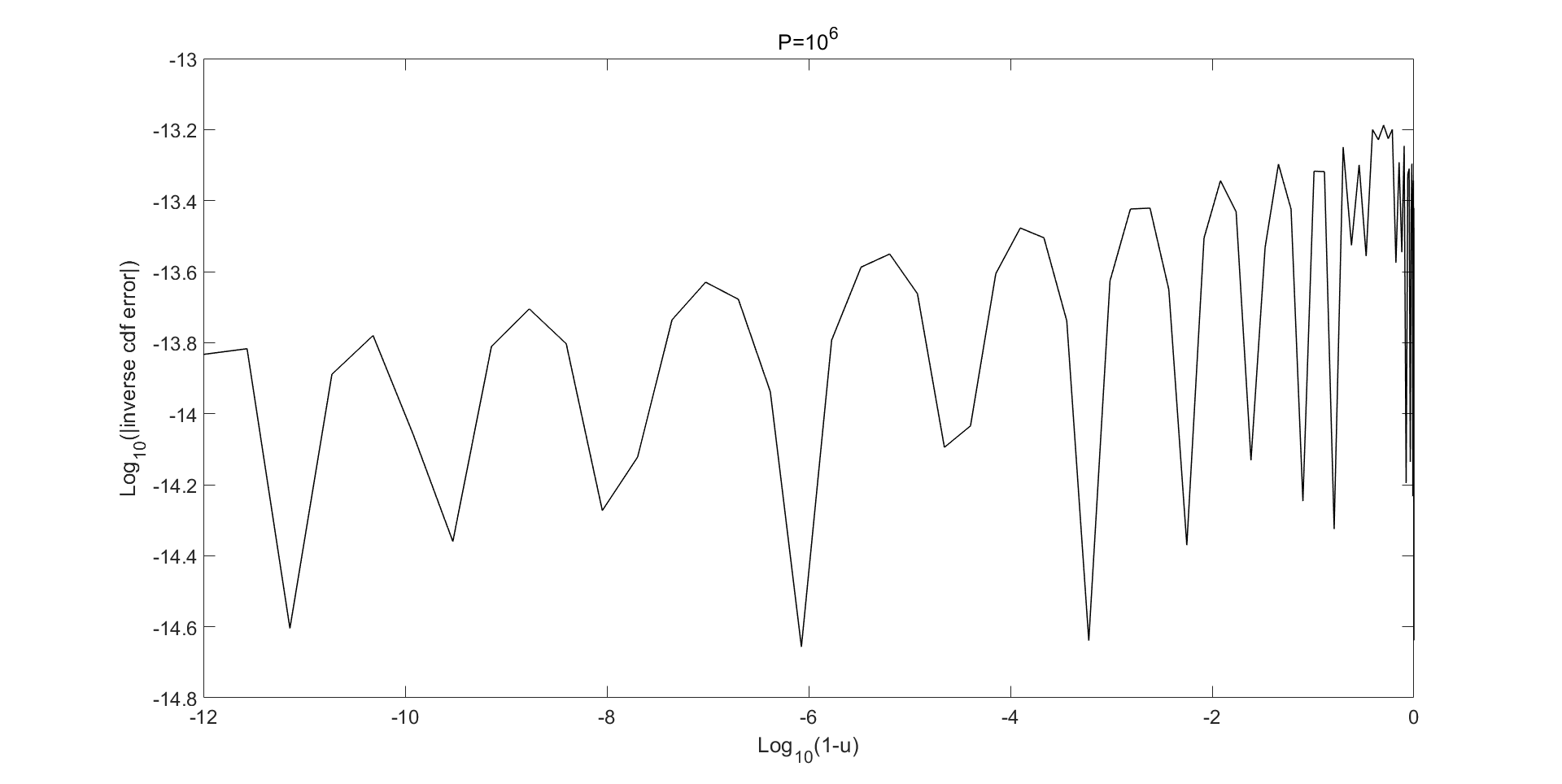}} \\
        \subfloat{\label{fig:P1}\includegraphics[width=0.95\textwidth]{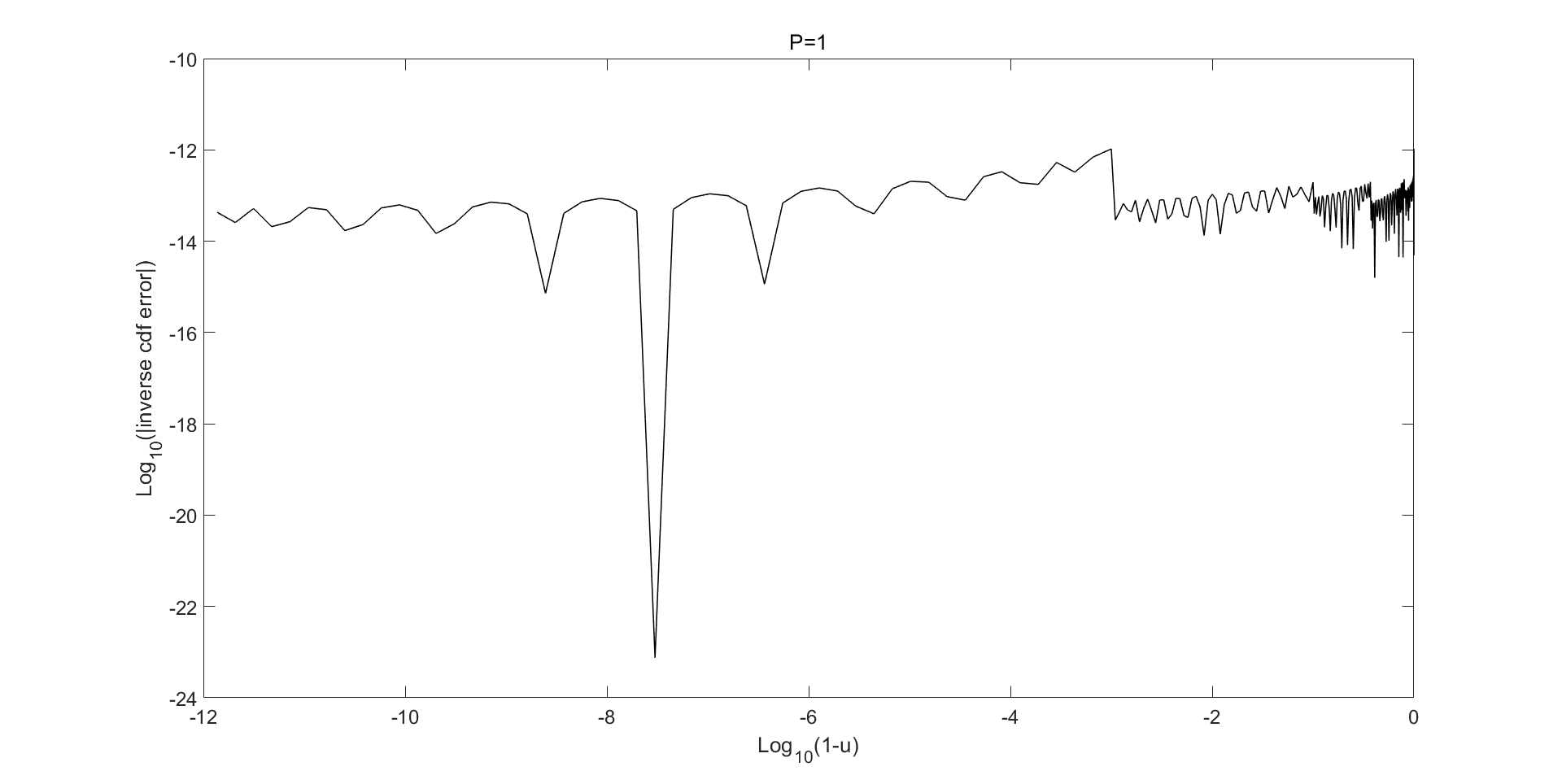}} 
    \caption{ We plot the errors in the Chebyshev polynomial approximations to the inverse distribution functions $F_{Z^P}^{-1} \left( u \right)$ with $P=10^6$ $( \text{top panel} )$ and $F_{S^P}^{-1} \left( u \right)$ with $P=1$ $( \text{bottom panel} )$ across all regimes, respectively. Note that to highlight the tail we use a log-log$_{10}$ scale with $1-u$ on the abscissa.}
    \label{fig:S^P}
\end{figure}

We have summarised the approximation techniques for the inverse distribution function $F_{Z^P}^{-1}(u)$ or $F_{S^P}^{-1}(u)$ taking into account the various values that $P$ and $u$ might take. Following this, we compute the coefficients for the Chebyshev polynomial approximations in the standard fashion (see Press \textit{et al.} \cite{press1992numerical}) for the cases $P=1, 10, 50, 5000, 10^4, 10^5$, $10^6$ using Maple. With all these accurate and reliable coefficients, quoted in the supplementary materials, then imported to Matlab, subsequent Chebyshev approximations are evaluated by Clenshaw's recurrence formula, which can be found in Press \textit{et al.} \cite{press1992numerical}. Therefore, for any $P > 0$, $S^P$ can be sampled repeatedly with high efficiency using \cref{alg:dir-inv-S^P}.

We end this section by showing the respective errors in the Chebyshev polynomial approximations to $F_{Z^P}^{-1} \left( u \right)$ and $F_{S^P}^{-1} \left( u \right)$ with $u \in \left[ 10^{-12},1-10^{-12} \right]$ when $P$ is $10^6$ and $1$ in \cref{fig:S^P}. For each $u$, the error of the approximation is $c - \hat{c}$, where $c$ is obtained by a high precision root-finding procedure applied to the expansions for the distribution functions and $\hat{c}$ is evaluated by the prescribed Chebyshev polynomials. To highlight the tail we plot the errors on a log-log$_{10}$ scale with $1-u$ on the abscissa. For $P=10^6$, we split the interval $\left[ 10^{-12}, 1-10^{-12}\right]$ into two regimes: $\left[ 10^{-12}, 0.5001 \right)$ and $ \left[ 0.5001, 1-10^{-12}\right] $, where both of the Chebyshev polynomials have degrees $16$. For $P=1$, we generate approximations for five regimes as described above where the right tail region is further split into two, the degrees for which are $25$ in the left with $u \in \left[ 10^{-12}, 0.2 \right)$, 
$18$ in the central with $u \in \left[ 0.2,0.63 \right)$, 
$15$ in the middle with $ u \in \left[ 0.63, 0.9 \right)$, 
$18$ and $13$ in the right tail regimes with $\left[ 0.9, 0.999 \right)$ and $\left[ 0.999, 1-10^{-12}\right)$, respectively. Notice the error in all cases is of order $10^{-12}$. Results for the other values, reported in Shen \cite{shen2019numerical}, have similar accuracy as those in \cref{fig:S^P}.

\subsection{Simulation of \boldmath$X_2$ and \boldmath$Z$} 
\label{sec:simulation of X_2}
Let us first introduce the notation $h=\delta/2$, where $\delta = 4 \kappa \theta / \sigma^2$ is the degrees of freedom. Note that for Heston models calibrated to real market data, the zero boundary of the variance process is typically attainable and strongly reflecting; see Haastrecht and Pelsser \cite{van2010efficient} and Lord, Koekkoek and Van Dijk \cite{lord2006comparison}. By the Feller condition, this requires $\delta < 2$, i.e. $h < 1$. Hence, after separating the time parameter, $X_2$ can be written in the form
\begin{align*}
    X_2=\sum_{n=1}^{\infty} \frac{\tau^2}{4^n} C_n^{h}= \tau^2 \sum_{n=1}^{\infty} \frac{1}{4^n} C_n^{h} = \tau^2 Y_2^{h},
\end{align*}
with $Y_2^h: = \sum_{n=1}^{\infty} C_n^{h}/4^n$ depending only on the parameter $h$. The structure of the random variable $Y_2^h$, along with its dependence on the model parameters, provides us with another possibility to sample $Y_2^h$ and thus $X_2$, apart from the truncation method. In fact, the Laplace transform of $Y_2^h$ is identical to
that for $S^h$. We can therefore try to extend the direct inversion of $S^h$ for any $h \in \mathbb{N}$ developed above for the simulation of $Y_2^h$. 

Recall that $Y_2^h$ has the Laplace transform 
\begin{align*}
  \Phi_{Y_2^h} (b)= \mathbb{E} \left[\exp{\left( -b Y_2^h \right)} \right] = \left( \frac{\sqrt{2 b}}{\sinh  \sqrt{2 b}} \right)^h,  \quad \text{for} \quad b > 0,
\end{align*}
which has the same expression as that of $S^h$ given by \cref{eq:Laplace_S^P} after replacing $P$ by $h$. The only difference is that the parameter $h$ now is restricted to $\left( 0,1 \right)$ rather than positive integers. This suggests the decomposition proposed in \textcolor{siaminlinkcolor}{section} \ref{sec:simulation of X_1} for integer $h$ and the resulting $S^h$ is no longer reasonable. However, motivated by Malham and Wiese \cite{malham2013chi}, we have the following alternative formula for $0< h < 1$, which is given to the first three decimal places
\begin{align*}
    h = \frac{h_5}{5} + \frac{h_{10}}{10} + \frac{h_{20}}{20} + \frac{h_{50}}{50} + \frac{h_{100}}{100} + \frac{h_{200}}{200} + \frac{h_{500}}{500} + \frac{h_{1000}}{1000} + \frac{h_{2000}}{2000},
\end{align*}
where $h_k \in \left\lbrace 0,1,2 \right\rbrace$ for $k \in \mathbb{H} = \left\{ 5, 10, 20, 50, 100, 200, 500, 1000, 2000   \right\}$.
The above decomposition works for the case when $h$ is rounded to three decimal digits, but it can be generalised to $h \in \left( 0, 1 \right)$ given to any decimal places in principle. Next, we give the direct inversion algorithm for generating $Y_2^h$ for any $h \in \left( 0,1 \right)$ given to the first three decimal places.
\begin{algorithm}
\caption{Direct inversion for $Y_2^h$}
\label{alg:dir-inv-Y_2^h}
\begin{algorithmic}[1]
\STATE{For each $k \in \mathbb{H}$,
sample $h_k$ independent random variables $Y_{2,i}^{1/k}$, $i=1, \ldots, h_k$ from the distribution of $Y_2^{1/k}$ by inverse transform sampling based on the corresponding Chebyshev polynomial approximations.}
\STATE{Compute the accumulated sum, i.e. $\sum_k \sum_{i=1}^{h_k}  Y_{2,i}^{1/k} \sim Y_2^h $. }
\end{algorithmic}
\end{algorithm}

\begin{figure}[tbhp!]
    \centering
    \subfloat{\includegraphics[width=0.95\textwidth]{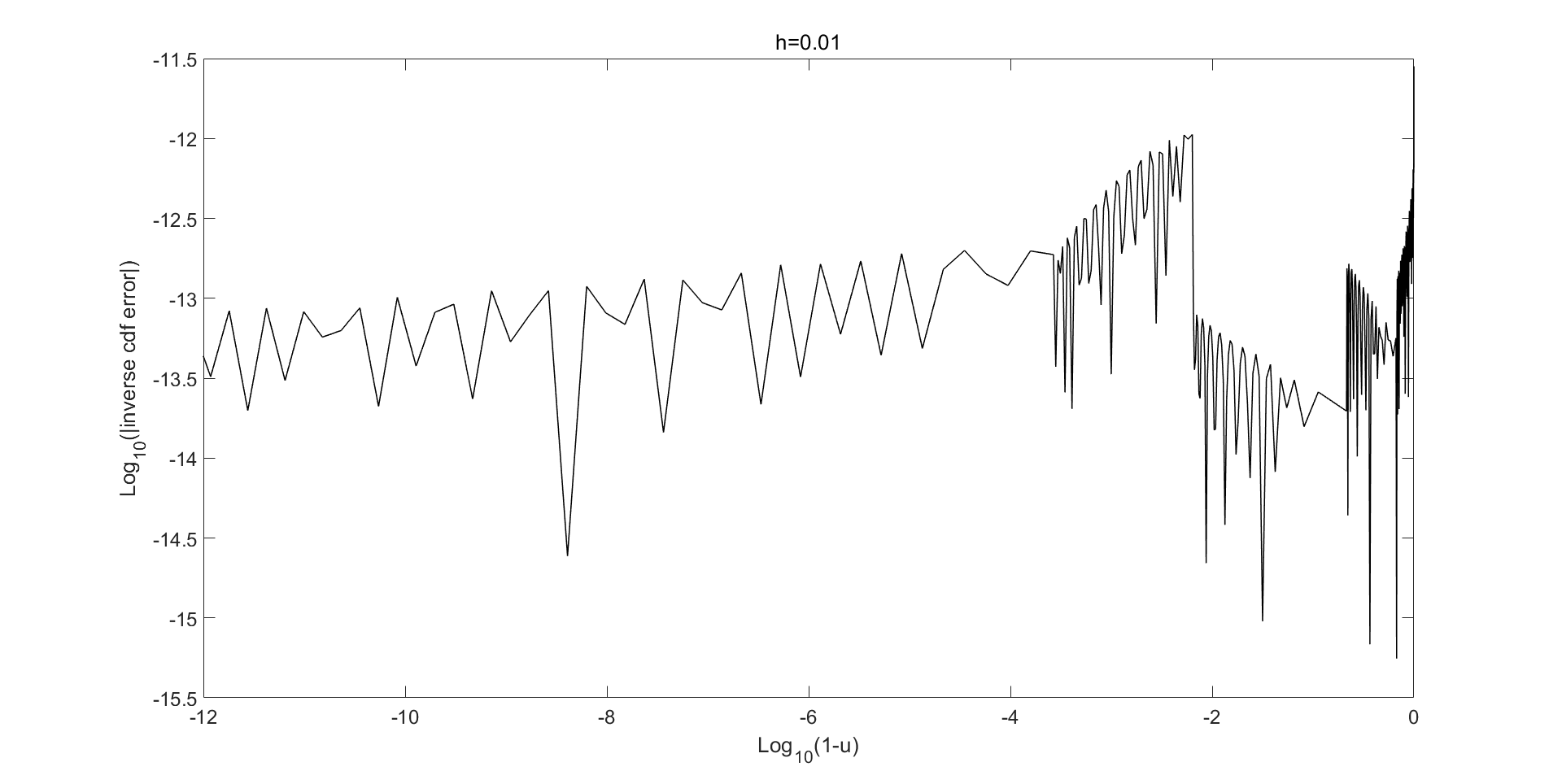}} 
    \caption{ We plot the errors in the Chebyshev polynomial approximations to the inverse distribution functions $F_{Y_2^h}^{-1} \left( u \right)$ with $h=0.01$ across all regimes. Note as above we use a log-log$_{10}$ scale with $1-u$ on the abscissa.}
    \label{fig:h100}
\end{figure}

\begin{figure}[tbhp!]
    \centering
     \includegraphics[width=0.95\textwidth]{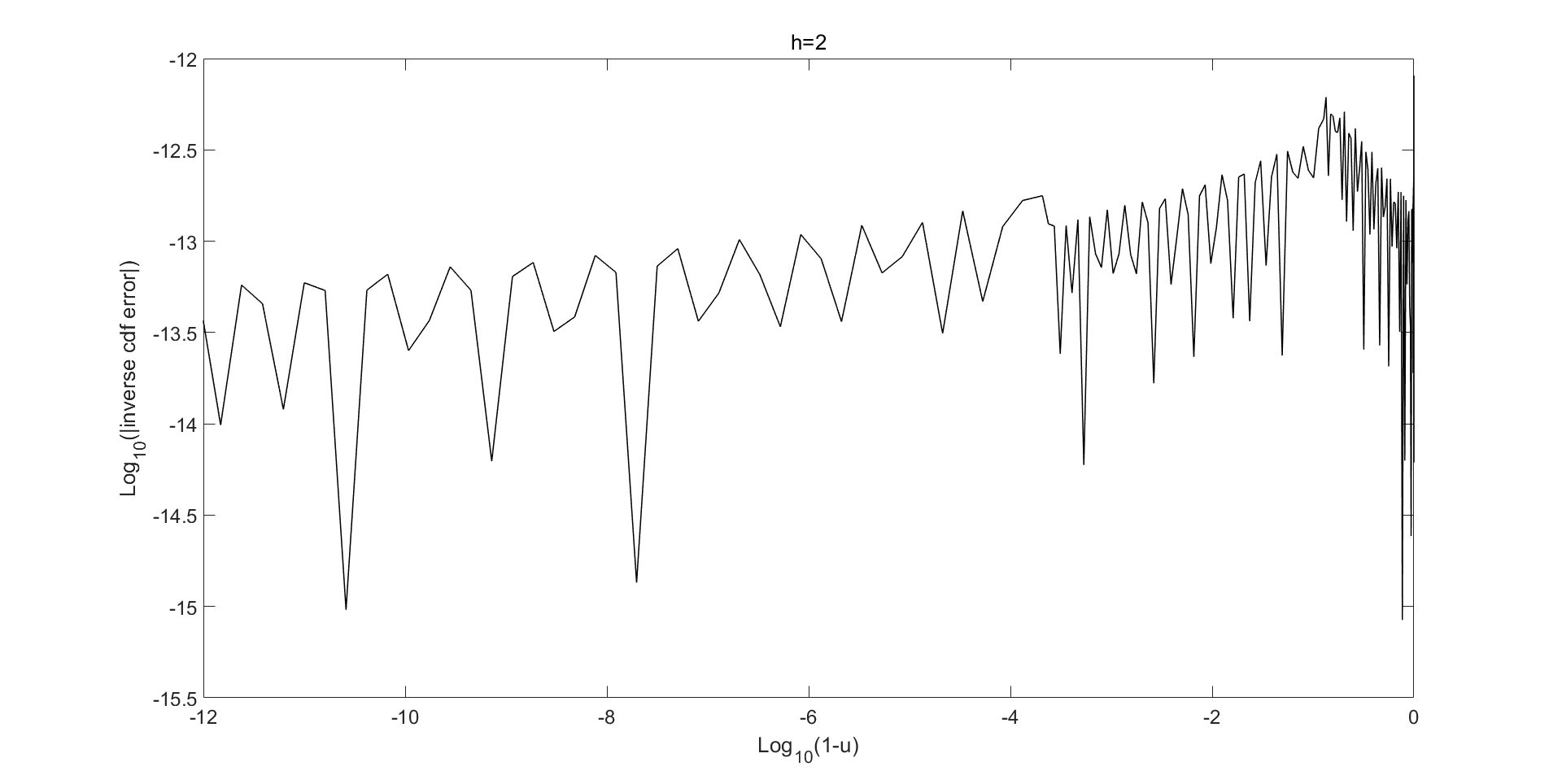} 
    \caption{We plot the errors in the Chebyshev polynomial approximations to the inverse distribution functions $F_{Y_2^h}^{-1} \left( u \right)$ with $h=2$, i.e. $F_{Z'}^{-1} \left( u \right)$ across all regimes. Note as above we use a log-log$_{10}$ scale with $1-u$ on the abscissa.}
    \label{fig:}
\end{figure}

Given \cref{alg:dir-inv-Y_2^h}, for a general $h$, the simulation of $Y_2^h $  is reduced to simulating several particular random variables such as $Y_2^{1/5}$, $Y_2^{1/10}$, $\ldots$ using their inverse distribution functions, which are approximated by the associated Chebyshev polynomials. To design these Chebyshev polynomial approximations, the approach for $S^P$ with small integer $P$ introduced in \textcolor{siaminlinkcolor}{section} \ref{sec:exp_samllP} and \textcolor{siaminlinkcolor}{section} \ref{sec:Cheby_S^P} can be fully used here. This is because the series expansion and the asymptotic approximation for the distribution function of $S^P$ remain valid for small non-integer $P$. Therefore, we apply the same strategy to calculate the Chebyshev polynomial approximations for the inverse distribution function of $Y_2^h$ with fixed
$h=1/k$, $k \in \mathbb{H}$,
the values for the coefficients of which are presented in the supplementary materials. \cref{fig:h100} shows the errors in the approximation for $F_{Y_2^h}^{-1}$ by Chebyshev polynomials across all regimes when $h=1/100$. Notice that under such circumstance, because of the heavy tail we further split the right tail region into two smaller regions where different Chebyshev polynomials are developed, making a total of five separate regions: $\left[ 10^{-12} , 0.3364 \right)$, $\left[ 0.3364, 0.7854 \right)$, $\left[ 0.7854, 0.9936 \right)$, $\left[ 0.9936, 0.9997 \right)$ and $\left[ 0.9997, 1-10^{-12} \right]$ with degrees $24$, $19$, $18$, $19$ and $31$, respectively. All errors are fluctuating at the level of $10^{-12}$. See Shen \cite{shen2019numerical} for more results of the other cases.

As $Z$ is a special case of $X_2$ when $h=2$, the approach to generating samples of $X_2$ discussed above is fully applicable here. In fact, we directly construct the Chebyshev polynomial approximations for the inverse distribution function $F_{Z'}^{-1}$ with $Z' =  Z/\tau^2$ since $Z'$ is independent of any parameters. We plot the resulting errors in \cref{fig:}. The polynomials have degrees between $22$ and $27$ in the four regions with errors of order $10^{-12}$.
\begin{table}[tbhp]
{\footnotesize
  \caption{Parameters for the Heston model.}\label{tab:parameter}}
\begin{center}
\begin{tabular}{ccccccc} 
\hline 
\addlinespace[0.5ex]
\multirow{2}{*}{\bf{Parameters}} & \multicolumn{4}{c}{\bf{European}} & \multicolumn{2}{c}{\bf{Path-dependent}} \\
\cmidrule(lr){2-5}
\cmidrule(lr){6-7}
&  \bf{ Case $\mathbf{1}$ }&  \bf{Case $\mathbf{2}$  } & \bf{ Case $\mathbf{3}$ }  &  \bf{Case $\mathbf{4}$ }   & \bf{Asian} & \bf{Barrier} \\
\addlinespace[0.5ex]
\hline 
\addlinespace[0.5ex]
$\kappa$ & $0.5$   & $0.3$    & $1$  & $6.21$ & $1.0407$ & $0.5$    \\   
$\theta$ & $0.04$   & $0.04$    & $0.09$   & $0.019$ & $0.0586$ & $0.04$  \\  
$\sigma$ & $1$   & $0.9$    & $1$   & $0.61$ & $0.5196$ & $1$  \\ 
$\rho$ & $-0.9$   & $-0.5$    & $-0.3$   & $-0.7$   & $-0.6747$ & $0$ \\ 
$t$ & $10$   & $15$    & $5$   & $1$ & $4$ & $1$  \\ 
$v_0$ & $0.04$   & $0.04$    & $0.09$   & $0.010201$ & $0.0194$ & $0.04$  \\ 
$S_0$ & $100$   & $100$    & $100$   & $100$  & $100$ & $100$\\ 
$r$ & $0$   & $0$    & $0.05$   & $0.0319$ &  $0$   & $0$ \\ 
\addlinespace[0.5ex]
\hline
\end{tabular}    
\end{center}
\end{table}
\section{Numerical analysis}
\label{sec:numerical_analysis}
In this section, we compare our direct inversion method with the gamma expansion of Glasserman and Kim \cite{glasserman2011gamma} by pricing four challenging European call options in the Heston model. The sets of parameters considered are given in \cref{tab:parameter}. These four sets are taken from Glasserman and Kim \cite{glasserman2011gamma}, which are found to be in the typical range of parameter values of the Heston model in practice. Two path-dependent options including an Asian option with yearly fixings (see Smith \cite{smith2007almost}, Haastrecht and Pelsser \cite{van2010efficient} and Malham and Wiese \cite{malham2013chi}) and a digital double no touch barrier option (see Lord, Koekkoek and Van Dijk \cite{lord2006comparison} and Malham and Wiese \cite{malham2013chi}) are also tested with parameter values shown in \cref{tab:parameter}.

\subsection{Time integrated conditional variance}
Before giving simulation results for option prices, we first illustrate the performance of the above direct inversion method in terms of accuracy based upon the series expansion given in \cref{thm:Series_Rep}. Recall that our objective is to sample from the distribution of the random variable $\int_0^t V_s \, ds$ given its endpoints $V_0$ and $V_t$, denoted by $\bar{I}$, i.e.
\begin{align*}
    \bar{I}  = \left( \int_0^t V_s \, ds \middle\vert V_0=v_0, V_t=v_t \right) = \frac{4}{\sigma^2}  \left( \int_0^\tau \tilde{A}_s \, ds \middle\vert \tilde{A}_0 = a_0, \tilde{A}_\tau = a_\tau \right) = \frac{4}{\sigma^2} I,
\end{align*}
under the probability measure $\mathbb{Q}$. We have decomposed the integral into the sum of three independent series after measure transformation. Among the realisation of those three series, the first one is truncated with tail approximated by a gamma random variable and the remaining two series are simulated exactly by direct inversion. 
\begin{figure}[tbhp!]
    \centering
    \subfloat[Case $1$: $v_0=v_t=0.04$]{\label{fig:case1_mom_errors_50000000_0.04}\includegraphics[scale=0.28]{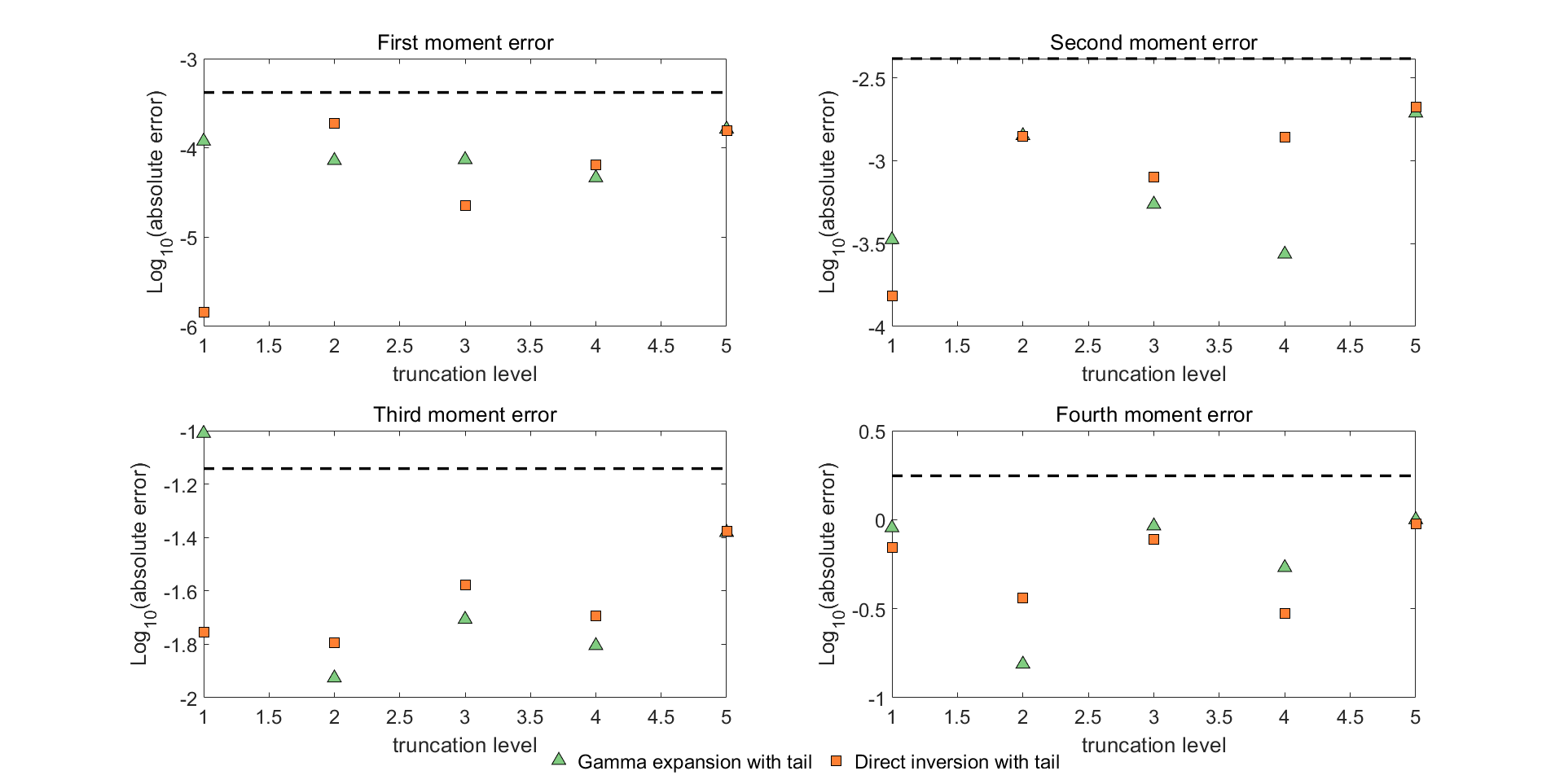}}  
    \caption{We indicate the absolute errors in the first four moments of the conditional integral $\bar{I}$ simulated by direct inversion and gamma expansion versus the truncation levels for Case $1$ with different values for $v_t$. Both methods are implemented with tail simulation. We perform $5 \cdot 10^7$ simulations for each case. Below the dashed line, the errors are not statistically significant at the level of three standard deviations.}
    \label{fig:case1_mom_errors_50000000_1}
\end{figure}

In \cref{fig:case1_mom_errors_50000000_1} and \cref{fig:case4_mom_errors_50000000_1} the absolute errors in the first four moments are displayed for simulating the conditional integral $\bar{I}$ with different values of $v_t$ using our method. For comparison, we include the results by employing the gamma expansion from Glasserman and Kim \cite{glasserman2011gamma} as well. For both methods, we apply tail approximations with truncation level increasing in integers. The number of samples generated in each case is $5 \cdot 10^7$. The three panels shown in \cref{fig:case1_mom_errors_50000000_1} from top to bottom correspond to the three representative values $v_t=0.04, 4, 0.000004$ for Case $1$ and the panels in \cref{fig:case4_mom_errors_50000000_1} correspond to the three fixed values $v_t=0.010201, 0.0025, 0.05$ (top to bottom) for Case $4$. Note that the true moments can be computed by evaluating the respective derivatives of the moment generating functions derived by Broadie and Kaya \cite{broadie2006exact} at the origin.
\begin{figure}[tbhp!]
\ContinuedFloat 
    \centering
    \subfloat[Case $1$: $v_0=0.04, v_t=4$]{\label{fig:case1_mom_errors_50000000_4}\includegraphics[scale=0.28]{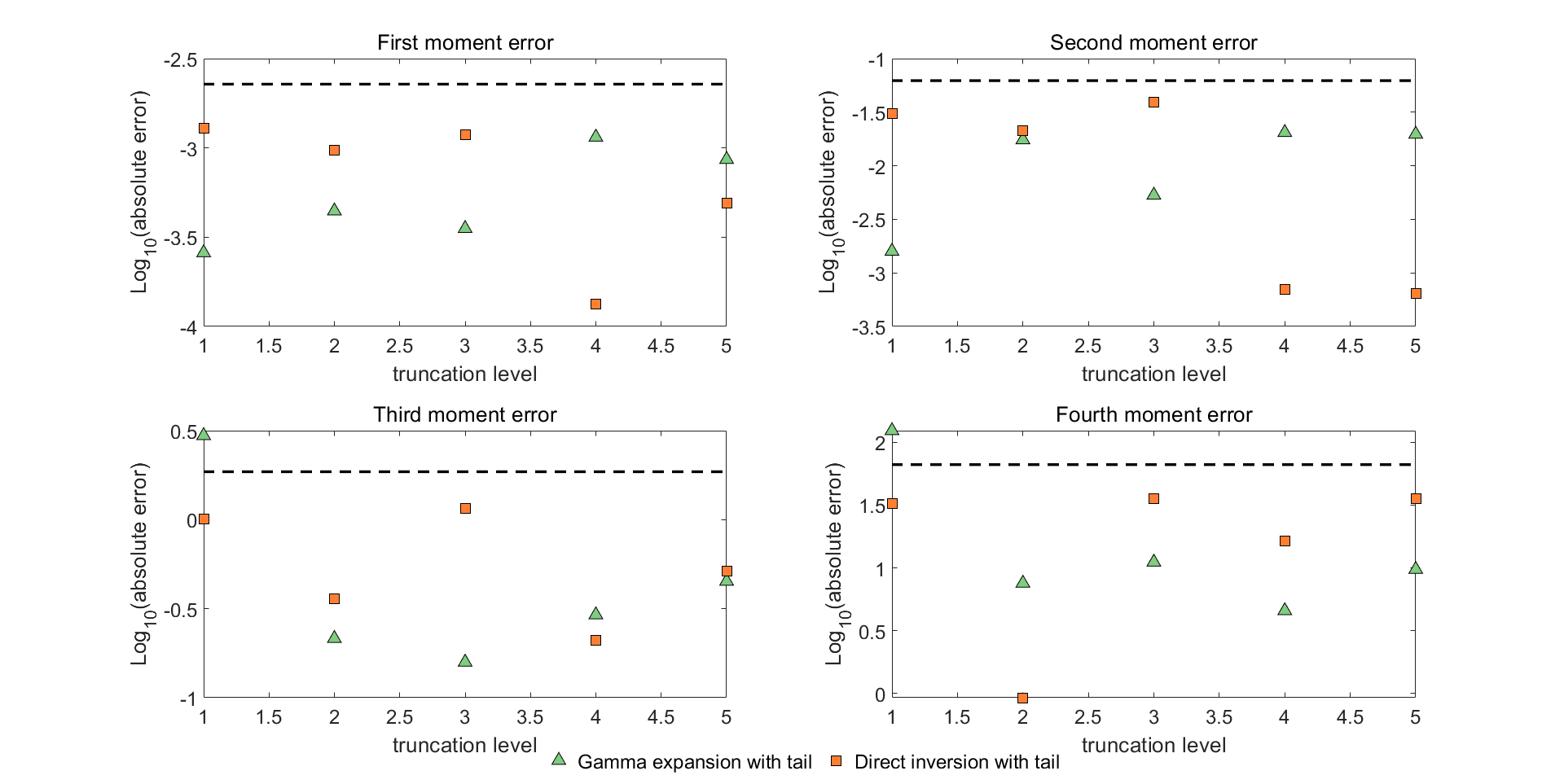}} \\
     \subfloat[Case $1$: $v_0=0.04, v_t=0.000004$]{\label{fig:case1_mom_errors_50000000_0.000004}\includegraphics[scale=0.28]{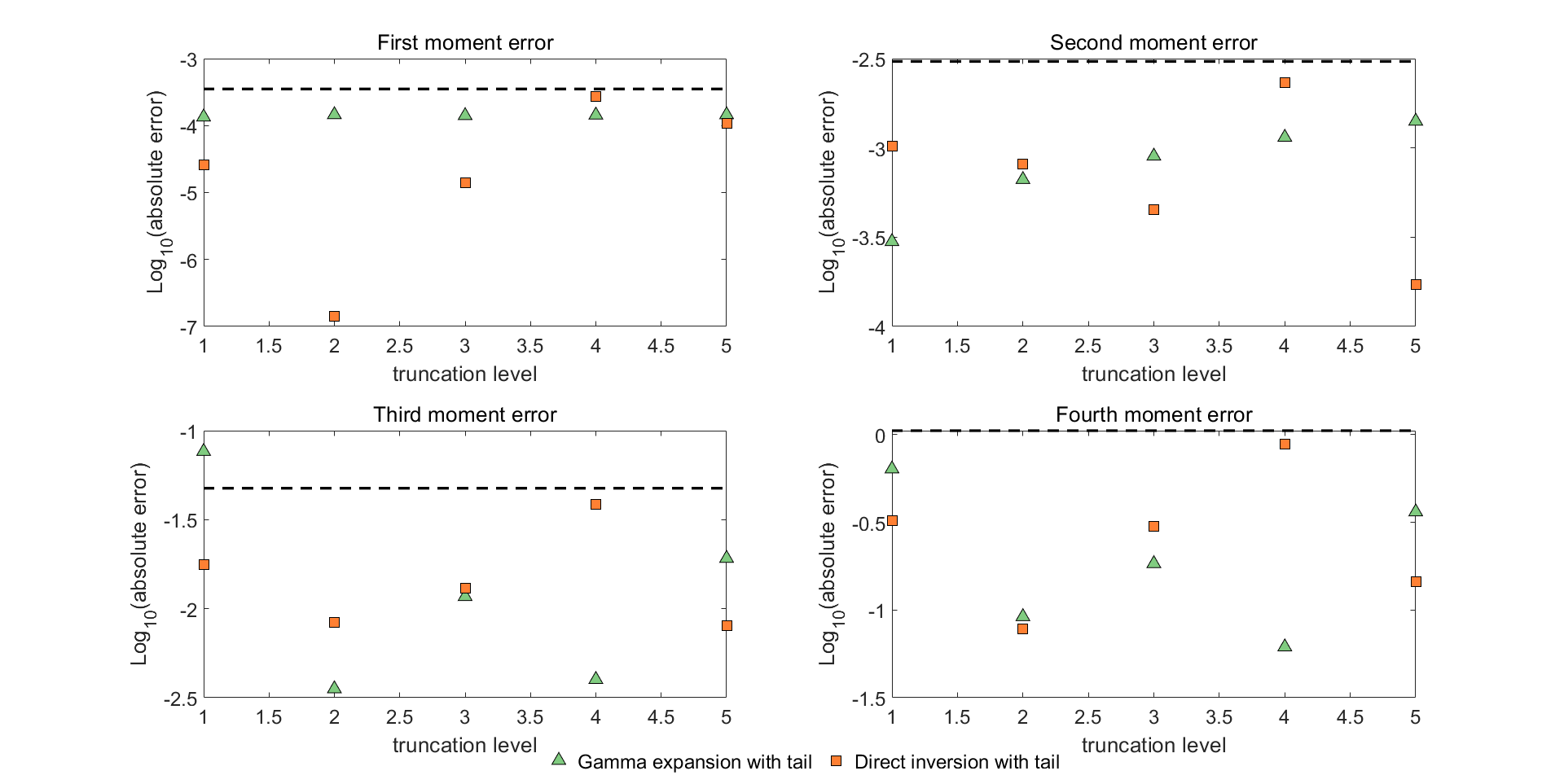}}  
    \caption{(cont.) We indicate the absolute errors in the first four moments of the conditional integral $\bar{I}$ simulated by direct inversion and gamma expansion versus the truncation levels for Case $1$ with different values for $v_t$. Both methods are implemented with tail simulation. We perform $5 \cdot 10^7$ simulations for each case. Below the dashed line, the errors are not statistically significant at the level of three standard deviations.}
    \label{fig:case1_mom_errors_50000000_2}
\end{figure}

We observe that most errors for the first two moments across different values of $v_t$ and truncation levels for both Case $1$ and $4$ are not significantly different from zero at the level of three standard deviations, suggesting both methods achieve high accuracy for these two moments as expected. This is consistent with the theory as tail simulation in each method is designed such that the first two moments are matched. In other words, the simulations should lead to the exact first and second moments in principle, whence only Monte Carlo noise with a scaling as the inverse of the square root of the sample size, i.e. $\left( 5 \cdot 10^7 \right)^{-{1}/{2}}$, is observed. 

For the higher moments, the errors of the direct inversion are fluctuating at some level below the statistical significance for all circumstances considered. These errors are so small that a decreasing trend is not visible when increasing the truncation level. In contrast, with the increment of the truncation levels, the errors of the gamma expansion first exhibit a decaying pattern until the curves become horizontal. For example, the behaviour of the decreasing errors of the third and fourth moments is obvious when the truncation level is increased from one to two. The falling tendency appears to be more significant when we increase the sample size, thus, reduce the Monte Carlo effect, which will be discussed later. This suggests that there exists some bias in the gamma expansion with small truncation levels while the direct inversion with lower truncation levels has the same accuracy as that with higher ones.
\begin{figure}[tbhp!]
    \centering
    \subfloat[Case $4$: $v_0=v_t=0.010201$]{\label{fig:case4_mom_errors_50000000_0.010201}\includegraphics[scale=0.28]{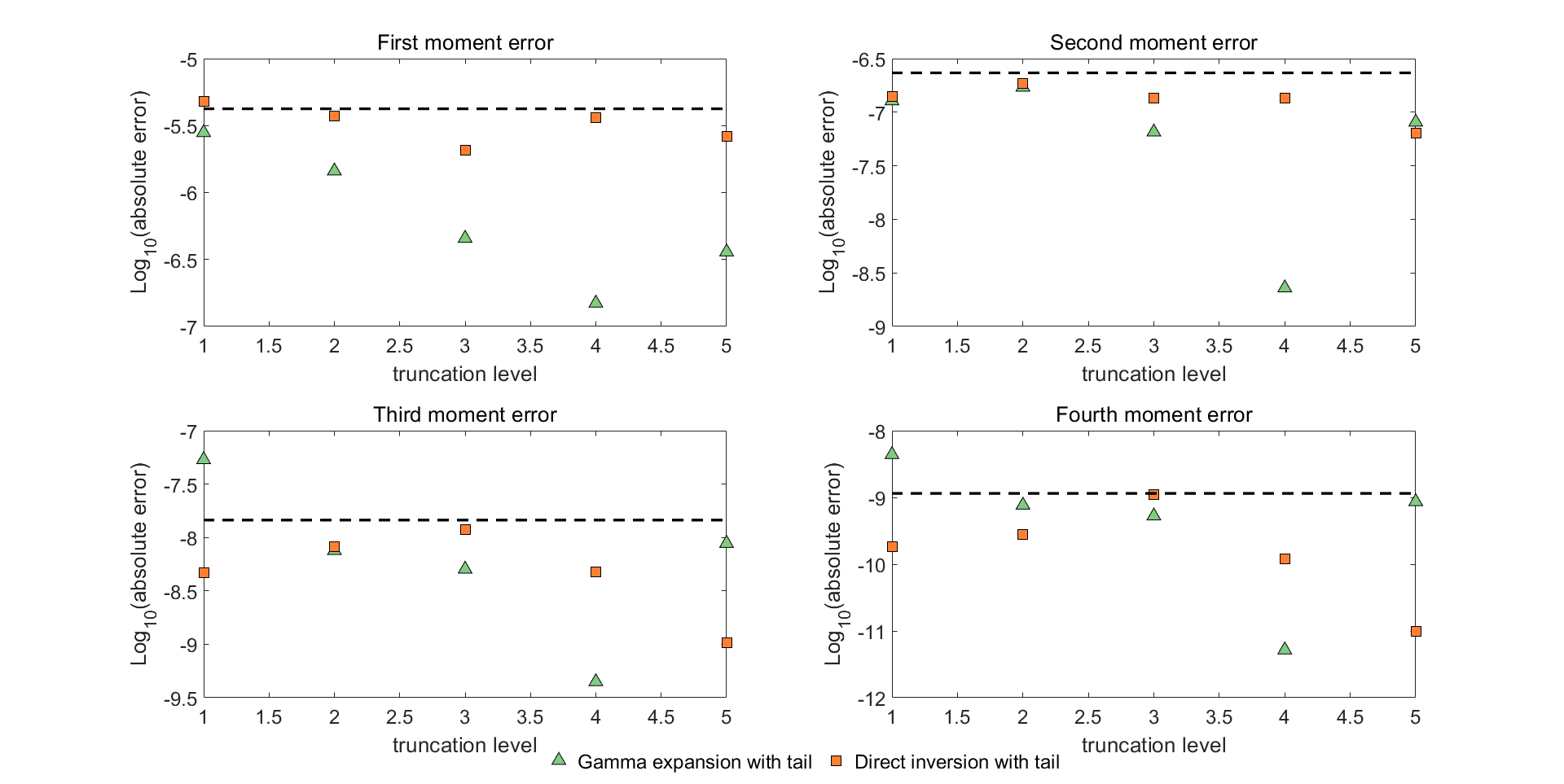}} 
    \caption{We indicate the absolute errors in the first four moments of the conditional integral $\bar{I}$ simulated by direct inversion and gamma expansion versus the truncation levels for Case $4$ with different values for $v_t$. Both methods are implemented with tail simulation. We perform $5 \cdot 10^7$ simulations for each case. Below the dashed line, the errors are not statistically significant at the level of three standard deviations.}
    \label{fig:case4_mom_errors_50000000_1}
\end{figure}
\begin{figure}[tbhp!]
\ContinuedFloat
    \centering
    \subfloat[Case $4$: $v_0=0.010201, v_t=0.0025$]{\label{fig:case4_mom_errors_50000000_0.0025}\includegraphics[scale=0.28]{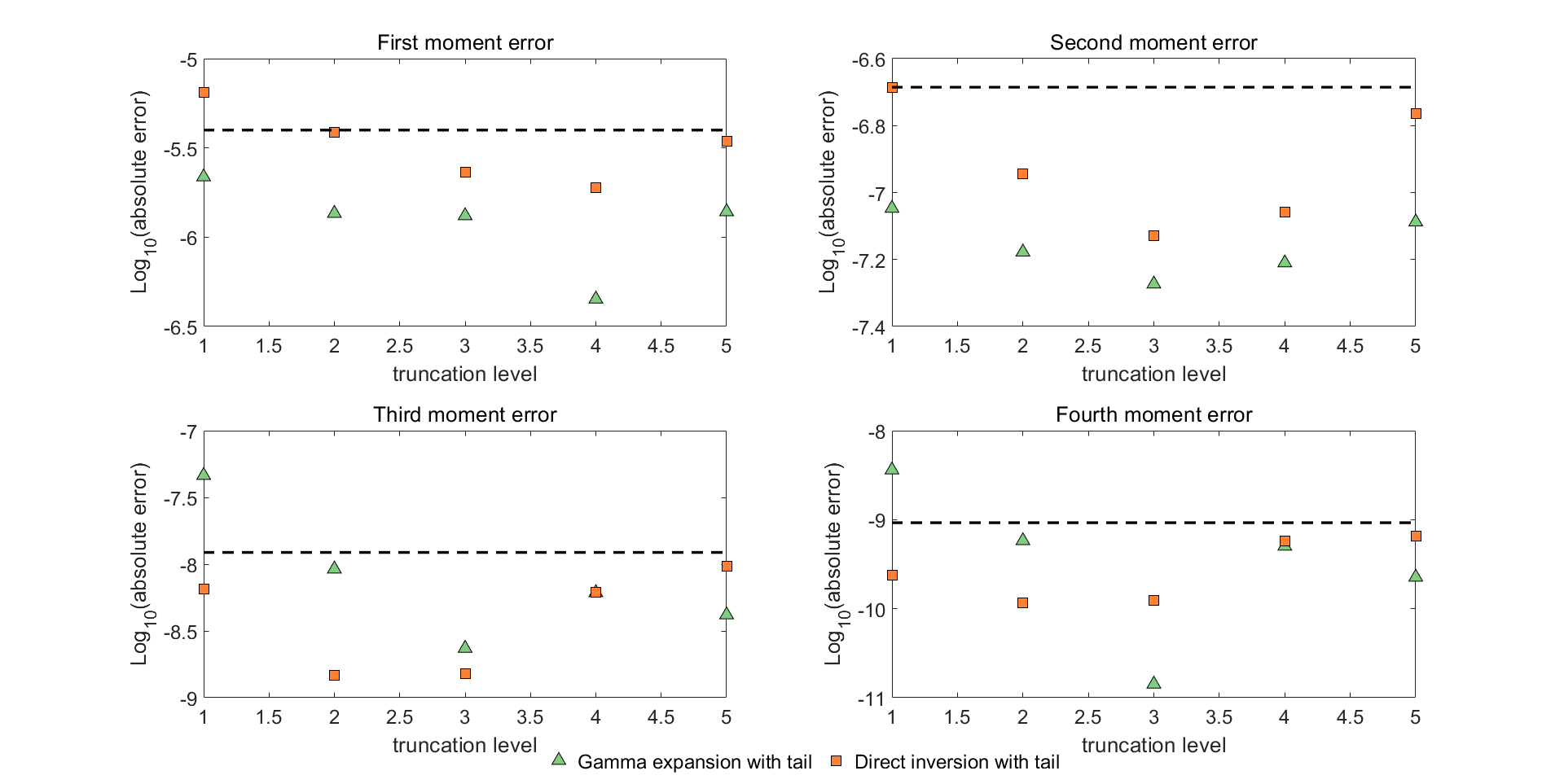}}  \\ 
    \subfloat[Case $4$: $v_0=0.010201, v_t=0.05$]{\label{fig:case4_mom_errors_50000000_0.05}\includegraphics[scale=0.28]{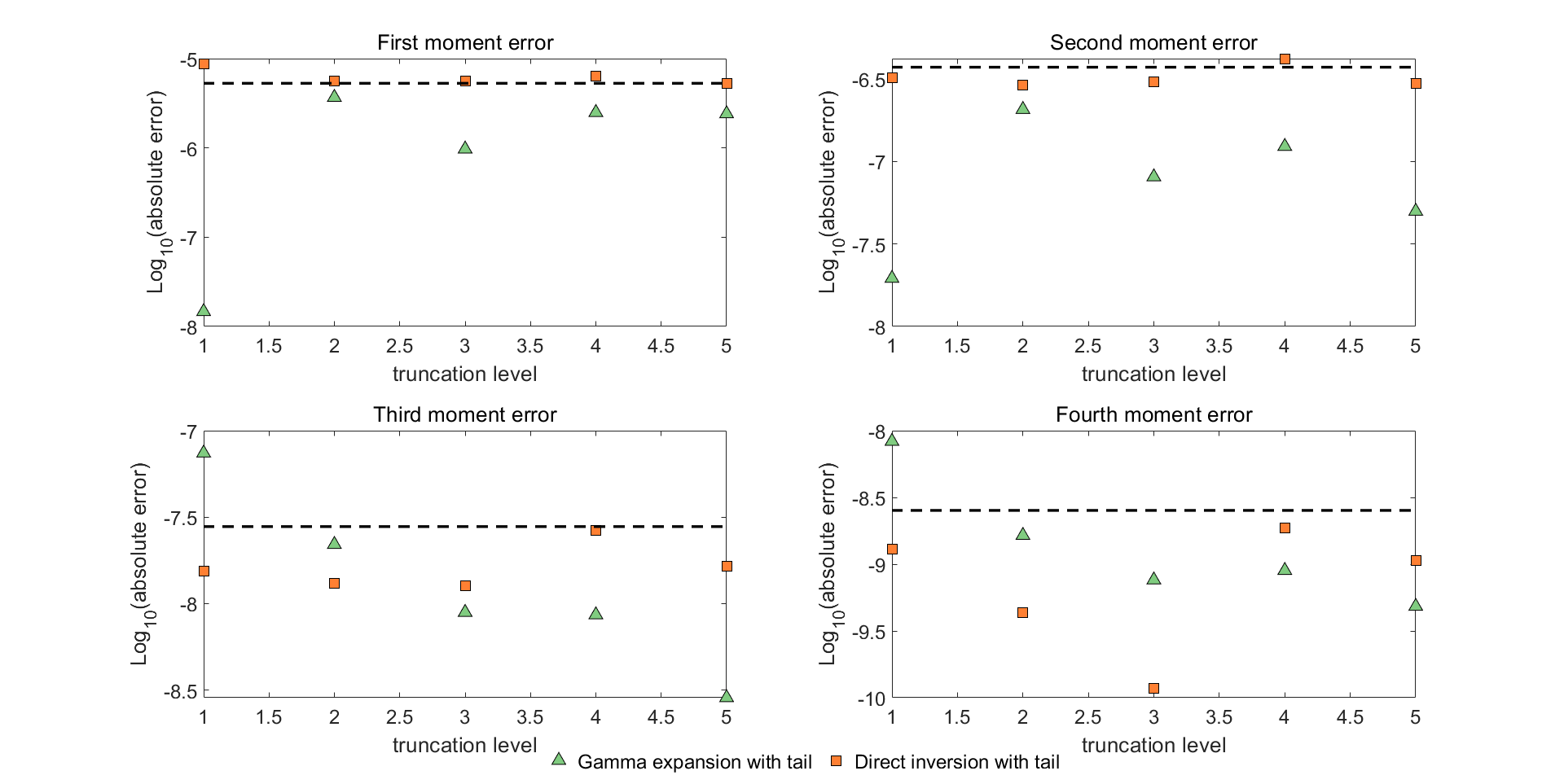}}  
    \caption{(cont.) We indicate the absolute errors in the first four moments of the conditional integral $\bar{I}$ simulated by direct inversion and gamma expansion versus the truncation levels for Case $4$ with different values for $v_t$. Both methods are implemented with tail simulation. We perform $5 \cdot 10^7$ simulations for each case. Below the dashed line, the errors are not statistically significant at the level of three standard deviations.}
    \label{fig:case4_mom_errors_50000000_2}
\end{figure}

While the figures for Case $1$ and $4$ have many details in common, they also reveal noteworthy differences in the first two moments. As illustrated in the upper panels in \cref{fig:case4_mom_errors_50000000_0.010201}, \cref{fig:case4_mom_errors_50000000_0.0025} and \cref{fig:case4_mom_errors_50000000_0.05} for Case $4$, most of the first and second moment errors in the direct inversion are slightly higher compared to those in the gamma expansion at the same truncation level. Errors of the two schemes considered in the first two moments for Case $1$ on the other hand seem to be of the same order to some extent with the same truncation level, which can be seen from the upper panels in \cref{fig:case1_mom_errors_50000000_0.04}, \cref{fig:case1_mom_errors_50000000_4} and \cref{fig:case1_mom_errors_50000000_0.000004}. In order to find a plausible explanation for this difference, we increase the sample size by a factor of $10$ and plot the resulting errors versus the truncation levels in \cref{fig:case1_mom_errors_500000000_1} and \cref{fig:case4_mom_errors_500000000_1}.
\begin{figure}[tbhp!]
    \centering
    \subfloat[Case $1$: $v_0=v_t=0.04$]{\label{fig:case1_mom_errors_500000000_0.04}\includegraphics[scale=0.28]{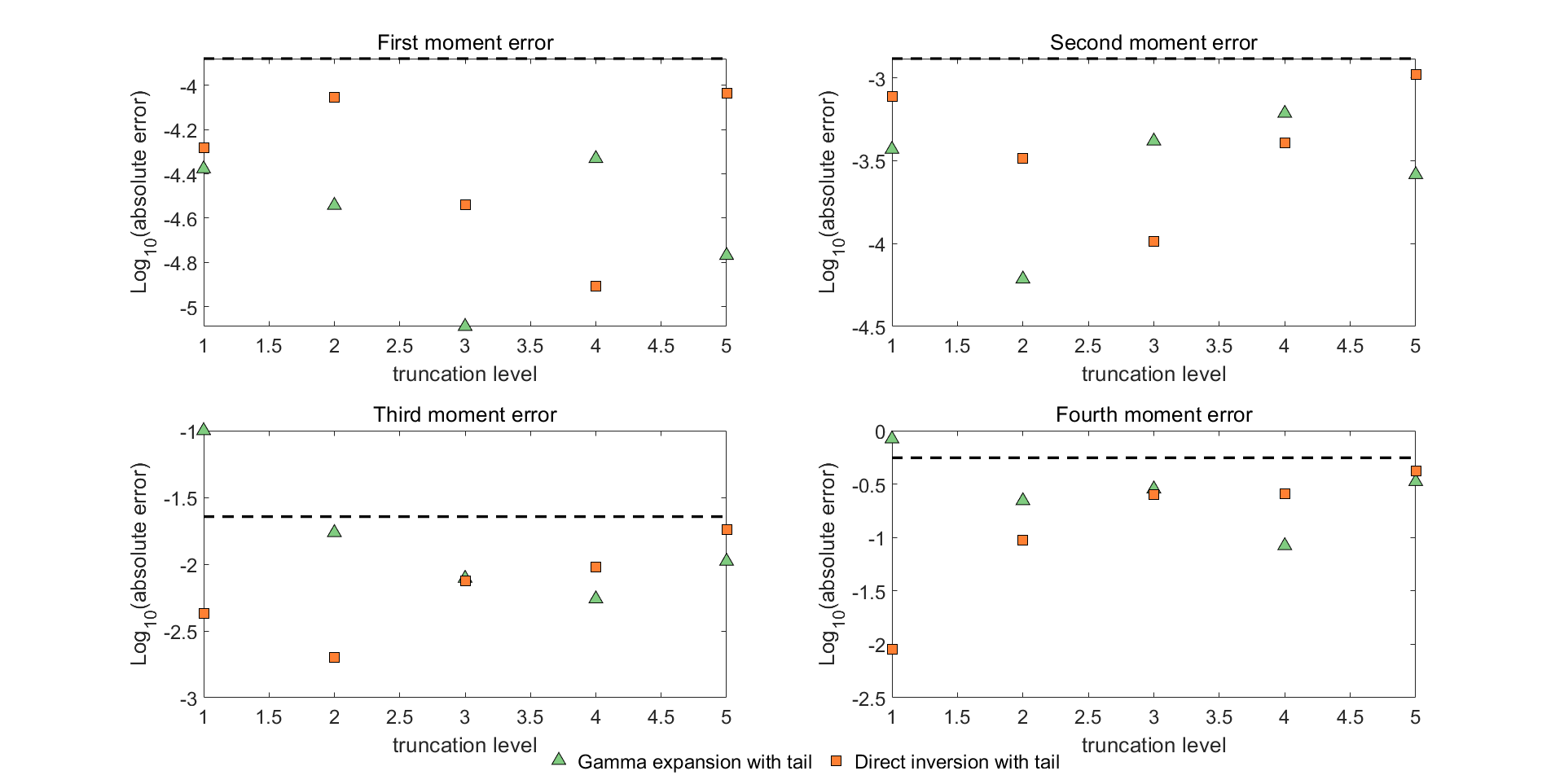}}   \\
    \subfloat[Case $1$: $v_0=0.04, v_t=4$]{\label{fig:case1_mom_errors_500000000_4}\includegraphics[scale=0.28]{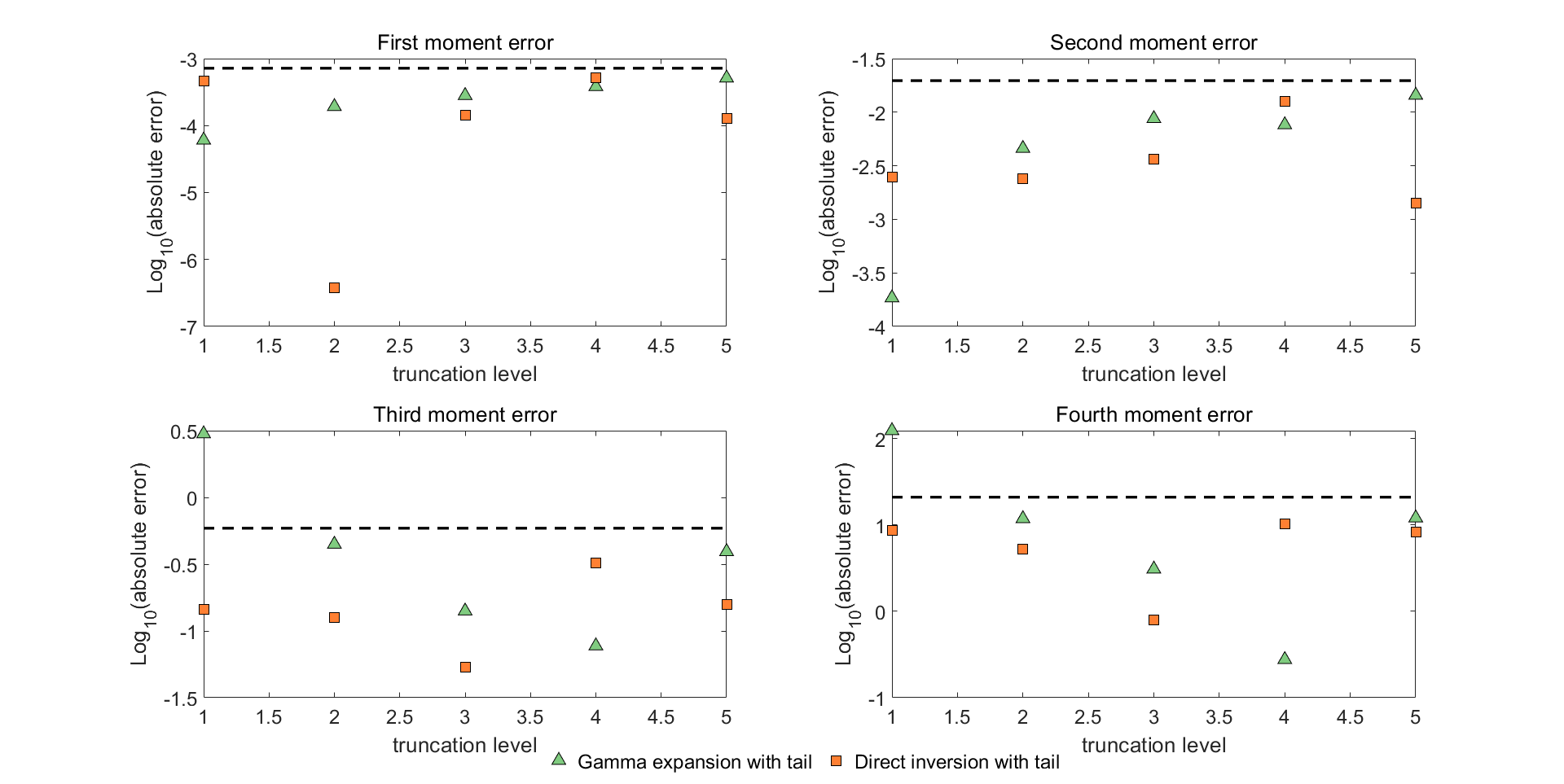}}
    \caption{We indicate the absolute errors in the first four moments of the conditional integral $\bar{I}$ simulated by direct inversion and gamma expansion versus the truncation levels for Case $1$ with different values for $v_t$. Both methods are implemented with tail simulation. We perform $5 \cdot 10^8$ simulations for each case. Below the dashed line, the errors are not statistically significant at the level of three standard deviations.}
    \label{fig:case1_mom_errors_500000000_1}
\end{figure}
\begin{figure}[tbhp!]
\ContinuedFloat 
    \centering
    \subfloat[Case $1$: $v_0=0.04, v_t=0.000004$]{\label{fig:case1_mom_errors_500000000_0.000004}\includegraphics[scale=0.28]{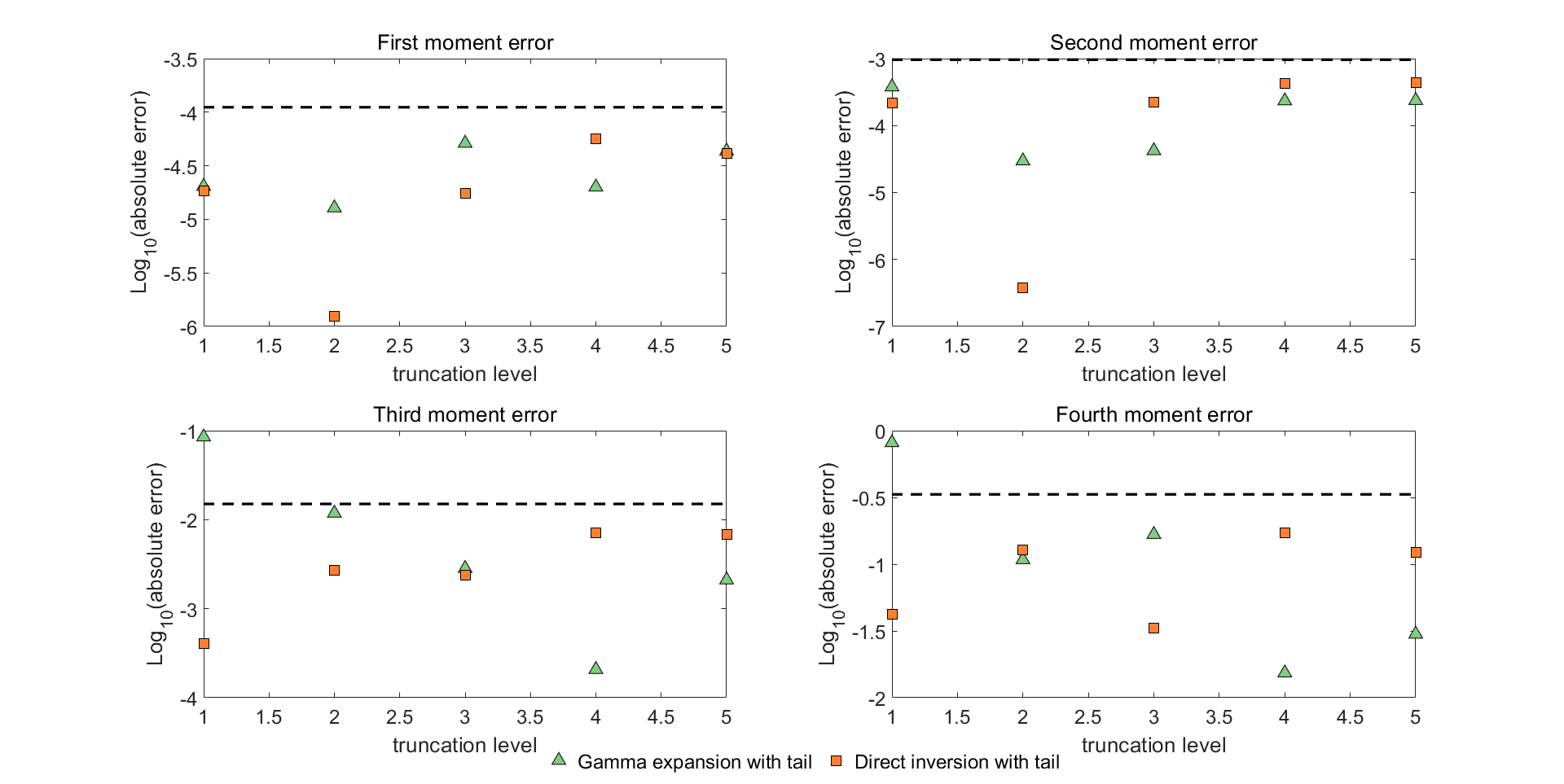}}  
    \caption{(cont.) We indicate the absolute errors in the first four moments of the conditional integral $\bar{I}$ simulated by direct inversion and gamma expansion versus the truncation levels for Case $1$ with different values for $v_t$. Both methods are implemented with tail simulation. We perform $5 \cdot 10^8$ simulations for each case. Below the dashed line, the errors are not statistically significant at the level of three standard deviations.}
    \label{fig:case1_mom_errors_500000000_2}
\end{figure}

For Case $1$, \cref{fig:case1_mom_errors_500000000_1} demonstrates the errors in all four moments based on the direct inversion are decreased as expected, i.e. proportional to the reciprocal of the square root of the sample size across all the values of $v_t$ and truncation levels. This confirms that the moment errors observed in \cref{fig:case1_mom_errors_50000000_1} using the direct inversion are dominated by the Monte Carlo error. On the other hand, for the gamma expansion we note in the upper panels of each subplots that the first two moments of the simulations for all five truncation levels are indeed matched with errors improving roughly according to the expected scaling when increasing the sample size. However, we see in the lower panels that the errors in the third and fourth moments hardly show any changes for lower truncation levels such as one and two while the accuracy for the other truncation levels is improved with the increase of the sample size. In fact, after reducing the Monte Carlo noise, there exists an even more clear decreasing trend for the higher order moment errors with the gamma expansion as the truncation level increases. This seems to corroborate the observations from \cref{fig:case1_mom_errors_50000000_1} for Case $1$, indicating that the gamma expansion with small truncation levels exhibits some bias while the direct inversion achieves the same accuracy, restricted by the Monte Carlo error, for all truncation levels.
\begin{figure}[tbhp!]
    \centering
    \subfloat[Case $4$: $v_0=v_t=0.010201$]{\label{fig:case4_mom_errors_500000000_0.010201}\includegraphics[scale=0.28]{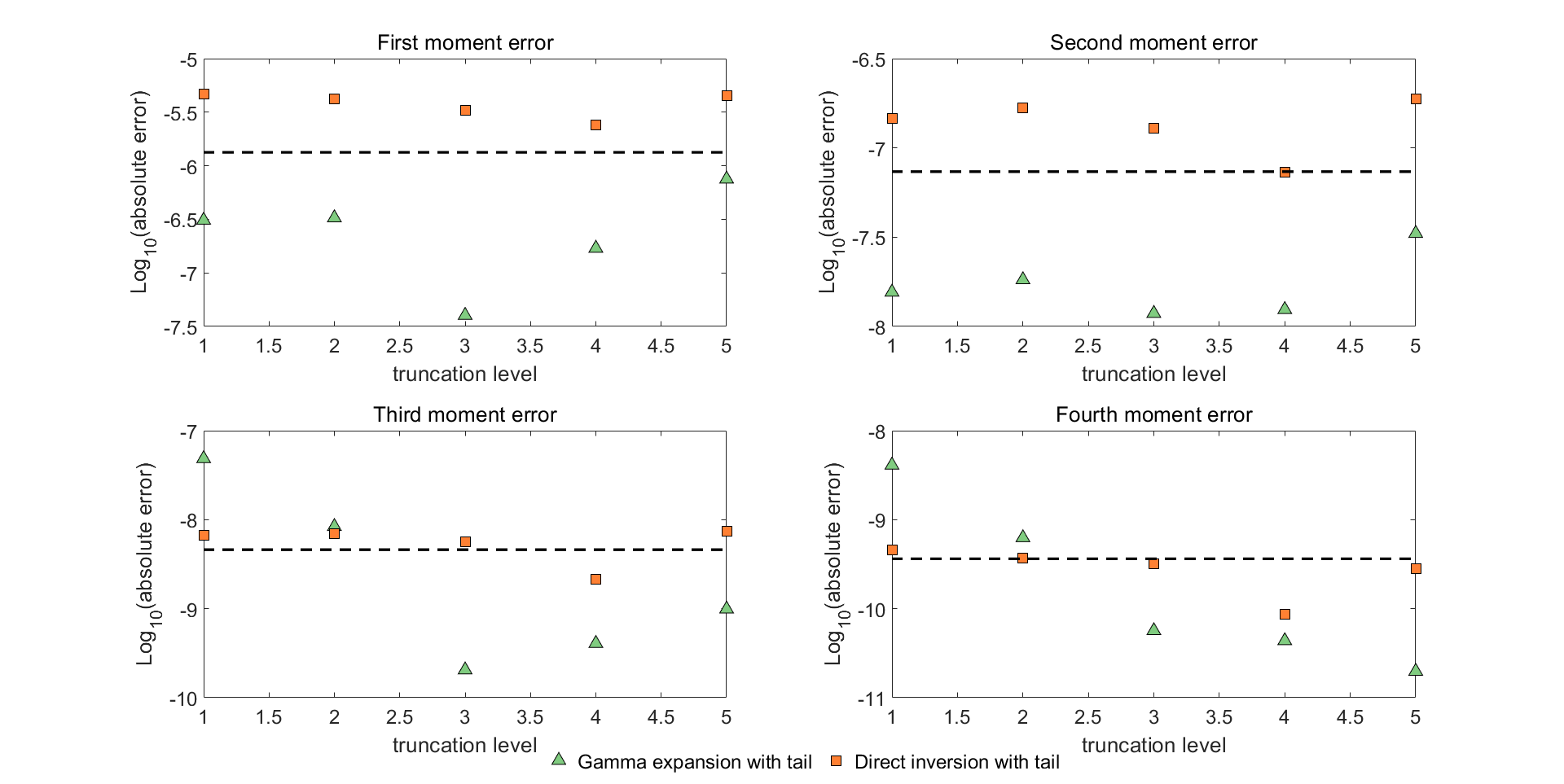}}  
    \caption{We indicate the absolute errors in the first four moments of the conditional integral $\bar{I}$ simulated by direct inversion and gamma expansion versus the truncation levels for Case $4$ with different values for $v_t$. Both methods are implemented with tail simulation. We perform $5 \cdot 10^8$ simulations for each case. Below the dashed line, the errors are not statistically significant at the level of three standard deviations.}
    \label{fig:case4_mom_errors_500000000_1}
\end{figure}
\begin{figure}[tbhp!]
\ContinuedFloat 
    \centering
    \subfloat[Case $4$: $v_0=0.010201, v_t=0.0025$]{\label{fig:case4_mom_errors_500000000_0.0025}\includegraphics[scale=0.28]{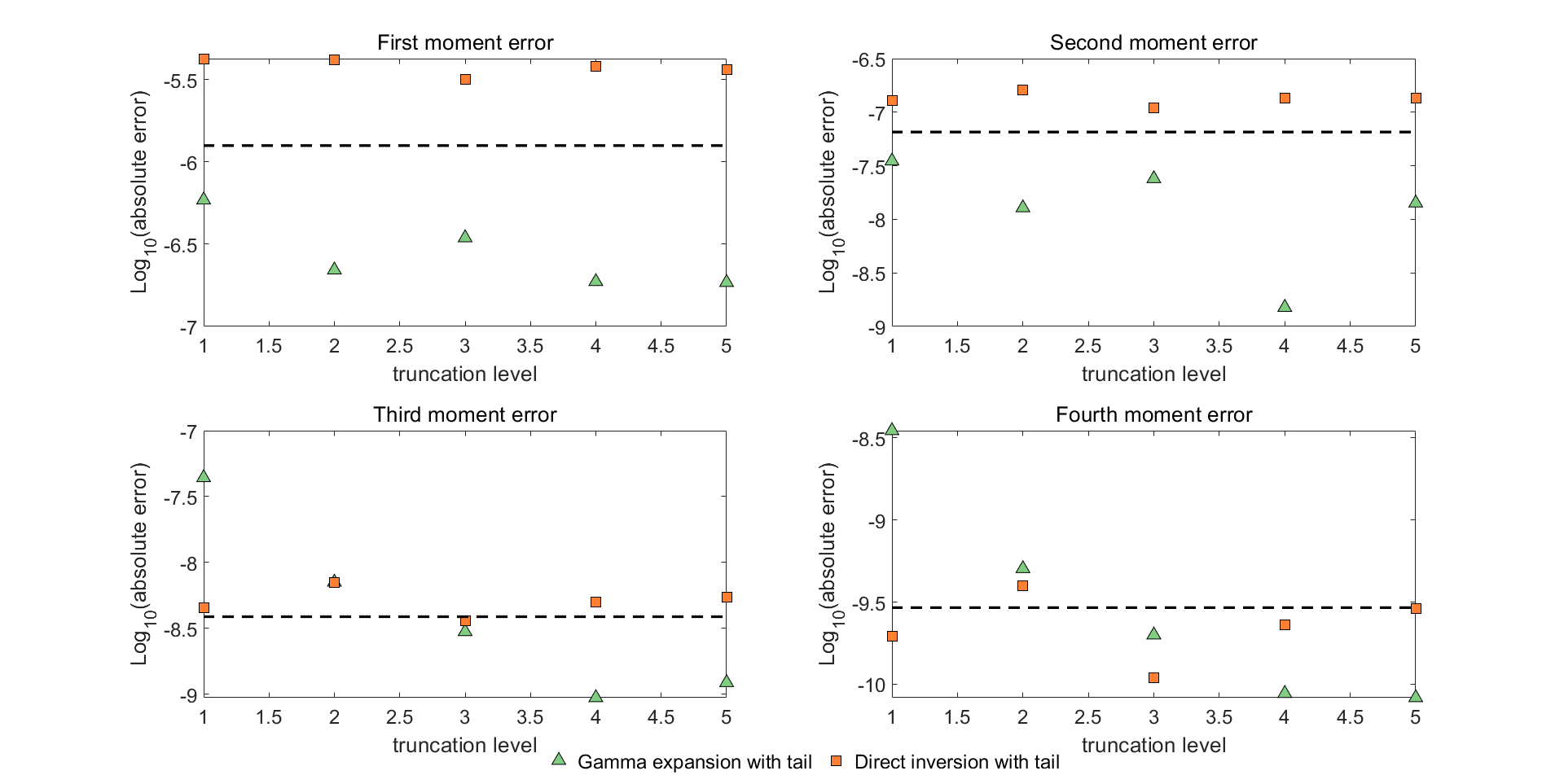}}  \\
     \subfloat[Case $4$: $v_0=0.010201, v_t=0.05$]{\label{fig:case4_mom_errors_500000000_0.05}\includegraphics[scale=0.28]{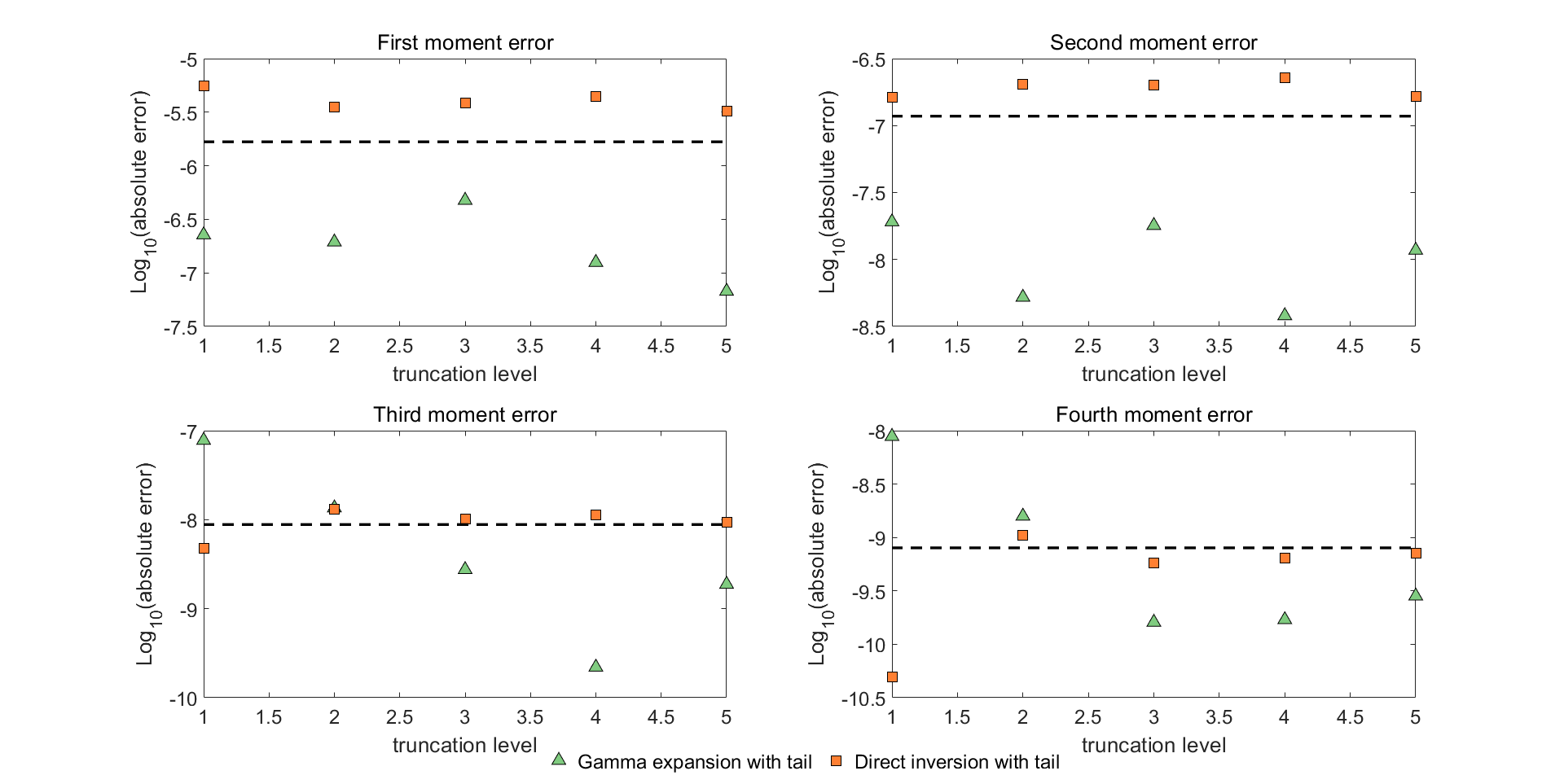}}  
    \caption{(cont.) We indicate the absolute errors in the first four moments of the conditional integral $\bar{I}$ simulated by direct inversion and gamma expansion versus the truncation levels for Case $4$ with different values for $v_t$. Both methods are implemented with tail simulation. We perform $5 \cdot 10^8$ simulations for each case. Below the dashed line, the errors are not statistically significant at the level of three standard deviations.}
    \label{fig:case4_mom_errors_500000000_2}
\end{figure}

In comparison, \cref{fig:case4_mom_errors_500000000_1} shows different behaviour for the errors related to the direct inversion for Case $4$ while similar conclusion can be reached for the gamma expansion as Case $1$. More specifically, we notice that all moment errors in direct inversion sampling for Case $4$ are invariant to increasing the sample size when the truncation levels are fixed. Further we observe that the errors, all remaining steady across a set of different truncation levels, become statistically significant when the number of samples is increased, especially for the first and second moments. Thus, this implies in Case $4$ the direct inversion performs equally well for all truncation levels, nevertheless, the accuracy of which is overridden by some bias. We should not fail to mention that the bias is roughly of the same order as the Monte Carlo error with $5 \cdot 10^7$ samples, whence it is not reflected in \cref{fig:case4_mom_errors_50000000_1}. This accounts for the finding for Case $4$ that the first and second moment errors for the direct inversion are always slightly larger than those for gamma expansion, where only Monte Carlo error is in presence.
We give a possible explanation for this bias as follows.

The reason for the bias with the direct inversion method for Case $4$ lies in the arithmetic precision we use for the parameter $h$, which is related to the random variable $X_2$. Recall that the proposed decomposition requires the rational parameter $h$ is given as a decimal with three significant figures. Let $\tilde{h}$ stand for the rounded number and $\tilde{X_2}$ denote the approximation to $X_2$ by replacing $h$ with $\tilde{h}$. Next, we give the exact errors in the first and second moments of $X_2$. Directly computing the first two moments using the series which defines $X_2$, we can write
\begin{align*}
    \mathbb{E} \left[X_2 \right]  
    = \frac{1}{3} \tau^2 h, \quad 
    \mathbb{E} \left[X_2^2 \right]  = \frac{2}{45} \tau^4 h + \frac{1}{9} \tau^4 h^2. 
\end{align*}
Then, the corresponding relative errors are
\begin{align*}
 \frac{\left\vert \mathbb{E} \left[X_2 \right]-  \mathbb{E} \left[\tilde{X_2} \right] \right\vert}{ \mathbb{E} \left[X_2 \right] } & = \frac{ \left\vert h - \tilde{h} \right\vert}{h} ,\\
    \frac{\left\vert \mathbb{E} \left[X_2^2 \right] -  \mathbb{E} \left[\tilde{X_2}^2 \right]    \right\vert}{ \mathbb{E} \left[X_2^2 \right]} & =  \frac{ \left\vert 2 \left( h - \tilde{h} \right) + 5 \left( h^2 - \tilde{h}^2 \right) \right\vert}{{2}  h + 5  h^2}.
\end{align*}
The above equations shows a linear scaling of the moment errors of $X_2$ in terms of the discrepancy between the true value $h$ and the approximated value $\tilde{h}$. \cref{tab:h} quotes the values for $h$ and $\tilde{h}$ for all the four European cases. Note that for Case $1$ and $3$ accurate values of $h$ are used while the relative errors for Case $2$ and $4$ are of order $10^{-3}$ and $10^{-4}$, respectively. In \cref{fig:moment_errors_X2} the panels show the relative errors in the first four moments of $X_2$ for Case $1$ to Case $4$ using $10^8$ and $10^9$ simulations. For Case $1$ and Case $3$, by successively increasing the sample size the high accuracy for the first four moments of $X_2$ sampled by direct inversion \cref{alg:dir-inv-Y_2^h} is indeed limited by the Monte Carlo error, which improves roughly according to the expected scale. However, the errors are invariant for Case $2$ and Case $4$ when increasing the sample size. For these two cases, the systematic Monte Carlo error is lower than the bias caused by replacing the true value $h$ with the approximated value $\tilde{h}$. Hence, the errors reflected in \cref{fig:moment_errors_X2}, dominated by the bias, fail to show improvement when the sample size is increased by a factor of $10$. 
\begin{table}[tbhp]
{\footnotesize
  \caption{True value $h$ and rounded value $\tilde{h}$.}\label{tab:h}}
\begin{center}
\begin{tabular}{ccccc} 
\hline
\addlinespace[0.5ex]
   &  \bf{ Case $\mathbf{1}$} & \bf{Case $\mathbf{2}$ } & \bf{ Case $\mathbf{3}$ } & \bf{ Case $\mathbf{4}$  }\\
\addlinespace[0.5ex]
\hline 
\addlinespace[0.5ex]
\bf{$h$} & $0.04000$   & $0.02963$    & $0.18000$  & $0.63418$    \\   
\bf{$\tilde{h}$} &  $0.04000$    &  $0.02950$       &    $0.18000$         &    $0.63400$        \\
\addlinespace[0.5ex]
\hline
\end{tabular}    
\end{center}
\end{table}
\begin{figure}[tbhp!]
    \centering
    \includegraphics[scale=0.28]{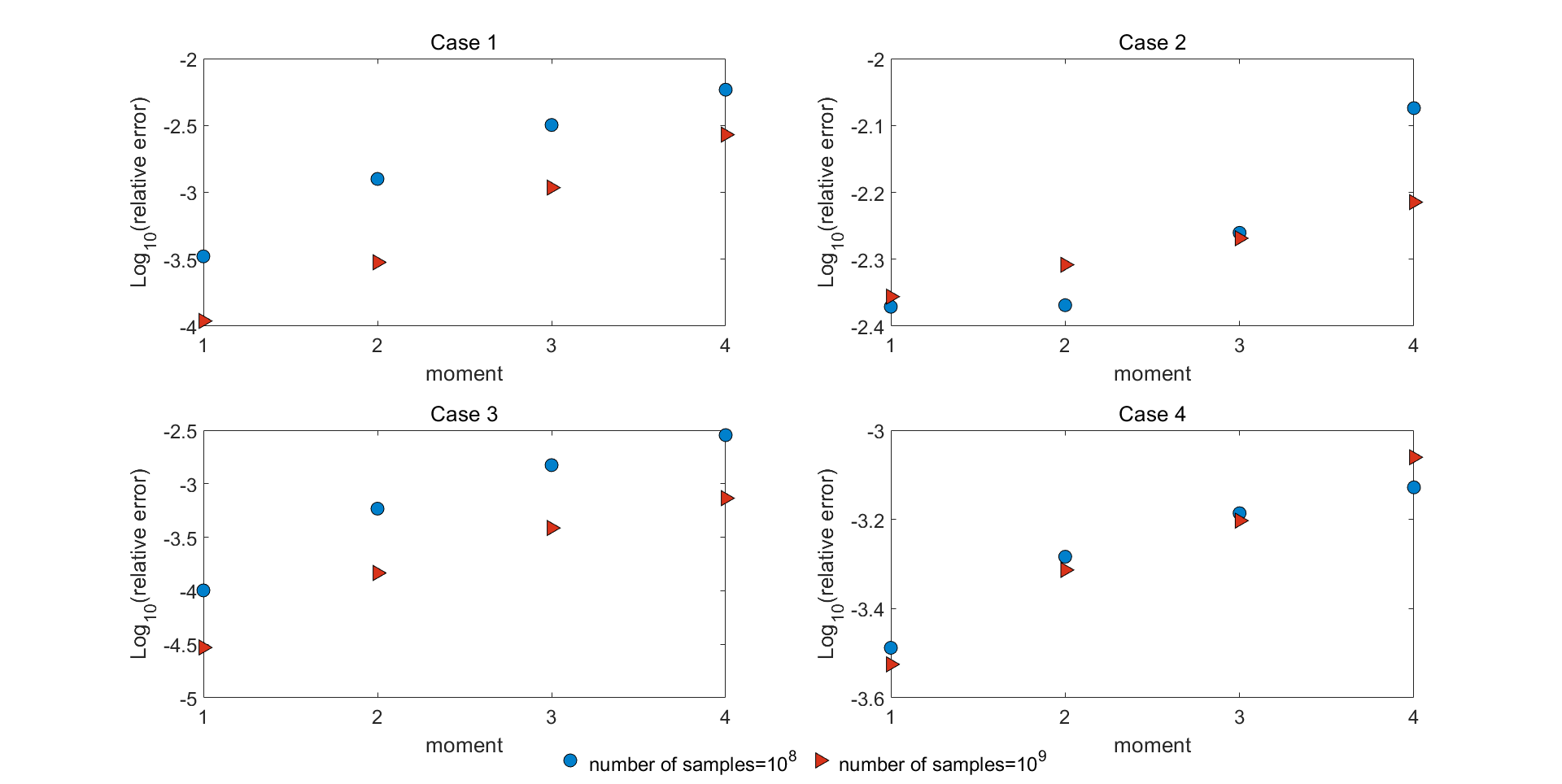}  \\
    \caption{We plot the relative errors in the first four moments of $X_2$ simulated by direct inversion \cref{alg:dir-inv-Y_2^h} for Case $1$ to Case $4$. By increasing the sample size by a factor of $10$, we note that the accuracy in the moment errors is improved as expected for Case $1$ and Case $3$. The four moment errors are invariant for Case $2$ and Case $4$ when increasing the sample size, suggesting possible bias in the direct inversion for these two cases.}
    \label{fig:moment_errors_X2}
\end{figure}

Based on the above analysis, we conclude that the direct inversion method for $X_2$ exhibits some small bias when approximation of the parameter $h$ is adopted. This can conceivably lead to the bias of the general direct inversion scheme for the conditional integral $\bar{I}$. However, this bias has nothing to do with the development of the method, but is associated with the decomposition technique and the arithmetic precision involved. Without loss of generality, this method can be extended to allow for a finer decomposition 
of the parameter $h$ given to any number of decimal places. In this sense, we expect that the accuracy available for this method will become more apparent.

\subsection{Option price}
In this section, we apply the direct inversion method and the gamma expansion to pricing four European call options and two path-dependent options including an Asian option and a barrier option. These two schemes are both based on the known conditional non-central chi-square distribution for the variance process and the conditional lognormal distribution for the asset price. We further compare the above two methods with the full truncation scheme of Lord, Koekkoek and Van Dijk \cite{lord2006comparison}, which is a time stepping method with asset price and variance simulated on discrete time grids. This type of equidistant discretization scheme for the one-dimensional CIR process has been shown to have an arbitrarily slow convergence rate in the strong sense in general; see Hefter and Jentzen \cite{Hefter2019arbitrarily}. Hence, developing other simulation methods becomes essential for practical purposes.   

\begin{figure}[tbhp!]
    \centering
    \subfloat[Case $1$]{\label{fig:option_price_case1}\includegraphics[scale=0.28]{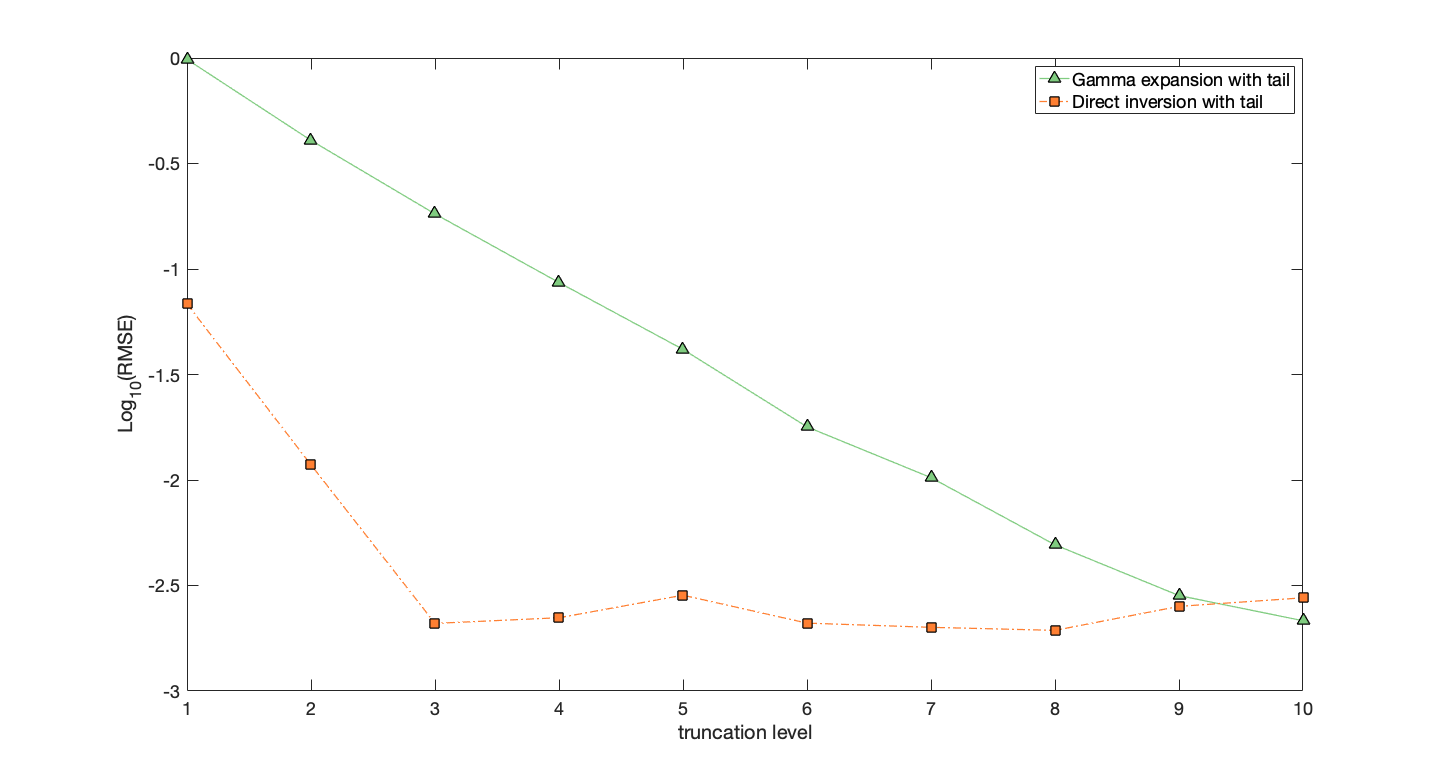}}  \\ 
    \subfloat[Case $4$]{\label{fig:option_price_case4}\includegraphics[scale=0.28]{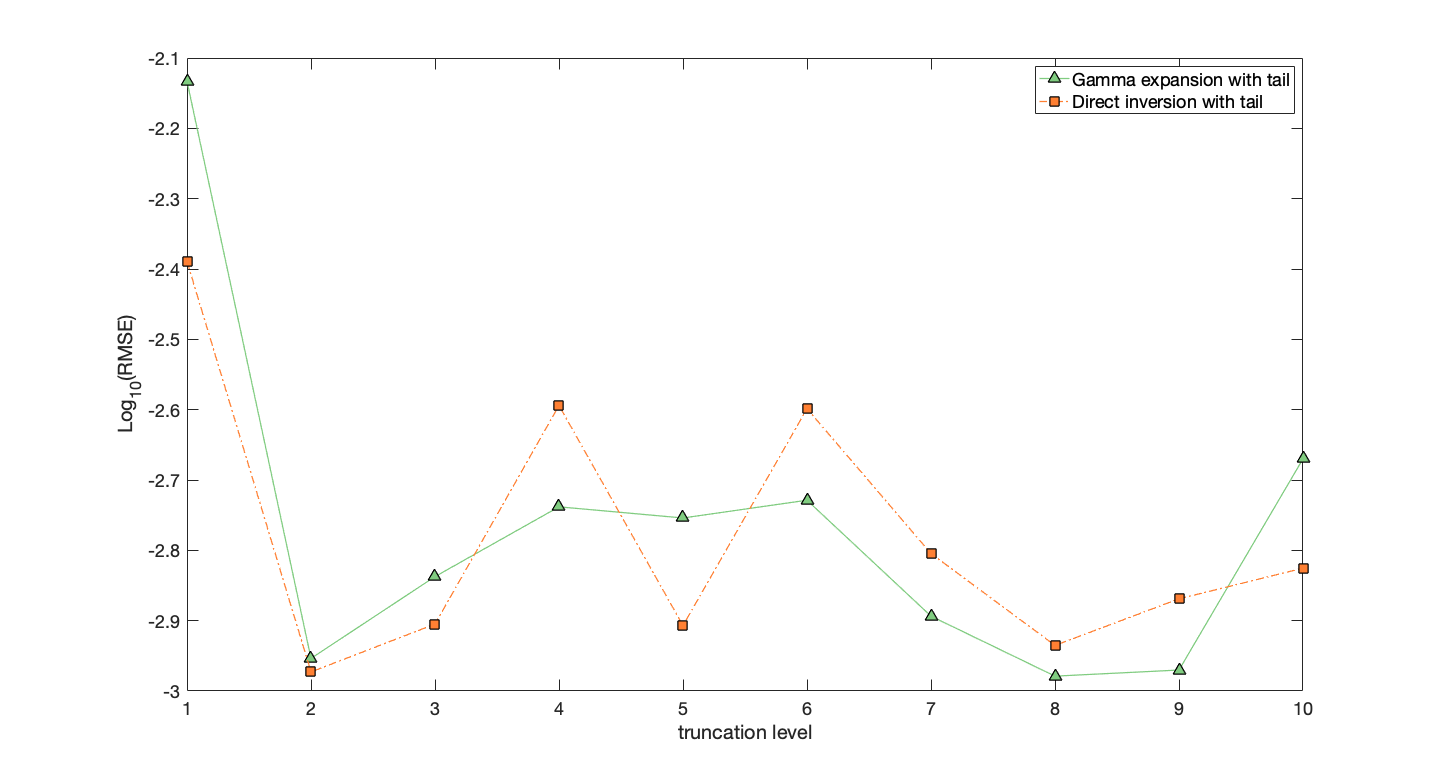}} 
    \caption{We show the root mean square error in the option price (strike $100$) versus the truncation level for Case $1$ and Case $4$. We use a sample size of $5 \cdot 10^7$ with truncation levels increasing in integers from $1$ to $10$.}
    \label{fig:option_price}
\end{figure}

We first demonstrate the tradeoff between the truncation level and accuracy for the direct inversion scheme and the gamma expansion. \cref{fig:option_price} plots the root mean square error for the price of an at the money European call option against the truncation level for Case $1$ and Case $4$. For both methods, we use a sample size of $5 \cdot 10^7$ and truncate after $M$ terms, increasing $M$ in integers from $1$ to $10$. For Case $1$, the direct inversion exhibits a faster convergence rate, revealed by the steeper slope in \cref{fig:option_price_case1}, in contrast with the gamma expansion. Indeed, truncation after three terms already provides a satisfactory estimator with error curve eventually becoming noisy in the larger $M$ regime. To obtain the same accuracy, many more terms up to $M=10$ are required for the gamma expansion. For Case $4$, increasing $M$ from one to two indeed helps to reduce the error. However, further increase in $M$ does not seem to bring improvement to the error for both methods, as seen from the horizontal error curves with small fluctuations in \cref{fig:option_price_case4}. This implies that approximations with small $M$ are sufficient to achieve acceptable accuracy.

\begin{remark}
\label{rem:numerical_results}
Similar conclusions as Case $1$ can be reached for Case $2$ and Case $3$ with numerical results presented in Shen \cite{shen2019numerical}. Simulations with different strikes for in the money and out of the money options for the above four test cases are also included there. 
\end{remark}
\begin{figure}[tbhp!]
\centering
     \centering
     \includegraphics[scale=0.28]{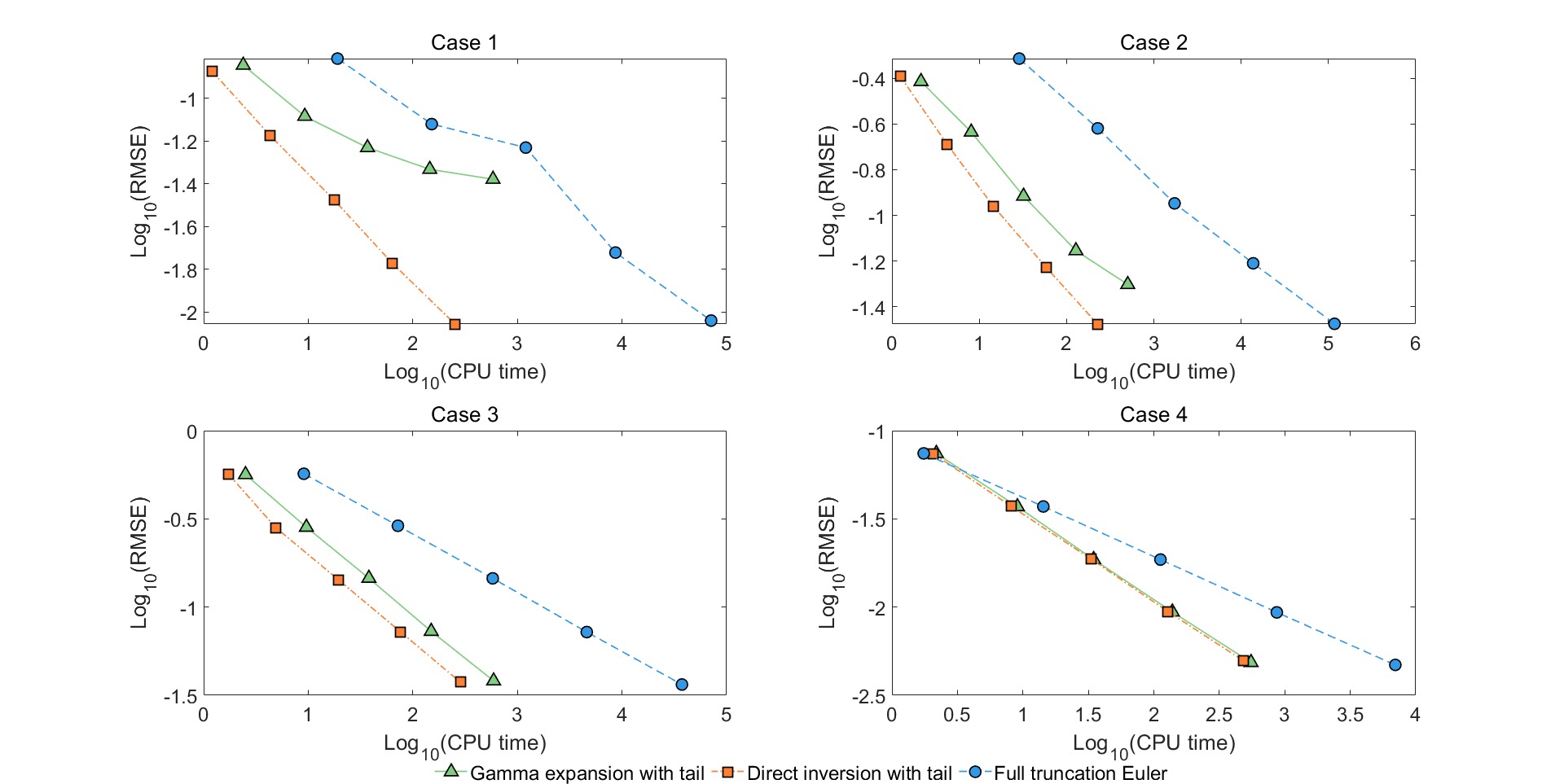}
    \caption{We show the convergence of the root mean square error in the option price (strike $100$) for Case $1$ to Case $4$ of gamma expansion and direction inversion, both at a truncation level $M=5$, and full truncation Euler scheme, with number of time steps equal to the square root of the sample size.}
    \label{fig:payoff_100}
\end{figure}

Next, we give the comparisons between the direct inversion method, the gamma expansion and the full truncation Euler scheme. In \cref{fig:payoff_100}, we plot the root mean square error in the option price as a function of the CPU time required on a log-log$_{10}$ scale for all schemes. For fair comparisons, we implement these three methods as efficiently as possible and generate the CPU time using compiled Matlab code. For the two non-discretization methods, we choose to use truncation level $M=5$. For the full truncation Euler method, we set the number of time steps equal to the square root of the sample size. This is motivated by the optimal allocation for the number of time steps from Duffie and Glynn \cite{duffie1995efficient}, which is proportional to the square root of the number of trials for methods with weak order of convergence being equal to one; see Broadie and Kaya \cite{broadie2006exact} and Lord, Koekkoek and Van Dijk \cite{lord2006comparison}. 

We can see from the upper panels in \cref{fig:payoff_100} that the bias in the gamma expansion with $M=5$ for Case $1$ and Case $2$ eventually dominates the root mean square error when the number of sample trails increases. By comparison, the root mean square errors for the direct inversion and full truncation Euler scheme are declining monotonically, with the former presenting a more rapid rate with reduced computational cost. For Case $3$ and Case $4$, the two non-discretization methods both outperform the full truncation Euler scheme, which has a slower convergence rate reflected by the less steeper slope in the graph. 

With regard to the computing time, the gamma expansion is at least two to three times slower than the direct inversion with similar accuracy for Case $1$, Case $2$ and Case $3$. For Case $4$, the new method takes much more time compared with Case $1$. This is because more effort is needed for the acceptance-rejection sampling of Case $4$ due to the slightly unfavourable values for the model parameters. Although the time needed for the direct inversion is marginally more than the gamma expansion for Case $4$ with $M=1$, as the desired accuracy is increased the new method requires less computational budget. In summary, we conclude that the performance of the direct inversion is the best among the three schemes considered here.

Now we turn to the pricing of options with payoffs depending on sample paths. We first consider an at the money Asian option with yearly fixings, the payoff of which is determined by the average of asset prices at the end of each year. We show in \cref{fig:asian} the root mean square error of the price versus the CPU time on a log-log$_{10}$ scale. For the direct inversion and the gamma expansion, we truncate the series after $M=1$ and simulate the asset prices for each year. Within one year, the terminal value is obtained directly using a single step. For the time discretization scheme, multiple steps are needed for each year. In this test, the number of time steps is taken to be the square root of the sample size in a similar manner to Broadie and Kaya \cite{broadie2006exact} and Smith \cite{smith2007almost}. 
\begin{figure}[H]
\centering
     \centering
     \includegraphics[scale=0.28]{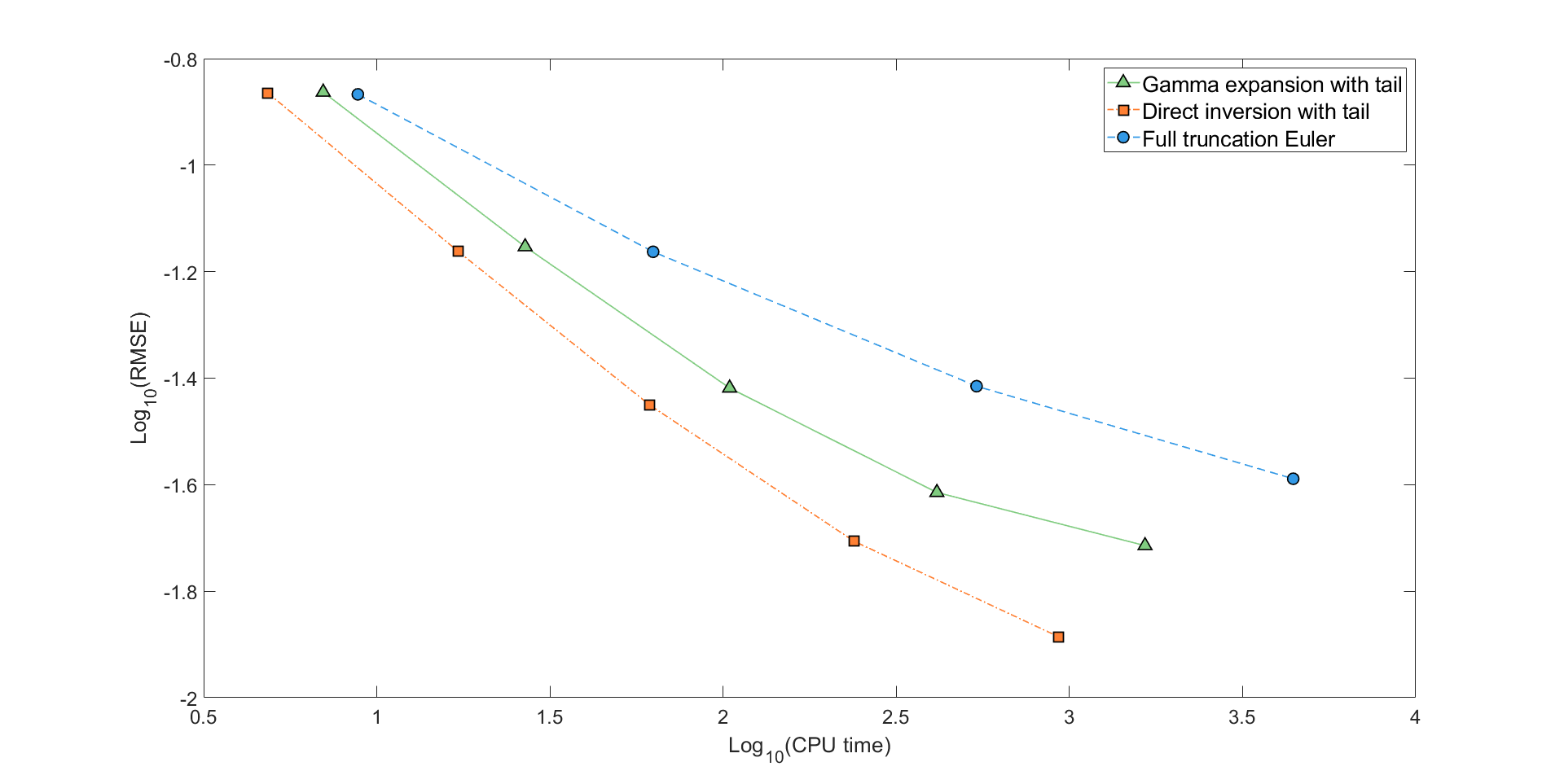}
    \caption{We show the convergence of the root mean square error in the option price (strike $100$) for the Asian option of gamma expansion and direction inversion, both at a truncation level $M=1$, and full truncation Euler scheme, with number of time steps equal to the square root of the sample size.}
    \label{fig:asian}
\end{figure}
\begin{table}[tbhp]
{\footnotesize
  \caption{Estimated prices with standard errors and CPU time for the barrier option. \label{tab:barrier}}}
\begin{center}
   \begin{tabular}{cccc}
\hline
\addlinespace[0.5ex]
\bf{Stepsize}              &                 & \bf{Direct inversion} & \bf{Gamma expansion} \\
\addlinespace[0.5ex]
\hline
\addlinespace[0.5ex]
\multirow{3}{*}{$1$}    & Estimated price & $0.68944$                    & $0.68908$                   \\
                      & Standard error & $0.00046$                    & $0.00046$                   \\
                      & CPU time        & $121.33$                  & $179.86$                 \\
\addlinespace[0.5ex] 
\hline
\addlinespace[0.5ex]                     
\multirow{3}{*}{$1/2$}  & Estimated price & $0.65891$                    & $0.65892$                   \\
                      & Standard error & $0.00047$                    & $0.00047$                   \\
                      & CPU time        & $202.52$                  & $328.76$                 \\
\addlinespace[0.5ex] 
\hline
\addlinespace[0.5ex]
\multirow{3}{*}{$1/4$}  & Estimated price & $0.63105$                    & $0.63182$                   \\
                      & Standard error & $0.00048$                    & $0.00048$                   \\
                      & CPU time        & $353.68$                  & $593.03$                 \\
\addlinespace[0.5ex] 
\hline
\addlinespace[0.5ex]
\multirow{3}{*}{$1/8$}  & Estimated price & $0.60544$                    & $0.60565$                   \\
                      & Standard error & $0.00049$                    & $0.00049$                   \\
                      & CPU time        & $653.35$                  & $1152.22$                \\
\addlinespace[0.5ex] 
\hline
\addlinespace[0.5ex]
\multirow{3}{*}{$1/16$} & Estimated price & $0.58364$                    & $0.58424$                   \\
                      & Standard error & $0.00049$                    & $0.00049$                   \\
                      & CPU time        & $1224.06$                 & $2278.41$                \\
\addlinespace[0.5ex] 
\hline
\addlinespace[0.5ex]
\multirow{3}{*}{$1/32$} & Estimated price & $0.56563$                    & $0.56556$                   \\
                      & Standard error & $0.00050$                    & $0.00050$                   \\
                      & CPU time        & $2173.38$                 & $3952.07$                \\
\addlinespace[0.5ex] 
\hline
\addlinespace[0.5ex]
\multirow{3}{*}{$1/64$}                  & Estimated price & $0.54983$                    & $0.54943$                   \\
                      & Standard error & $0.00050$                    & $0.00050$                   \\
                      & CPU time        & $4984.79$                 & $9904.97$                \\
\addlinespace[0.5ex] 
\hline
\addlinespace[0.5ex]
\multirow{3}{*}{$1/128$}                 & Estimated price & $0.53743$                    & $0.53763$                   \\
                      & Standard error & $0.00050$                    & $0.00050$                   \\
                      & CPU time        & $8020.93$                 & $24521.84$ \\
\addlinespace[0.5ex]
\hline
\end{tabular} 
\end{center}
\end{table}

We observe that both the gamma expansion and the direct inversion, even with a lower truncation level, deliver similar accuracy compared to the full truncation scheme for small sample sizes. However when the number of simulations increases, bias of the estimated price starts to dominate the root mean square error for all three methods with the standard deviation decreasing according to the expected scaling, i.e. the inverse of the square root of the sample size. This dominance by the bias eventually decelerates the decrease of the root mean square error. Among the above three methods, the direct inversion produces the smallest bias. In terms of the computing time, very similar conclusions can be drawn as the European option cases. For similar accuracy, the direct inversion is approximately $2$ to $7$ times faster than the gamma expansion. The time required by the full truncation Euler scheme is by far the largest.

We end this section with a test for pricing a digital double no touch barrier option. The payoff for such an option is either one or zero unit of currency depending on whether the barriers have been crossed. In \cref{tab:barrier}, we report the estimated price and standard error together with the CPU time of the direct inversion and the gamma expansion at truncation level $M=1$ for a double no touch barrier option with barriers $90$ and $110$. We sample a total of $10^6$ paths for each case. We increase the number of time steps per year from $1$ to $128$ and monitor at each time step if the asset price has hit one of the two barriers. 

We see from \cref{tab:barrier} that as we decrease the stepsize, the estimated price of both the direct inversion and the gamma expansion is decreasing monotonically. This is in accordance with our expectation since when more dates are being monitored, there are more chances for the asset price to cross the barriers. Because of the nature of these two methods, we expect their estimated price will eventually be almost exact with negligible truncation errors when the asset price is monitored on a more frequent basis, for instance, every trading day. The results here are also consistent with those of the four schemes tested in Malham and Wiese \cite[Table $5$]{malham2013chi} and the PT, FT and ABR scheme in Lord, Koekkoek and Van Dijk \cite[Table $7$]{lord2006comparison} in terms of accuracy. Similar conclusions can be reached as the cases for European and Asian options in terms of the computational time. The time required for the gamma expansion is $1.5$ to $3$ times more than the direct inversion.

\section{Conclusion}
\label{sec:conclusion}
In this paper, we have designed a new series expansion for the time integrated variance process under the Heston stochastic volatility model. Our expansion is built on a change of measure argument and the decomposition techniques for the integral of squared Bessel bridges by Pitman and Yor \cite{pitman1982decomposition} and Glasserman and Kim \cite{glasserman2011gamma}. Acceptance-rejection and direct inversion methods are developed to realise the conditional integral. On combining this result with the method of Broadie and Kaya \cite{broadie2006exact}, almost exact simulations of the stock price and variance can be generated on the basis of their exact distributions. We compare our approach with Glasserman and Kim \cite{glasserman2011gamma} through pricing four practical and challenging options. Apart from that, two path-dependent options including an Asian option and a barrier option are also tested using the above two methods. Further comparisons with a standard time discretization method, i.e. the full truncation Euler scheme, are performed as well. Evidence implies faster computational speed with comparable error in our method.

The series representation and sampling techniques above can also be transferred to the generalised squared Ornstein-Uhlenbeck process $x_t$ with parameters $b \in \mathbb{R}$ and $\delta > 0$ given by
\begin{align*}
d x_t = \left( \delta + 2 b x_t \right) \, dt + 2 \sqrt{x_t} \, dW_t,
\end{align*}
where $W_t$ denotes a standard Brownian motion. Although in this paper we focus only on the case $ 0 < \delta <2$, the present result can be applied to other cases $\delta \geq 2$. In essence we need to find an appropriate decomposition for $\delta$ and hence establish efficient Chebyshev polynomial approximations required for the resulting direct inversion algorithm. The expansions derived in \cref{sec:simulation} will be helpful in determining the coefficients.

Lastly, we recommend a direction for future research. Our method entails an acceptance-rejection algorithm with acceptance probability depending on model parameters. Thus, it is difficult to measure its general computational complexity, i.e. the average number of iterations needed. Besides, in the application of risk management and trading, the acceptance-rejection scheme is less favourable as it will introduce considerable Monte Carlo noise in sensitivity analysis. For these reasons, an alternative should be considered. One realistic way to avoid the use of acceptance-rejection is to sample the Radon-Nikod\'{y}m derivative directly under the new probability measure instead.

\appendix
\section{Proof of \cref{lem:moments_of_remainder}}
\label{sec:appendix_proof_3.1}
For the remainder $R_1^K$, as stated in \cref{thm:Series_Rep}, the $S_{n,k}$ are independent and identically distributed random variables and the $P_n$ are independent Poisson random variables with mean $ \left( a_0 + a_\tau \right) 2^{n-1}/ \tau$. Taking the expectation of $R_1^K$ directly, we have
\begingroup
\allowdisplaybreaks
\begin{align*}
    \mathbb{E} \left[ R_1^K \right] 
    & =  \sum_{n=K+1}^\infty \frac{\tau^2}{4^n} \mathbb{E} \left[ P_n \right]  \mathbb{E} \left[  S_{n,k}    \right] \\
    & = \sum_{n=K+1}^\infty \frac{\tau^2}{4^n}   \left( \frac{a_0+a_\tau}{\tau} 2^{n-1} \right) \left( \frac{2}{\pi^2} \sum_{l=1}^\infty \frac{1}{l^2} \right) \\
    & =  \frac{ \left( a_0+a_\tau \right) \tau}{6} \frac{1}{2^K},
\end{align*}
\endgroup
where the last identity holds since
$\sum_{l=1}^\infty l^{-2} = \pi^2/6$ and 
$\sum_{n=K+1}^{\infty}  2^{-\left( n+1\right)} = 2^{-\left( K+1 \right)}$. Similarly, we can compute
\begingroup
\allowdisplaybreaks
\begin{align*}
    \mathrm{Var}\left[ R_1^K  \right]  
    & =  \sum_{n=K+1}^\infty \frac{\tau^4}{16^n}  \left( \mathrm{Var} \left[  P_n \right] \left( \mathbb{E} \left[ S_{n,k} \right] \right)^2 + \mathbb{E} \left[ P_n \right] \mathrm{Var} \left[ S_{n,k} \right]           \right) \\
    & = \sum_{n=K+1}^\infty \frac{\tau^4}{16^n} \left(    \frac{ a_0 + a_\tau  }{\tau} 2^{n-1} \right) \left(  \left(  \frac{2}{\pi^2} \sum_{l=1}^\infty \frac{1}{l^2} \right)^2 +  \frac{4}{ \pi^4} \sum_{l=1}^{\infty} \frac{1}{l^4}     \right) \\
    & =  \frac{ \left( a_0+a_\tau \right) \tau^3}{90} \frac{1}{8^K} ,
\end{align*}
\endgroup
where we use the formulae $\sum_{l=1}^\infty l^{-4} = \pi^4/90$ and $\sum_{n=K+1}^{\infty}  8^{-n} = 8^{-K} / 7$.

For the remainder $R_2^K$, similar to the calculations for the moments of $R_1^K$, we find  \begin{align*}
    \mathbb{E} \left[ R_2^K \right] & = \sum_{n=K+1}^\infty \frac{\tau^2}{4^n} \mathbb{E} \left[ C_{n}^{\delta/2} \right] \\
     & =  \sum_{n=K+1}^\infty \frac{\tau^2}{4^n} \left( \frac{2}{\pi^2} \sum_{l=1}^\infty    \frac{\delta}{2}    \frac{ 1  }{\left(l - \frac{1}{2} \right)^2}   \right) \\
     & = \frac{\delta \tau^2}{6} \frac{1}{4^K}.
\end{align*}
Note that the last step is a direct result of $\sum_{l=1}^\infty \left( 2 l -1 \right)^{-2} = \pi^2/8$ and $ \sum_{n=K+1}^\infty 4^{-n} = 4^{-K}/3$. Further we can proceed with the computation of the variance:
\begingroup
\allowdisplaybreaks
\begin{align*}
    \mathrm{Var} \left[ R_2^K \right] & = \sum_{n=K+1}^\infty \frac{\tau^4}{16^n} \mathrm{Var} \left[C_{n}^{\delta/2} \right] \\
    & = \sum_{n=K+1}^\infty \frac{\tau^4}{16^n}  \left(  \frac{4}{\pi^4} \sum_{l=1}^\infty \frac{\delta}{2} \frac{1}{\left( l- \frac{1}{2} \right)^4} \right) \\
    & =\frac{\delta \tau^4}{45} \frac{1}{16^K},
\end{align*}
\endgroup
in which we apply $\sum_{l=1}^\infty \left(2l-1 \right)^{-4} = \pi^4/96$ and $\sum_{n=K+1}^\infty 16^{-n} = 16^{-K}/15  $.

\section{Proof of \cref{thm:Asymp_Exp_pdf_Z^P}}
\label{sec:appendix_proof_3.2}
We use the method of steepest descents to find the limiting behaviour as $P\to +\infty$ of the probability density function $f_{Z^P}$ of the standardised random variable $Z^P$ given by 
\begin{align}
f_{Z^P}\left(x\right)  = \frac{1}{4 \pi} \sqrt{\frac{2 P}{45}}  \int_{- \infty}^{+ \infty} \exp{ \left(P \rho \left( z; \beta \right) \right) } \, dz,
\label{eq:f_^P_integral}
\end{align}
where $\beta=x \sqrt{2/45}  / \sqrt{P}$ and 
\begin{align*}
\rho \left( z; \beta \right) = \log{ \left( \frac{\sqrt{z \mathrm{i}}}{\sinh{\sqrt{z \mathrm{i}}}} \right)} + z \mathrm{i} \left( \frac{1}{6} + \frac{1}{2} \beta \right).
\end{align*}
The complex logarithm is chosen to be the principal branch. Note that the support of $Z^P$ is $\left[ - \sqrt{5P} / \sqrt{2}, + \infty \right)$ and hence we focus on the case when $\beta \geq -1/3$.

We are interested in the saddle point $z_0$ such that $\rho^{\prime} \left( z_0;\beta \right)=0$. Before that, let 
\begin{align*}
    \varsigma \left( z \right) \coloneqq \frac{ \sqrt{z \mathrm{i} }}{ \sinh{\sqrt{z \mathrm{i} }}},
\end{align*}
which gives 
\begin{align*}
    \rho \left( z; \beta \right) = \log{ \left( \varsigma \left( z \right)\right)} + z \mathrm{i} \left( \frac{1}{6} + \frac{1}{2} \beta \right).
\end{align*}
We observe $\varsigma \left( z \right)$ and $\rho \left( z; \beta \right)$ are  analytic when $\mathrm{Im} \left( z \right) < \pi^2$ after defining $\varsigma\left( 0 \right) \coloneqq 1$ and $\rho \left( 0; \beta \right) \coloneqq 0$.

Bleistein and Handelsman \cite[Chapter $7.6$]{bleistein1986asymptotic} suggest that we should seek a saddle point near $z=0$, which will be the dominant one. To obtain an explicit form for the saddle point, we take advantage of the series expansions of $\rho \left( z;\beta \right)$ and its differentiation $\rho^{\prime} \left( z; \beta \right)$. Let us first consider the series expansion of $\rho \left( z;\beta \right)$ about $z=0$, which is of the form
\begin{align}
\rho \left(z; \beta \right) = \frac{1}{2} \beta \mathrm{i} z + \sum_{k=2}^{\infty} \hat{r}_k    z^k,
\label{eq:rho_taylor_0}  
\end{align}
where 
\begin{align}
    \label{eq:coeff_r}
   \hat{r}_2  = - \frac{1}{180}, \quad \hat{r}_3  = \frac{\mathrm{i}}{2835}
\end{align}
and so forth. Although we can compute the coefficients $\hat{r}_k$ analytically through Taylor expansion of $\rho \left( z;\beta \right)$ up to any order, we compute them using Maple in practice. Hence, its differentiation can be expressed as 
\begin{align*}
\rho^{\prime} \left( z; \beta \right) = \frac{1}{2} \beta \mathrm{i} + \sum_{k=2}^{\infty} k \hat{r}_k z^{k-1} .
\end{align*}
The above two series converge pointwise in the domain where $\left\vert z \right\vert < \pi^2$. It seems that a precise form for the saddle point $z_0$ such that $\rho^{\prime} \left( z_0;\beta \right)=0$ is not obtainable. However, we can get a good approximation by making use of the smallness of $\beta$. For $\beta \ll 1$, i.e. $\left\vert x \right\vert \ll \sqrt{P} / \sqrt{2/45}$, we solve the above equation by iteration. In fact, after two iterations, we have the following expression for the desired saddle point
\begin{align*}
z_0= 45 \mathrm{i} \beta - \frac{1350}{7} \mathrm{i} \beta^2 + \mathrm{O} \left( \beta^3 \right).
\end{align*}

We show in \cref{sec:appendix_unique_existence} that the saddle point exists and is unique in the domain of analyticity of $\rho \left( z; \beta \right)$, i.e. $\left\{ z \in \mathbb{C} : \mathrm{Im} \left(z\right) < \pi^2 \right\}$ for any $\beta > -1/3$, i.e. $x > -\sqrt{5P/2}$. Further since $\rho^{\prime} \left( z; \beta \right)$ is analytic in this region, we can show the solution $z_0 = z_0 \left( \beta \right)$ is also analytic around $\beta = 0$ by the analytic implicit function theorem and thus $z_0 = z_0 \left( \beta \right)$ has a Taylor series expansion about $\beta =0$ valid for sufficiently small $\left\vert \beta \right\vert$. Hence we can write
\begin{align}
z_0 = \beta \sum_{k=0}^{\infty} \hat{\xi}_k \beta^k
\label{eq:saddle_point}
\end{align}
for sufficiently small $\left\vert \beta \right\vert$, where
\begin{align}
    \label{eq:coeff_xi}
\hat{\xi}_0 = 45 \mathrm{i}, \quad \hat{\xi}_1 = - \frac{1350 \mathrm{i}}{7}  
\end{align}
and so on. Again, all these coefficients $\hat{\xi}_k$ are calculated via Maple in practice. Notice that the saddle point $z_0$ is near the origin and along the imaginary axis. 

Now, we can expand $\rho \left( z; \beta \right)$ as a Taylor series near the saddle point $z_0$
\begin{align}
\rho \left( z; \beta \right) =\rho \left( z_0; \beta \right) + \frac{1}{2!} \rho^{\prime \prime} \left( z_0; \beta \right) \left( z-z_0 \right)^2 + \left( z-z_0 \right)^3 \sum_{k \geq 0} \frac{\rho^{ \left( k+3 \right)} \left( z_0; \beta \right)}{ \left( k+3 \right)!} \left( z-z_0 \right)^{k},
\label{eq:rho_taylor_z0}
\end{align}
which converges in a neighbourhood of $z_0$. For preparations, we must evaluate $\rho^{ \left( k \right)} \left( z; \beta \right)$ for $k \geq 2$ at $z=z_0$. Differentiating the series expansion for $\rho$ in equation \cref{eq:rho_taylor_0} leads to 
\begin{align*}
\rho^{ \left( n \right) } \left( z; \beta \right)  = \sum_{k=n}^{\infty} k \left( k-1 \right) \cdots \left( k-n+1 \right) \hat{r}_k z^{k-n} 
 = \sum_{k=n}^{\infty} \hat{\varphi}_{k,n} z^{k-n}
\end{align*}
for $n \geq 2$, where 
\begin{align}
    \label{eq:coeff_varphi}
\hat{\varphi}_{k,n} \coloneqq k \left( k-1 \right) \cdots \left( k-n+1 \right) \hat{r}_k  
\end{align}
for $k \geq n$. This series converges in the same domain as the expansion for $\rho$; see equation \cref{eq:rho_taylor_0}. To substitute the Taylor series \cref{eq:saddle_point} regarding the saddle point $z_0$ into the above equation, we require $z_0$ to be within its radius of convergence, i.e. $\left\vert z_0 \right\vert < \pi^2$. Hence we restrict $\beta$ such that $\beta > -1/3 - \left( 1- \pi \coth{\pi} \right) / \pi^2$ (see \cref{rem:beta_restriction}). By noting that 
\begin{align}
    z_0 ^ j  = 
    \beta^j \sum_{l=0}^{\infty} \hat{\upsilon}_{l,j} \beta^l
\label{eq:z0_product}
\end{align}
for $j \geq 0$, where 
\begin{align}
\label{eq:coeff_upsilon}
\begin{split}
 \hat{\upsilon}_{0,j} & =  {\hat{\xi}_0}^j , \\
\hat{\upsilon}_{l,j} & =  \frac{1}{l \hat{\xi}_0} \sum_{k=1}^l \left( kj - l + k \right) \hat{\xi}_k \hat{\upsilon}_{l-k,j}, \quad  \text{for}  \quad l \geq 1,   
\end{split}
\end{align}
we get that for $n \geq 2$,
\begin{align}
\rho^{ \left( n \right)} \left( z_0; \beta \right)   =
   \sum_{l=0}^{\infty} \hat{\phi}_{l,n} \beta^l
 \label{eq:rho_taylor_diff_z0}
\end{align}
for $\beta > -1/3 - \left( 1- \pi \coth{\pi} \right) / \pi^2$ and $\left\vert \beta \right\vert$ sufficiently small, where 
\begin{align}
    \label{eq:coeff_phi}
\hat{\phi}_{l,n} =  \sum_{k=0}^l \hat{\varphi}_{n+l-k,n} \hat{\upsilon}_{k,l-k}  
\end{align}
for $ l \geq 0$.

With the completion of the foregoing, we are now ready to determine the paths of steepest descent through $z_0$ given by 
\begin{align*}
    \mathrm{Im} \left( \rho \left( z; \beta \right) \right) - \mathrm{Im} \left( \rho \left( z_0; \beta \right) \right)  = 0 \Longleftrightarrow 
    \mathrm{Im} \left( \rho \left( z; \beta \right) - \rho \left( z_0; \beta \right) \right) =0.
\end{align*}
The condition determining the paths of steepest descent just above, with an error $\mathrm{O} \left( \left( z -z_0 \right)^3 \right)$, is given by 
\begin{align}
    \mathrm{Im} \left( \frac{1}{2!} \rho^{\prime \prime} \left( z_0;\beta \right) \left( z-z_0 \right)^2    \right) = 0. \label{eq:steepest_descent}
\end{align}
Direct computation reveals that
\begin{align*}
    \rho^{\prime \prime} \left( z_0;\beta \right) < 0.
\end{align*}
If we set $z \coloneqq u + \mathrm{i} v$ for $ u , v \in \mathbb{R}$, then \cref{eq:steepest_descent} implies that the paths of steepest descent and ascent from $z_0$ lie along the curves
\begin{align*} 
      2 u \left( v - \mathrm{Im} \left( z_0 \right) \right) = 0
\end{align*}
as $z_0$ is purely imaginary. These paths, close enough to the saddle point $z_0$, that is when $\left\vert z- z_0 \right\vert$ is small, consist of the straight lines $u=0$ and $v = \mathrm{Im} \left( z_0 \right)$. To distinguish between the ascent and descent paths, we consider $\mathrm{Re} \left( \rho \left( z; \beta \right) \right)$ along the two lines near $z=z_0$. Along $u=0$, we have
\begin{align*}
    \mathrm{Re} \left( \rho \left( z ; \beta \right) \right) 
    = \mathrm{Re} \left(  \rho \left( z_0; \beta \right) \right) -\frac{1}{2!} \rho^{\prime \prime} \left( z_0; \beta \right)   \left( v - \mathrm{Im} \left( z_0 \right) \right)^2 + \mathrm{O} \left( \left\vert z- z_0 \right\vert^3 \right) 
     \geq  \mathrm{Re} \left(  \rho \left( z_0; \beta \right) \right)
\end{align*}
when $z$ is near $z_0$. Along $v = \mathrm{Im} \left( z_0 \right)$, we have
\begin{align*}
    \mathrm{Re} \left( \rho \left( z ; \beta \right) \right) = \mathrm{Re} \left(  \rho \left( z_0; \beta \right) \right) + \frac{1}{2!} \rho^{\prime \prime} \left( z_0; \beta \right) u^2 +\mathrm{O} \left( \left\vert z- z_0 \right\vert^3 \right)  
   \leq \mathrm{Re} \left(  \rho \left( z_0; \beta \right) \right)
\end{align*}
for $z$ close enough to $z_0$. Thus, the path of steepest descents from $z_0$ is $v = \mathrm{Im} \left( z_0 \right)$, parallel to the real axis. 

As $z_0$ is in the domain of analyticity of $\rho \left( z; \beta \right)$, we can deform the original contour of the integration \cref{eq:f_^P_integral} onto the steepest descent paths through the saddle point $z_0$, denoted by $\mathcal{C}_l$ for $u < 0$ and $\mathcal{C}_r$ for $u > 0$, both pointing a direction away from $z_0$. It follows from Cauchy's theorem that
\begin{align*}
      f_{Z^P} \left( x \right)  =  \frac{1}{4 \pi} \sqrt{\frac{2P}{45}}  \int_{-\infty}^{+ \infty} \exp{\left( P \rho \left( z; \beta \right) \right)} \, dz = 
    \frac{1}{4 \pi} \sqrt{\frac{2P}{45}} \int_{\mathcal{C}_r - \mathcal{C}_l} \exp{\left( P \rho \left( z; \beta \right) \right)} \, dz,
\end{align*}
whence the main contributions to the asymptotic expansion of the integral now comes from a small neighbourhood of $z_0$ for large $P$. We use Laplace's method to evaluate this integral. For some $\epsilon > 0$, we have the following asymptotic relation:
\begin{align}
    f_{Z^P} \left( x \right) \sim  \frac{1}{4 \pi} \sqrt{\frac{2P}{45}} \int_{\mathrm{Im}  \left( z_0 \right) \mathrm{i} - \epsilon}^{\mathrm{Im} \left( z_0 \right) \mathrm{i}+ \epsilon}  \exp{\left( P \rho \left( z; \beta \right) \right)} \, dz, \quad \text{as} \quad P \to + \infty,
    \label{eq:f_z_laplace}
\end{align}
where by replacing the contour of integration $\mathcal{C}_r - \mathcal{C}_l$ with a narrow interval centred around $z_0$, only exponentially small errors are introduced for large $P$. Now, $\epsilon$ can be chosen so small that we can replace $\rho \left( z; \beta \right)$ by its Taylor expansion \cref{eq:rho_taylor_z0}, which converges on the interval $\left(\mathrm{Im} \left( z_0 \right) \mathrm{i} - \epsilon, \mathrm{Im} \left( z_0 \right) \mathrm{i} +\epsilon \right)$. Then, separating the quadratic term from all the higher-order terms of the series expansion \cref{eq:rho_taylor_z0} in $\exp{ \left( P  \rho \left( z;\beta \right) \right)}$ and setting 
\begin{align}
\label{eq:g_function}
g \left( z;\beta \right)  \coloneqq  \exp{\left( P  \left( z-z_0 \right)^3 \sum_{k \geq 0} \frac{\rho^{ \left( k+3 \right)} \left( z_0; \beta \right)}{ \left( k+3 \right)!} \left( z-z_0 \right)^{k} \right)},
\end{align}
the integral \cref{eq:f_z_laplace} becomes
\begin{align}
\label{eq:f_Z^P_series_z0}
     f_{Z^P} \left( x \right) \sim  \frac{1}{4 \pi} \sqrt{\frac{2P}{45}} \exp{\left( P  \rho \left( z_0; \beta \right) \right)} \int_{\mathrm{Im}  \left( z_0 \right) \mathrm{i} - \epsilon}^{\mathrm{Im} \left( z_0 \right) \mathrm{i}+ \epsilon}  \exp{\left( P  \frac{1}{2!} \rho^{\prime \prime} \left( z_0; \beta \right) \left( z-z_0 \right)^2 \right)} g \left( z; \beta \right) \, dz,
\end{align}
as $P \to +\infty$.

To find $\rho \left( z_0; \beta \right)$, we use \cref{eq:rho_taylor_0}, \cref{eq:saddle_point} and \cref{eq:z0_product} to write
\begin{align*}
 \rho \left( z_0; \beta \right)  = \frac{1}{2} \beta \mathrm{i} z_0 + \sum_{k=2}^{\infty} \hat{r}_k z_0^k 
      = \sum_{k=2}^\infty \hat{\rho}_k \beta^k,
\end{align*}
where for $k \geq 2$, 
\begin{align}
\label{eq:coeff_rho}
    \hat{\rho}_k \coloneqq \frac{\mathrm{i} \hat{\xi}_{k-2}}{2} + \sum_{m=2}^k \hat{r}_m \hat{\upsilon}_{k-m,m}.  
\end{align}

Since the series in the argument of the exponential function which defines $g \left( z; \beta \right)$ in \cref{eq:g_function} is convergent near $z_0$, we can write as $z \to z_0$,
\begin{align}
    g \left( z; \beta \right) = &   \exp{\left( P \left( z-z_0 \right)^3 \sum_{k \geq 0} \hat{\sigma}_k \left( \beta \right) \left( z-z_0 \right)^k   \right)} \nonumber \\ 
        \label{eq:exp_high}
     =  &   \sum_{n=0}^{\infty} \frac{1}{n!} P^n \left( z-z_0 \right)^{3n} \left( \sum_{k \geq 0} \hat{\sigma}_k \left( \beta \right) \left( z-z_0 \right)^k \right)^n,
\end{align}
where  $\hat{\sigma}_k \left( \beta \right) \coloneqq {\rho^{ \left( k+3 \right)} \left( z_0; \beta \right)}/{ \left( k+3 \right)!} $ for $k \geq 0$. Further, the Taylor series expansion \cref{eq:rho_taylor_diff_z0} for $ \rho^{ \left( k+3 \right)} \left( z_0; \beta \right)$ gives us $\hat{\sigma}_k \left( \beta \right)  = \sum_{l=0}^\infty \hat{\gamma}_{l,k}  \beta^l$ with
\begin{align}
    \label{eq:coeff_gamma}
\hat{\gamma}_{l,k} \coloneqq \frac{ \hat{\phi}_{l,k+3}}{{ \left( k+3 \right)!}}
\end{align}
for $l,k \geq 0$ when $\beta > -1/3 - \left( 1- \pi \coth{\pi} \right) / \pi^2$ and $\left\vert \beta \right\vert$ is sufficiently small. As an immediate consequence of the properties for Taylor series expansions, we have for $n \geq 2$, 
\begin{align*} 
\left( \sum_{k \geq 0} \hat{\sigma}_k \left( \beta \right) \left( z-z_0 \right)^k \right)^n = \sum_{k_1=0}^{\infty} \sum_{k_2=0}^{k_1}  \cdots \sum_{k_n=0}^{k_{n-1}}  \hat{\sigma}_{k_n} \left( \beta \right)  \hat{\sigma}_{k_{n-1} - k_n} \left( \beta \right) \cdots \hat{\sigma}_{k_{1} - k_2} \left( \beta \right)  \left( z-z_0 \right)^{k_1},
\end{align*}
as $z \to z_0$. In addition, we observe for $n \geq 2$ and $ 0 \leq k_n \leq k_{n-1} \leq \cdots \leq k_1$,
\begingroup
\allowdisplaybreaks
\begin{align*}
\lefteqn{\hat{\sigma}_{k_n} \left( \beta \right) \hat{\sigma}_{k_{n-1} - k_n} \left( \beta \right) \cdots \hat{\sigma}_{k_{1} - k_2} \left( \beta \right) \nonumber } \\
 = &   \left( \sum_{l_1=0}^{\infty} \hat{\gamma}_{l_1, k_n} \beta^{l_1} \right)     \left( \sum_{l_2=0}^{\infty} \hat{\gamma}_{l_2, k_{n-1}-k_n} \beta^{l_2} \right) \cdots  \left( \sum_{l_n=0}^{\infty} \hat{\gamma}_{l_n, k_1-k_2} \beta^{l_n} \right) \nonumber \nonumber \\
 = &  \sum_{l_1=0}^{\infty} \sum_{l_2=0}^{l_1} \cdots \sum_{l_n=0}^{l_{n-1}}  \hat{\gamma}_{l_n, k_n} \hat{\gamma}_{l_{n-1}-l_n, k_{n-1}-k_n} \cdots \hat{\gamma}_{l_1-l_2, k_1-k_2} \beta^{l_1} \nonumber \\
 = &  \sum_{l_1=0}^{\infty} \hat{\mathcal{C}}_{l_1, k_1, k_2, \cdots, k_n} \beta^{l_1}
\end{align*}
\endgroup
for $\beta > -1/3 - \left( 1- \pi \coth{\pi} \right) / \pi^2$ and $\left\vert \beta \right\vert$ sufficiently small, where 
\begin{align}
    \label{eq:coeff_C}
\hat{\mathcal{C}}_{l_1, k_1, k_2, \cdots, k_n} \coloneqq \sum_{l_2=0}^{l_1} \cdots \sum_{l_n=0}^{l_{n-1}}  \hat{\gamma}_{l_n, k_n} \hat{\gamma}_{l_{n-1}-l_n, k_{n-1}-k_n} \cdots \hat{\gamma}_{l_1-l_2, k_1-k_2}
\end{align}
for $l_1 \geq 0$. Generally, for $n \geq 0$ we see that
\begin{align*}
    \left( \sum_{k \geq 0} \hat{\sigma}_k \left( \beta \right) \left( z-z_0 \right)^k   \right)^n   =  \sum_{k = 0}^{\infty} \hat{\theta}_{k,n} \left( \beta \right) \left( z-z_0 \right)^k, \quad \text{as} \quad z \to z_0.
\end{align*}
Here $\hat{\theta}_{k,n} (\beta)$ are functions of $\beta$ satisfying $\hat{\theta}_{k,n} (\beta) = \sum_{l=0}^{\infty} \hat{\mathcal{E}}_{l,k,n} \beta^{l}$ for $k \geq0$ with the constants $\hat{\mathcal{E}}_{l,k,n}$ as stated below:
for $n=0$,
\begin{align}
\label{eq:coeff_E1}
   \hat{\mathcal{E}}_{l,k,n} =  \hat{\mathcal{E}}_{l,k,0} = \begin{cases}
1, \quad & \text{for} \quad k=l=0, \\
0, \quad & \text{otherwise, } 
\end{cases}
\end{align}
for $n=1$,
\begin{align}
\label{eq:coeff_E2}
     \hat{\mathcal{E}}_{l,k,n} =  \hat{\mathcal{E}}_{l,k,1} = \hat{\gamma}_{l,k}, \quad \text{for} \quad k, l \geq 0,
\end{align}
for $n=2$,
\begin{align}
\label{eq:coeff_E3}
    \hat{\mathcal{E}}_{l,k,n} =  \hat{\mathcal{E}}_{l,k,2}= \sum_{k_2=0}^{k}  \hat{\mathcal{C}}_{l, k, k_2}, \quad \text{for} \quad k,l \geq 0,
\end{align}
for $n \geq 3$,
\begin{align}
\label{eq:coeff_E4}
   \hat{\mathcal{E}}_{l,k,n} = \sum_{k_2=0}^{k}  \sum_{k_3=0}^{k_2} \cdots \sum_{k_n=0}^{k_{n-1}} \hat{\mathcal{C}}_{l, k, k_2, \cdots, k_n}, \quad  \text{for} \quad k, l \geq 0.
\end{align}
Using these factors, we can rewrite $g \left( z; \beta \right)$ in \cref{eq:exp_high} as 
\begingroup
\allowdisplaybreaks
\begin{align*}
g \left( z;\beta \right) &  = 
 \sum_{n=0}^{\infty} \frac{1}{n!} P^n \left( z-z_0 \right)^{3n}  \left( \sum_{k=0}^{\infty} \hat{\theta}_{k,n} \left( \beta \right) \left( z-z_0 \right)^{k}  \right)  \\
&  =  \sum_{n=0}^{\infty} \sum_{k=0}^{\infty} \frac{1}{n!} P^n \hat{\theta}_{k,n} \left( \beta \right) \left( z-z_0 \right)^{3n+k}     \\
& =  \sum_{j=0}^{\infty} \sum_{3n+k=j} \frac{1}{n!} P^n \hat{\theta}_{k,n} \left( \beta \right)  \left( z-z_0 \right)^j  \\
&  = \sum_{j=0}^{\infty} \sum_{n=0}^{\lfloor \frac{j}{3} \rfloor} \frac{1}{n!} P^n \hat{\theta}_{
j-3n,n} \left( \beta \right) \left( z-z_0 \right)^j  \\
&  = \sum_{j=0}^{\infty} \hat{g}_j \left( \beta \right) \left( z-z_0 \right)^j
\end{align*}
\endgroup
in the limit $z \to z_0$, where $ \hat{g}_j \left( \beta \right) 
\coloneqq \sum_{n=0}^{\lfloor j/3 \rfloor}  P^n \hat{\theta}_{j-3n,n} \left( \beta \right)/n!$ for $j \geq 0  $. Hence, we have 
\begin{align*}
    g \left( z; \beta \right) - \sum_{j=0}^J \hat{g}_j \left( \beta \right) \left( z-z_0 \right)^j = \mathrm{o} \left( \left( z - z_0 \right)^J \right), \quad \text{as} \quad z \to z_0 
\end{align*}
for any $J \geq 0$. From this it follows that for any $\epsilon^* > 0$ there is an interval $\left\vert z-z_0 \right\vert < L$ for some $L >0$, in which 
\begin{align*}
    \left\vert  g \left( z; \beta \right) - \sum_{j=0}^J \hat{g}_j \left( \beta \right) \left( z-z_0 \right)^j \right\vert \leq \epsilon^* \left\vert  \left( z - z_0 \right)^J \right\vert.
\end{align*}
Therefore for any $0 < \epsilon < L $, we have
\begin{align*}
\lefteqn{\left\vert  \int_{\mathrm{Im}  \left( z_0 \right) \mathrm{i} - \epsilon}^{\mathrm{Im}  \left( z_0 \right) \mathrm{i} }  \exp{\left( P  \frac{1}{2} \rho^{\prime \prime} \left( z_0; \beta \right) \left( z-z_0 \right)^2 \right)} \left(   g \left( z; \beta \right) - \sum_{j=0}^J \hat{g}_j \left( \beta \right) \left( z-z_0 \right)^j  \right) \, dz   \right\vert} \\
  \leq & \int_{\mathrm{Im}  \left( z_0 \right) \mathrm{i} - \epsilon}^{\mathrm{Im}  \left( z_0 \right) \mathrm{i} }  \exp{\left( P  \frac{1}{2} \rho^{\prime \prime} \left( z_0; \beta \right) \left( z-z_0 \right)^2 \right)} \left\vert g \left( z; \beta \right) - \sum_{j=0}^J \hat{g}_j \left( \beta \right) \left( z-z_0 \right)^j \right\vert \, dz \\
  \leq & \epsilon^* \int_{\mathrm{Im}  \left( z_0 \right) \mathrm{i} - \epsilon}^{\mathrm{Im}  \left( z_0 \right) \mathrm{i} }  \exp{\left( P  \frac{1}{2} \rho^{\prime \prime} \left( z_0; \beta \right) \left( z-z_0 \right)^2 \right)} \left\vert  \left( z - z_0 \right)^J \right\vert \, dz \\
  = & \epsilon^* \left( -1 \right)^J  \int_{\mathrm{Im}  \left( z_0 \right) \mathrm{i} - \epsilon}^{\mathrm{Im}  \left( z_0 \right) \mathrm{i} }  \exp{\left( P  \frac{1}{2} \rho^{\prime \prime} \left( z_0; \beta \right) \left( z-z_0 \right)^2 \right)} \left(z -z_0 \right)^J \, dz.
\end{align*}
Then as $\epsilon \to 0^+$, we can write
\begin{align*}
\lefteqn{ \int_{\mathrm{Im}  \left( z_0 \right) \mathrm{i} - \epsilon}^{\mathrm{Im}  \left( z_0 \right) \mathrm{i} }  \exp{\left( P  \frac{1}{2} \rho^{\prime \prime} \left( z_0; \beta \right) \left( z-z_0 \right)^2 \right)} g\left( z ; \beta \right) \, dz } \\
= & \sum_{j=0}^J \hat{g}_j \left( \beta \right) \int_{\mathrm{Im}  \left( z_0 \right) \mathrm{i} - \epsilon}^{\mathrm{Im}  \left( z_0 \right) \mathrm{i} } \exp{\left( P  \frac{1}{2} \rho^{\prime \prime} \left( z_0; \beta \right) \left( z-z_0 \right)^2 \right)}  \left( z -z_0 \right)^j \, dz  \\
    & + \mathrm{o} \left( \int_{\mathrm{Im}  \left( z_0 \right) \mathrm{i} - \epsilon}^{\mathrm{Im}  \left( z_0 \right) \mathrm{i} } \exp{\left( P  \frac{1}{2} \rho^{\prime \prime} \left( z_0; \beta \right) \left( z-z_0 \right)^2 \right)}  \left( z -z_0 \right)^J \, dz  \right),
\end{align*}
which gives
\begingroup
\begin{align*}
\lefteqn{    \int_{\mathrm{Im}  \left( z_0 \right) \mathrm{i} - \epsilon}^{\mathrm{Im}  \left( z_0 \right) \mathrm{i} }  \exp{\left( P  \frac{1}{2} \rho^{\prime \prime} \left( z_0; \beta \right) \left( z-z_0 \right)^2 \right)} g\left( z ; \beta \right) \, dz} \\
   \sim & \sum_{j=0}^\infty \hat{g}_j \left( \beta \right) \int_{\mathrm{Im}  \left( z_0 \right) \mathrm{i} - \epsilon}^{\mathrm{Im}  \left( z_0 \right) \mathrm{i} } \exp{\left( P  \frac{1}{2} \rho^{\prime \prime} \left( z_0; \beta \right) \left( z-z_0 \right)^2 \right)}  \left( z -z_0 \right)^j \, dz
\end{align*}
\endgroup
for small $\epsilon$. Now the above integrals can be evaluated by change of variables. For arbitrary $j \geq 0$, the substitution $z= \mathrm{Im} \left( z_0 \right) \mathrm{i} + x$ yields 
\begingroup
\allowdisplaybreaks
\begin{align*}
    & \int_{\mathrm{Im} \left( z_0 \right) \mathrm{i} - \epsilon}^{\mathrm{Im} \left( z_0 \right) \mathrm{i} }   \exp{\left( P \frac{1}{2} \rho^{\prime \prime} \left( z_0; \beta \right) \left( z-z_0 \right)^2 \right)} \left( z-z_0 \right)^{ j} \, dz\\  
= &  \int_{-\epsilon}^{0} \exp{ \left( P \frac{1}{2} \rho^{\prime \prime} \left( z_0; \beta \right) x^2    \right)} x^j \, dx \\
= & \frac{1}{2} \left( -1 \right)^j  \left( \frac{-2}{\rho^{\prime \prime} \left( z_0;\beta \right)} \right)^{\frac{1}{2} \left( j+1 \right)}  \cdot \int_0^{-\frac{1}{2}  \rho^{\prime \prime} \left( z_0; \beta \right) \epsilon^2 } \zeta^{\frac{1}{2} \left( j-1 \right)}   \exp{\left( -P \zeta \right)}  \, d \zeta,  
\end{align*}
\endgroup
where the last step is a result of the change of variable $\rho^{\prime \prime} \left( z_0; \beta \right) x^2 /2 = - \zeta$. Thus as $\epsilon \to 0^+$, we have
\begin{align*}
\lefteqn{  \int_{\mathrm{Im}  \left( z_0 \right) \mathrm{i} - \epsilon}^{\mathrm{Im}  \left( z_0 \right) \mathrm{i} }  \exp{\left( P  \frac{1}{2} \rho^{\prime \prime} \left( z_0; \beta \right) \left( z-z_0 \right)^2 \right)} g\left( z ; \beta \right) \, dz} \\
\sim & \sum_{j=0}^\infty \hat{g}_j \left( \beta \right) \frac{1}{2} \left( -1 \right)^j  \left( \frac{-2}{\rho^{\prime \prime} \left( z_0;\beta \right)} \right)^{\frac{1}{2} \left( j+1 \right)} \int_0^{-\frac{1}{2}  \rho^{\prime \prime} \left( z_0; \beta \right) \epsilon^2 } \zeta^{\frac{1}{2} \left( j-1 \right)}   \exp{\left( -P \zeta \right)}  \, d \zeta.
\end{align*}
Similar arguments give us that
\begin{align*}
\lefteqn{ \int_{\mathrm{Im}  \left( z_0 \right) \mathrm{i}}^{\mathrm{Im}  \left( z_0 \right) \mathrm{i} + \epsilon }  \exp{\left( P  \frac{1}{2} \rho^{\prime \prime} \left( z_0; \beta \right) \left( z-z_0 \right)^2 \right)} g\left( z ; \beta \right) \, dz} \\
\sim & \sum_{j=0}^\infty \hat{g}_j \left( \beta \right) \frac{1}{2}   \left( \frac{-2}{\rho^{\prime \prime} \left( z_0;\beta \right)} \right)^{\frac{1}{2} \left( j+1 \right)} \int_0^{-\frac{1}{2}  \rho^{\prime \prime} \left( z_0; \beta \right) \epsilon^2 } \zeta^{\frac{1}{2} \left( j-1 \right)}   \exp{\left( -P \zeta \right)}  \, d \zeta,
\end{align*}
as $\epsilon \to 0^+$. 

Hence, the integration in \cref{eq:f_Z^P_series_z0} can be expanded in an asymptotic series for small $\epsilon$ as follows:
\begin{align*}
\lefteqn{  \int_{\mathrm{Im}  \left( z_0 \right) \mathrm{i} - \epsilon}^{\mathrm{Im} \left( z_0 \right) \mathrm{i}+ \epsilon}  \exp{\left( P  \frac{1}{2!} \rho^{\prime \prime} \left( z_0; \beta \right) \left( z-z_0 \right)^2 \right)} g \left( z; \beta \right) \, dz} \\
\sim & \sum_{j=0}^\infty    \frac{1}{2}  \left( 1 + \left( -1 \right)^j \right)  \hat{g}_j \left( \beta \right) \left( \frac{-2}{\rho^{\prime \prime} \left( z_0;\beta \right)} \right)^{\frac{1}{2} \left( j+1 \right)} \int_0^{-\frac{1}{2}  \rho^{\prime \prime} \left( z_0; \beta \right) \epsilon^2} \zeta^{\frac{1}{2} \left( j-1 \right)}   \exp{\left( -P \zeta \right)}  \, d \zeta,
\end{align*}
where terms with odd $j$ vanish. For large $P$, we can extend the integration region in each integral to infinity.  With this replacement, we introduce only exponentially small errors for large $P$, whence we have as $P \to + \infty$,
\begin{align*}
 \int_0^{-\frac{1}{2}  \rho^{\prime \prime} \left( z_0; \beta \right) \epsilon^2} \zeta^{\frac{1}{2} \left( j-1 \right)}   \exp{\left( -P \zeta \right)}  \, d \zeta 
\sim & \int_0^{+ \infty} \zeta^{\frac{1}{2} \left( j-1 \right)}   \exp{\left( -P \zeta \right)}  \, d \zeta \\
= & P^{- \frac{1}{2} \left( j+{1} \right)} \Gamma \left(\frac{j}{2} +\frac{1}{2} \right)
\end{align*}
for $j \geq 0$. Assembling the above results, we have the following asymptotic series
\begin{align}
f_{Z^P} \left( x \right) \sim \frac{1}{4 \pi} \sqrt{\frac{2 P}{45}} \exp{\left( P \sum_{l=2}^{\infty} \hat{\rho}_l \beta^l \right)} \sum_{j=0}^{\infty} \hat{g}_{2  j} \left( \beta \right) \left( \frac{-2}{\rho^{\prime \prime} \left( z_0;\beta \right)} \right)^{j+\frac{1}{2}} P^{- \left( j+\frac{1}{2} \right)} \Gamma \left(j+\frac{1}{2} \right) \label{eq:f_Z_cha_var}
\end{align}
in the limit $P \to + \infty$ for fixed $x$ with $x > - \sqrt{5P/2} \left( 1+ 3 / \pi^2 -3 \coth{\pi} / \pi \right)$ and $\left\vert x \right\vert$ sufficiently small.

Lastly, we wish to express the terms involving $\beta$, i.e. $ \hat{g}_{2  j} \left( \beta \right)   \left( -2/\rho^{\prime \prime} \left( z_0;\beta \right) \right)^{j+1/2}$ as 
a Taylor series expansion in $\beta$. This can be achieved by collecting the coefficients from the product of their individual series. Assume that
\begin{align}
\sqrt{\frac{-2}{\rho^{\prime \prime} \left( z_0; \beta \right)}} = \sum_{n=0}^{\infty} \hat{K}_{n} \beta^n
\label{eq:power_beta_K}
\end{align}
for $\left\vert \beta \right\vert$ sufficiently small and some constants $\hat{K}_{n}$ for $ n \geq 0$.
Then taking the square on both sides yields
\begin{align*}
    \frac{-2}{\rho^{\prime \prime} \left( z_0; \beta \right)} 
    =  \sum_{l=0}^\infty \sum_{k=0}^l \hat{K}_k \hat{K}_{l-k} \beta^l =  \sum_{l=0}^\infty \hat{\mu}_l \beta^l,
\end{align*}
where $\hat{\mu}_{l} \coloneqq \sum_{k=0}^l \hat{K}_{k} \hat{K}_{l-k}$ for $l \geq 0$. On the other hand, by performing simple arithmetical operations on the Taylor series \cref{eq:rho_taylor_diff_z0} with $n=2$  for $\rho^{\prime \prime} \left( z_0; \beta \right)$, we see
\begin{align*}
     \frac{-2}{\rho^{\prime \prime} \left( z_0; \beta \right)} =  \sum_{l=0}^\infty \hat{\varpi}_l \beta^l
\end{align*}
for  $\beta > -1/3 - \left( 1- \pi \coth{\pi} \right) / \pi^2$ and $\left\vert \beta \right\vert$ sufficiently small, where 
\begin{align}
    \label{eq:coeff_varpi}
\begin{split}
\hat{\varpi}_0  & =  - \frac{2}{\hat{\phi}_{0,2}}, \\
\hat{\varpi}_l  &=  - \frac{1}{{\hat{\phi}_{0,2}}} \sum_{m=0}^{l-1} \hat{\varpi}_m \hat{\phi}_{l-m,2}, \quad \text{for} \quad l \geq 1.    
\end{split}
\end{align}
Hence, by equating the coefficients, we find 
\begin{align}
\label{eq:coeff_K}
    \hat{\varpi}_l = \hat{\mu}_l = \sum_{k=0}^{l} \hat{K}_k \hat{K}_{l-k}, \quad \text{for} \quad l \geq 0,
\end{align}
providing the values for the constants $\hat{K}_k$ with $k \geq 0$.

Based on the previous analysis, we are now ready to derive the Taylor series expansion of $ \left( -2/\rho^{\prime \prime} \left( z_0;\beta \right) \right)^{j+1/2}$ for $\beta > -1/3 - \left( 1- \pi \coth{\pi} \right) / \pi^2$ and $\left\vert \beta \right\vert$ sufficiently small. Indeed, from \cref{eq:power_beta_K} we have for $j \geq 0$,
\begin{align*}
 \left( \frac{-2}{\rho^{\prime \prime} \left( z_0;\beta \right)} \right)^{j+\frac{1}{2}} 
= \left( \sum_{n=0}^{\infty} \hat{K}_{n} \beta^n \right)^{2j+1} = \sum_{n=0}^{\infty} \hat{\omega}_{n,j} \beta^n,
\end{align*} 
where for $j=0$,
\begin{align}
\label{eq:coeff_omega1}
    \hat{\omega}_{n,j} = \hat{\omega}_{n,0} = \hat{K}_n, \quad \text{for} \quad n \geq 0,
\end{align}
for $j \geq 1$,
\begin{align}
\label{eq:coeff_omega2}
    \hat{\omega}_{n,j} = \sum_{n_2 =0}^{n} \sum_{n_3=0}^{n_2} \cdots \sum_{n_{2j+1}=0}^{n_{2j}}   \hat{K}_{n_{2j+1}} \hat{K}_{n_{2j}-n_{2j+1}} \cdots \hat{K}_{n_2-n_3} \hat{K}_{n-n_2}, \quad \text{for} \quad n \geq 0.
\end{align}
If we combine the series for $\hat{g}_{2j} \left( \beta \right)$ with the explicit series expansion given above, we obtain for $j \geq 0$,
\begingroup
\allowdisplaybreaks
\begin{align*}
 \hat{g}_{2j} \left( \beta \right) \left(  \frac{-2}{\rho^{\prime \prime} \left( z_0; \beta \right)  } \right)^{j+\frac{1}{2}} 
 = & \left( \sum_{n=0}^{\lfloor \frac{2}{3} j \rfloor} \frac{1}{n!} P^n \sum_{l=0}^{\infty} \hat{\mathcal{E}}_{l,2j-3n,n} \beta^l \right) \left( \sum_{n=0}^{\infty} \hat{\omega}_{n,j}  \beta^n \right) \\
  = & \left( \sum_{l=0}^{\infty} \left( \sum_{n=0}^{\lfloor \frac{2}{3} j \rfloor} \frac{1}{n!} P^n \hat{\mathcal{E}}_{l,2j-3n,n} \right) \beta^l \right) \left(  \sum_{k=0}^{\infty} \hat{\omega}_{k,j}  \beta^k \right) \\
  = & \sum_{l=0}^{\infty} \left( \sum_{k=0}^l \left( \sum_{n=0}^{\lfloor \frac{2}{3} j \rfloor} \frac{1}{n!} P^n \hat{\mathcal{E}}_{k,2j-3n,n} \right) \hat{\omega}_{l-k,j} \right) \beta^l \\
 = & \sum_{l=0}^{\infty}  \sum_{n=0}^{\lfloor \frac{2}{3} j \rfloor}  \sum_{k=0}^l \frac{1}{n!}   \hat{\omega}_{l-k, j}  \hat{\mathcal{E}}_{k,2j-3n,n}  P^n \beta^l \\
=  & \sum_{l=0}^{\infty} \sum_{n=0}^{\lfloor \frac{2}{3} j \rfloor} \hat{\alpha}_{n,l,j} P^n \beta^l,
\end{align*}
\endgroup
where
\begin{align}
\label{eq:coeff_alpha}
   \hat{\alpha}_{n,l,j} \coloneqq \frac{1}{n!} \sum_{k=0}^l  \hat{\omega}_{l-k, j} \hat{\mathcal{E}}_{k,2j-3n,n}
\end{align}
for $ 0 \leq n \leq \lfloor 2j /3 \rfloor$ and $l  \geq 0$. Then \cref{eq:f_Z_cha_var} becomes
\begin{align}
f_{Z^P} \left( x \right) & \sim \frac{1}{4 \pi} \sqrt{\frac{2 P}{45}} \exp{ \left( P \sum_{l=2}^{\infty} \hat{\rho}_l \beta^l \right)} \sum_{j=0}^{\infty}  \sum_{l=0}^{\infty} \sum_{n=0}^{\lfloor \frac{2}{3} j \rfloor} \hat{\alpha}_{n,l,j} P^n \beta^l  P^{- \left( j+\frac{1}{2} \right)} \Gamma \left( j +\frac{1}{2} \right), \label{eq:f_Z_asymp_beta}
\end{align}
as $P \to +\infty$ for  $\beta > -1/3 - \left( 1- \pi \coth{\pi} \right) / \pi^2$ and $\left\vert \beta \right\vert$ sufficiently small, which completes the proof.

\section{Proof of \cref{thm:Asymp_Exp_CDF_Z^P}}
\label{sec:appendix_proof_3.3}
Before integrating the density function, we first rewrite its asymptotic expansion in \cref{thm:Asymp_Exp_pdf_Z^P} in terms of the original variable $x$ by using the identity $\beta = x \sqrt{2/45} / \sqrt{P}$. Accordingly in the limit $P \to +\infty$ with $x > - \sqrt{5P/2} \left( 1+ 3 / \pi^2 -3 \coth{\pi} / \pi \right)$ and $\left\vert x \right\vert$ sufficiently small, the probability density function has the following asymptotic expansion
\begin{align}
\label{eq:f_Z_asymp}
f_{Z^P} \left( x \right) 
  \sim & \frac{1}{4 \pi} \sqrt{\frac{2}{45}} \exp{ \left( \sum_{l=2}^{\infty} \hat{\rho}_l \left( \frac{2}{45} \right)^{\frac{l}{2}}  P^{1- \frac{l}{2}} x^l  \right)} \\
   & \cdot \sum_{j=0}^{\infty} \sum_{l=0}^{\infty} \sum_{n=0}^{\lfloor \frac{2}{3} j \rfloor}  \hat{\alpha}_{n,l,j} \left( \frac{2}{45} \right)^{\frac{l}{2}} \Gamma \left( j+\frac{1}{2} \right) P^{n-\frac{l}{2}-j} x^l. \nonumber  
\end{align}
To justify that the integrated series is indeed asymptotic to the distribution function, we adjust the terms in \cref{eq:f_Z_asymp} to form a more appropriate expression for easier computation. 

Specifically, we separate the quadratic term $\hat{\rho}_2 \left( 2/45 \right) x^2$ from the argument $\sum_{l=2}^{\infty}$ $\hat{\rho}_l \left( 2/45 \right)^{{l}/{2}}  P^{1- l/2}  x^l$ of the exponential function. As the integration is taken with respect to $x$, we expand the remaining term in an asymptotic approximation in $P$ with all the coefficients given as polynomials of $x$. Note that $\hat{\rho}_2 =  \mathrm{i}  \hat{\xi}_0 /2 + \hat{r}_2    ( \hat{\xi}_0 )^2   =- {45}/{4}$, which gives
\begin{align}
  \exp{ \left( \sum_{l=2}^{\infty} \hat{\rho}_l \left( \frac{2}{45} \right)^{\frac{l}{2}}   P^{1- \frac{l}{2}} x^l  \right)}  
= \exp{\left( - \frac{1}{2} x^2 \right)} \exp{ \left( \sum_{l=3}^{\infty} \hat{\rho}_l \left( \frac{2}{45} \right)^{\frac{l}{2}}   P^{1- \frac{l}{2}} x^l  \right)}. \label{eq:pre_exp}
\end{align}
Now note, as $P \to + \infty$, we have
\begin{align}
\exp{ \left( \sum_{l=3}^{\infty} \hat{\rho}_l \left( \frac{2}{45} \right)^{\frac{l}{2}}   P^{1- \frac{l}{2}} x^l  \right)} \nonumber 
\sim &\sum_{n=0}^{\infty} \frac{1}{n!} \left( \sum_{l=3}^{\infty} \hat{\rho}_l \left( \frac{2}{45} \right)^{\frac{l}{2}}  P^{1- \frac{l}{2}}  x^l  \right)^n \nonumber \\
\sim & \sum_{n=0}^{\infty} \frac{1}{n!} \left( \frac{2}{45} \right)^{\frac{3n}{2}}  P^{-\frac{n}{2}} x^{3n} \left( \sum_{k=0}^{\infty} \hat{\rho}_{k+3} \left( \frac{2}{45} \right)^{\frac{k}{2}}  P^{- \frac{k}{2}}  x^{k}  \right)^n. \label{pre_erp_high}
\end{align}
Analogous to the previous computations, the generalisation of multiplication of asymptotic expansions tells us for $n \geq 0$,
\begin{align}
    \left( \sum_{k=0}^{\infty} \hat{\rho}_{k+3} \left( \frac{2}{45} \right)^{\frac{k}{2}}  P^{- \frac{k}{2}}  x^{k} \right)^n \sim \sum_{k=0}^{\infty} \hat{\vartheta}_{k,n} \left( \frac{2}{45}\right)^{\frac{k}{2}} P^{- \frac{k}{2}}  x^k , \label{eq:pre_exp_high_indeptvar_series}
\end{align}
where
\begin{align}
\label{eq:coeff_vartheta}
\begin{split}
    \hat{\vartheta}_{0,n} & = \left( \hat{\rho}_3 \right)^n, \\
    \hat{\vartheta}_{k,n} & = \frac{1}{k \hat{\rho}_3} \sum_{m=1}^{k} \left( m n - k + m \right) \hat{\rho}_{k+3} \hat{\vartheta}_{k-m,n}, \quad \text{for} \quad k \geq 1.    
\end{split}
\end{align}
Using \cref{eq:pre_exp}$-$\cref{eq:pre_exp_high_indeptvar_series} amounts to
\begin{align*}
 \exp{ \left( \sum_{l=2}^{\infty} \hat{\rho}_l \left( \frac{2}{45} \right)^{\frac{l}{2}}  P^{1- \frac{l}{2}}  x^l \right)} 
 \sim & \exp{\left( - \frac{1}{2} x^2 \right)} \sum_{n=0}^{\infty} \frac{1}{n!}  \left( \frac{2}{45} \right)^{\frac{3n}{2}}  P^{-\frac{n}{2}} x^{3n} \\
 & \cdot \sum_{k=0}^{\infty} \hat{\vartheta}_{k,n} \left( \frac{2}{45} \right)^{\frac{k}{2}}  P^{-\frac{k}{2}} x^{k} \\
\sim & \exp{\left( - \frac{1}{2} x^2 \right)}  \sum_{n=0}^{\infty} \sum_{k=0}^{\infty} \frac{1}{n!}  \left(  \frac{2}{45}\right)^{\frac{3n+k}{2}} \hat{\vartheta}_{k,n} P^{-\frac{n+k}{2}}  x^{3n+k} \\
\sim & \exp{\left( - \frac{1}{2} x^2 \right)}  \sum_{j=0}^{\infty} \left( \sum_{n=0}^{j} \frac{1}{n!} \left( \frac{2}{45} \right)^{\frac{2n+j}{2}} \hat{\vartheta}_{j-n,n} x^{2n+j} \right) P^{-\frac{j}{2}} \\
\sim & \exp{\left( - \frac{1}{2} x^2 \right)}  \sum_{j=0}^{\infty} \hat{A}_j \left( x \right) P^{-\frac{j}{2}},
\end{align*}
as $P \to + \infty$, where 
\begin{align*}
    \hat{A}_j \left( x \right) \coloneqq \sum_{n=0}^{j} \frac{1}{n!} \left( \frac{2}{45} \right)^{\frac{2n+j}{2}} \hat{\vartheta}_{j-n,n} x^{2n+j} = \sum_{n=0}^j \hat{\eta}_{n,j} x^{2n+j}, \quad \text{for} \quad j \geq 0
\end{align*}
with 
\begin{align}
    \label{eq:coeff_eta}
\hat{\eta}_{n,j} \coloneqq \left( \frac{2}{45} \right)^{n+ \frac{1}{2}j} \frac{\hat{\vartheta}_{j-n,n}}{n!}   
\end{align}
for $0 \leq n \leq j$. Further, we see that
\begin{align*}
\lefteqn{   \sum_{j=0}^{\infty} \sum_{l=0}^{\infty} \sum_{n=0}^{\lfloor \frac{2}{3} j \rfloor}  \hat{\alpha}_{n,l,j} \left( \frac{2}{45} \right)^{\frac{l}{2}} \Gamma \left( j+\frac{1}{2} \right)  x^l P^{n-\frac{l}{2}-j}} \\
\sim & \sum_{m=0}^\infty \sum_{l=0}^\infty \sum_{j=m}^{3m}  \hat{\alpha}_{j-m,l,j} \left( \frac{2}{45} \right)^{\frac{l}{2}} \Gamma \left( j+\frac{1}{2} \right)  x^l P^{- \left( m + \frac{l}{2} \right)} \\
\sim & \sum_{\substack{
         r = 0\\
         r \text{ even}}}^\infty
        \sum_{\substack{
         l = 0\\
         l \text{ even}}}^r
         \sum_{j= \frac{r- l }{2} }^{\frac{3 \left( r- l \right)}{2}} \hat{\alpha}_{j-\frac{{ r-l }}{2} ,l,j} \left( \frac{2}{45} \right)^{\frac{l}{2}} \Gamma \left( j+\frac{1}{2} \right)  x^l P^{- \frac{r}{2} } \\
& + \sum_{\substack{
         r = 1\\
         r \text{ odd}}}^\infty
        \sum_{\substack{
         l = 1\\
         l \text{ odd}}}^r
         \sum_{j= \frac{r- l }{2} }^{\frac{3 \left( r- l \right)}{2}} \hat{\alpha}_{j-\frac{{ r-l }}{2} ,l,j} \left( \frac{2}{45} \right)^{\frac{l}{2}} \Gamma \left( j+\frac{1}{2} \right)  x^l P^{- \frac{r}{2} } \\
\sim & \sum_{r=0}^\infty \hat{B}_r \left( x \right) P^{-\frac{r}{2}}  
\end{align*}
when $P \to + \infty$ with
\begin{align*}
     \hat{B}_r \left( x \right) \coloneqq \sum_{l=0}^r \hat{\lambda}_{l,r}  x^l, \quad \text{for} \quad r \geq 0,
\end{align*}
where for even $r$,
\begin{align}
\label{eq:coeff_lambda1}
    \hat{\lambda}_{l,r}  \coloneqq  \begin{cases}
    0, \quad & \text{for odd} \quad l, \\
    \sum\limits_{j=\frac{r-l}{2}}^{\frac{3 \left( r-l \right)}{2}} \hat{\alpha}_{j-\frac{r-l}{2}, l , j} \left( \frac{2}{45} \right)^{\frac{l}{2}} \Gamma \left( j+\frac{1}{2} \right), \quad & \text{for even} \quad l, 
    \end{cases}
\end{align}
and for odd $r$,
\begin{align}
\label{eq:coeff_lambda2}
    \hat{\lambda}_{l,r}  \coloneqq  \begin{cases}
    0, \quad & \text{for even} \quad l, \\
    \sum\limits_{j=\frac{r-l}{2}}^{\frac{3 \left( r-l \right)}{2}} \hat{\alpha}_{j-\frac{r-l}{2}, l , j} \left( \frac{2}{45} \right)^{\frac{l}{2}} \Gamma \left( j+\frac{1}{2} \right), \quad & \text{for odd} \quad l. 
    \end{cases}
\end{align}
Following the above discussion, \cref{eq:f_Z_asymp} can be rearranged as
\begingroup
\begin{align*}
f_{Z^P}(x) \sim & \frac{1}{4 \pi} \sqrt{\frac{2}{45}} \exp{\left( - \frac{1}{2} x^2 \right)}   \left( \sum_{j=0}^\infty \hat{A}_j \left( x \right) P^{-\frac{j}{2}} \right) \left( \sum_{r=0}^\infty \hat{B}_r \left( x \right) P^{-\frac{r}{2}} \right)  \\
\sim & \frac{1}{4 \pi} \sqrt{\frac{2}{45}} \exp{\left( - \frac{1}{2} x^2 \right)} \sum_{j=0}^\infty  \left( \sum_{r=0}^j \hat{A}_r \left( x \right) \hat{B}_{j-r} \left( x \right) \right) P^{- \frac{j}{2}} \\
\sim & \frac{1}{4 \pi} \sqrt{\frac{2}{45}} \exp{\left( - \frac{1}{2} x^2 \right)} \sum_{j=0}^\infty \hat{\psi}_j \left( x \right)  P^{- \frac{j}{2}}, 
\end{align*}
\endgroup
where for $j \geq 0$,
\begin{align*}
    \hat{\psi}_j \left( x \right)  \coloneqq  \sum_{r=0}^j \hat{A}_r \left( x \right) \hat{B}_{j-r} \left( x \right) 
     = \sum_{r=0}^j  \sum_{n=0}^r \sum_{l=0}^{j-r} \hat{\eta}_{n,r} \hat{\lambda}_{l,j-r} x^{2n+r+l}.
\end{align*}
Then, by the definition for asymptotic expansions, we have the order relation given below: for any $J \geq 0$,
\begin{align*}
  f_{Z^ P} \left( x \right) - \frac{1}{4 \pi} \sqrt{\frac{2}{45}} \exp{\left( - \frac{1}{2} x^2 \right)} \sum_{j=0}^{J} \hat{\psi}_j \left( x \right)  P^{-\frac{j}{2}} =  \mathrm{o} \left( P^{-\frac{J}{2}} \right),
\end{align*}
as $P \to + \infty$ with $x > - \sqrt{5P/2} \left( 1+ 3 / \pi^2 -3 \coth{\pi} / \pi \right)$ and $\left\vert x  \right\vert$ sufficiently small. Integrating on finite interval $\left( z_1, z_2 \right]$ such that $ z_i  > - \sqrt{5P/2} \left( 1+ 3 / \pi^2 -3 \coth{\pi} / \pi \right)$ and $\left\vert z_i \right\vert$ sufficiently small for $i=1,2$, we have
\begin{align*}
     \int_{z_1}^{z_2} f_{Z^P} \left( x \right) \, dx 
=  \frac{1}{4 \pi} \sqrt{\frac{2}{45}} \sum_{j=0}^J P^{-\frac{j}{2}} \int_{z_1}^{z_2}  \exp{\left( - \frac{1}{2} x^2 \right)}  \hat{\psi}_j \left( x \right) \, dx + \mathrm{o} \left( P^{-\frac{J}{2}} \right).
\end{align*}
Next, we show the integrals on the right hand side are finite. In fact, for $0 \leq j \leq J$, we can write
\begin{align*}
    \left\vert  \int_{z_1}^{z_2}  \exp{\left( - \frac{1}{2} x^2 \right)}  \hat{\psi}_j \left( x \right) \, dx  \right\vert 
\leq & \int_{z_1}^{z_2}  \exp{\left( - \frac{1}{2} x^2 \right)} \left\vert  \hat{\psi}_j \left( x \right) \right\vert \, dx \\
\leq & \sum_{r=0}^j  \sum_{n=0}^r \sum_{l=0}^{j-r} \left\vert \hat{\eta}_{n,r} \hat{\lambda}_{l,j-r}  \right\vert \int_{z_1}^{z_2}  \exp{\left( - \frac{1}{2} x^2 \right)}  \left\vert  x^{2n+r+l} \right\vert \, dx \\
< & + \infty.
\end{align*}
Notice that the constants $\hat{\eta}_{n,r}$ and $\hat{\lambda}_{l,j-r}$ are finite. Hence, we have the following asymptotic expansion: when $P \to +\infty$,
\begin{align}
  \int_{z_1}^{z_2}  f_{Z^ P} \left( x \right) \, dx   
  \sim &  \frac{1}{4 \pi} \sqrt{\frac{2}{45}} \sum_{j=0}^\infty P^{-\frac{j}{2}} \int_{z_1}^{z_2}  \exp{\left( - \frac{1}{2} x^2 \right)}  \hat{\psi}_j \left( x \right) \, dx \nonumber \\
  \sim &  \frac{1}{4 \pi} \sqrt{\frac{2}{45}} \sum_{j=0}^\infty P^{-\frac{j}{2}} \int_{z_1}^{z_2}  \exp{\left( - \frac{1}{2} x^2 \right)}  \sum_{r=0}^j  \sum_{n=0}^r \sum_{l=0}^{j-r} \hat{\eta}_{n,r} \hat{\lambda}_{l,j-r} x^{2n+r+l} \, dx \nonumber \\
  \sim & \frac{1}{4 \pi} \sqrt{\frac{2}{45}} \sum_{j=0}^\infty P^{-\frac{j}{2}}  \sum_{r=0}^j  \sum_{n=0}^r \sum_{l=0}^{j-r} \hat{\eta}_{n,r} \hat{\lambda}_{l,j-r} \int_{z_1}^{z_2}  \exp{\left( - \frac{1}{2} x^2 \right)}  x^{2n+r+l} \, dx. \label{eq:cdf_integral}
\end{align}
We apply the change of variable $v=x^2/2$ to compute the above integrals. For $z_1 < z_2 <0$ and $q \geq 0$, we have  
\begingroup
\begin{align*}
     \int_{z_1}^{z_2} \exp{ \left( -\frac{1}{2} x^2 \right)} x^q \, dx  
    = & \left( -1 \right)^q \left( \sqrt{2} \right)^{q-1} \int_{\frac{z_2^2}{2}}^{\frac{z_1^2}{2}}   \exp{\left( -v \right)} v^{ \frac{q+1}{2}  -1} \, dv \\
     = & \left( -1 \right)^q \left( \sqrt{2} \right)^{q-1} \left(  \gamma \left( \frac{q+1}{2}, \frac{\left(z_1\right)^2}{2}  \right) -  \gamma \left( \frac{q+1}{2}, \frac{\left( z_2 \right)^2}{2} \right) \right),
\end{align*}
\endgroup
where $\gamma \left( s, z \right)$ is the lower incomplete gamma function. For $z_1 < 0 \leq z_2 $ and $q \geq 0$, we consider the integral on the two sub-intervals $\left[ z_1,0 \right)$ and $ \left[ 0,z_2 \right]$ separately. By additivity, we get
\begin{align*}
     \int_{z_1}^{z_2} \exp{ \left( -\frac{1}{2} x^2 \right)} x^q \, dx  
    = & \left( -1 \right)^q \left( \sqrt{2} \right)^{q-1} \int_{0}^{\frac{z_1^2}{2}} \exp{\left( -v \right)} v^{ \frac{q+1}{2}  -1} \, dv  \\
    & + \left( \sqrt{2} \right)^{q-1} \int_{0}^{\frac{z_2^2}{2}} \exp{\left( -v \right)} v^{ \frac{q+1}{2}  -1} \, dv  \\
    = & \left( \sqrt{2} \right)^{q-1} \left( \left(-1 \right)^q \gamma \left( \frac{q+1}{2}, \frac{\left(z_1 \right)^2}{2} \right) +  \gamma \left( \frac{q+1}{2}, \frac{\left( z_2 \right)^2}{2} \right) \right).
\end{align*}
Substituting the explicit form for the integrals back into \cref{eq:cdf_integral} yields the stated representation, completing the proof.


\section{Proof of \cref{thm:Series_Exp_CDF_S^P}}
\label{sec:appendix_proof_3.5}
Recall that the probability density function $f_{S^P}$ of $S^P$ has the form
\begin{align*}
    f_{S^P} \left( y \right)  =\frac{1}{\sqrt{2 \pi}} \frac{2^P}{\Gamma \left(P\right)} y^{-\frac{1}{2}\left( P+2 \right)} \sum_{n=0}^\infty \frac{\Gamma \left( n+P \right)}{\Gamma \left( n+1 \right)} \exp{\left( -\frac{ \left( 2n+P \right)^2}{4y} \right)} D_{P+1} \left( \frac{2n+P}{\sqrt{y}} \right),
\end{align*}
in which $D_{P+1} \left( z \right)$ is a parabolic cylinder function of order $P+1$. The distribution function $F_{S^P}$ is derived by term-wise integration of the above series. First we show that $\sum_{n=0}^\infty \int_0^x \left\vert f_n \left( y \right) \right\vert \, dy$ $ < \infty$ for any finite $x \geq 0$, where
\begin{align*}
    f_n \left( y \right) \coloneqq \frac{\Gamma \left( n+P \right)}{\Gamma \left( n+1 \right)} \exp{\left( -\frac{ \left( 2n+P \right)^2}{4y} \right)} D_{P+1} \left( \frac{2n+P}{\sqrt{y}} \right)  y^{-\frac{1}{2}\left( P+2 \right)}.
\end{align*}
For fixed $n \geq 0$, application of the variable transformation $z = \left( 2 n +P \right) / \sqrt{y}$ gives 
\begin{align*}
     \int_0^x \left\vert f_n\left(y \right) \right\vert \, dy 
    =  2  \frac{\Gamma \left( n+P \right)}{\Gamma \left( n+1 \right)} \left( 2 n +P \right)^{-P} \int_{\frac{2n+P}{\sqrt{x}}}^{+ \infty}  \exp{\left( -\frac{1}{4} z^2 \right)} \left\vert D_{P+1} \left( z \right) \right\vert z^{P-1} \, dz.
\end{align*}
Notice that the parabolic cylinder function $D_{P+1} \left( z \right)$ is square integrable on $\left[ 0, \infty \right)$ (Gradshteyn and Ryzhik \cite[Chapter $7.711$]{gradshteyn2014table}), i.e. 
\begin{align*}
  \left\lVert D_{P+1} \right\rVert_2 \coloneqq \left( \int_0^{+\infty} \left\vert D_{P+1} \left( z \right) \right\vert^2 \, dz \right)^{\frac{1}{2}}< \infty.
\end{align*}
By H\"{o}lder's inequality, we have 
\begin{align}
    \int_{y}^{+ \infty}  \exp{\left( -\frac{1}{4} z^2 \right)} \left\vert D_{P+1} \left( z \right) \right\vert z^{P-1} \, dz \leq \left( \int_y^{+\infty} z^{2 \left( P-1 \right)} \exp{\left( - \frac{1}{2} z^2 \right)} \, dz \right)^{\frac{1}{2}}  \left\lVert D_{P+1} \right\rVert_2 \label{eq:convergence_inequality}
\end{align}
for $y \geq 0$. Next we consider the following two cases for $P$ separately: $P \in \left( 0, 1 \right)$ and $P \in \mathbb{N}$. 

For any $P \in \left( 0,1 \right)$, $z^{2 \left( P-1 \right)}$ is monotonically decreasing, which yields
\begin{align*}
    \int_{y}^{+ \infty}  \exp{\left( -\frac{1}{4} z^2 \right)} \left\vert D_{P+1} \left( z \right) \right\vert z^{P-1} \, dz 
  \leq &  y^{P-1} \left( \int_y^{+\infty}  \exp{\left( - \frac{1}{2} z^2 \right)} \, dz \right)^{\frac{1}{2}} \left\lVert D_{P+1} \right\rVert_2 \\
     = & \left( 2 \pi \right)^{\frac{1}{4}} \left\lVert D_{P+1} \right\rVert_2 y^{P-1}  \left( 1- \Phi \left( y \right) \right)^{\frac{1}{2}}
\end{align*}
for $y > 0$, where $\Phi \left( y \right)$ is the distribution function of a standard normal random variable. Then, it follows that the sequence $\int_0^x \left\vert f_n \left( y \right) \right\vert \, dy$ with finite $x$ is bounded by
\begin{align*}
      \int_0^x \left\vert f_n \left( y \right) \right\vert \, dy  
      \leq b_n, 
\end{align*}
where
\begin{align*}
    b_n \coloneqq 2 \left( 2 \pi \right)^{\frac{1}{4}} \left\lVert D_{P+1} \right\rVert_2 x^{\frac{1}{2} \left( 1- P \right)} \frac{\Gamma \left( n+P \right) }{\Gamma \left( n+1 \right)}  \frac{1}{2n+P}  \left( 1- \Phi \left( \frac{2 n+P}{\sqrt{x}} \right) \right)^{\frac{1}{2}}.
\end{align*}
By the ratio test, we can deduce that the series $\sum_{n=0}^\infty b_n$ is convergent. In fact, we have 
\begin{align*}
    \left\vert \frac{b_{n+1}}{b_n} \right\vert = \frac{n+P}{n+1}  \frac{2n+P}{2 \left( n+1 \right) +P}  \left( \frac{1- \Phi \left( \frac{2 \left(n+1 \right)+P}{\sqrt{x}} \right)}{1- \Phi \left( \frac{2 n+P}{\sqrt{x}} \right)} \right)^{ \frac{1}{2}} \to 0, \quad \text{as} \quad n \to \infty.
\end{align*}
The comparison test implies that the series $\sum_{n=0}^\infty \int_0^x \left\vert f_n \left( y \right) \right\vert \, dy $ is also convergent for any finite $x$. 

For any $P \in \mathbb{N}$, the integral on the right hand side of \cref{eq:convergence_inequality} can be regarded as the moment of some transformation of a standard normal random variable $Z$, i.e.
\begin{align*}
    \int_y^{+\infty} z^{2 \left( P-1 \right)} \exp{\left( - \frac{1}{2} z^2 \right)} \, dz  
    & = \sqrt{2 \pi} \mathbb{E} \left[ Z^{2 \left( P-1 \right)} \mathbf{1}_{\left\{ Z \geq y \right\}}  \right] \\
    & \leq  \sqrt{2 \pi} \left( \mathbb{E} \left[ Z^{4 \left( P-1 \right)} \right] \right)^{\frac{1}{2}}   \left( \mathbb{E} \left[ \left( \mathbf{1}_{ \left\{ Z \geq y \right\}}  \right)^2 \right] \right)^{\frac{1}{2}}  \\
    & = \sqrt{2 \pi} \left( \mathbb{E} \left[ Z^{4 \left( P-1 \right)} \right] \right)^{\frac{1}{2}} \left( 1 - \Phi \left( y \right) \right)^{\frac{1}{2}}
\end{align*}
for $y \geq 0$, where $\mathbf{1}_{\left\{ z \geq y \right\}}$ is the indicator function and the inequality follows from H\"{o}lder's inequality. Hence, the above argument gives the bounds for $\int_0^x \left\vert f_n \left( y \right) \right\vert \, dy$ with $x < \infty$ as
\begin{align*}
       \int_0^x \left\vert f_n \left( y \right) \right\vert \, dy  \leq b_n,
\end{align*}
where
\begin{align*}
    b_n \coloneqq 2 \left( 2 \pi \right)^{\frac{1}{4}}   \left\lVert D_{P+1} \right\rVert_2
    \left(  \mathbb{E} \left[ Z^{4 \left( P-1 \right)} \right]  \right)^{\frac{1}{4}}  \frac{\Gamma \left( n+P \right) }{\Gamma \left( n+1 \right)} \left( 2n+P \right)^{-P}  \left(1- \Phi \left( \frac{2n+P}{\sqrt{x}} \right)  \right)^{\frac{1}{4}}.
\end{align*}
Similarly, the series $\sum_{n=0}^\infty b_n$ converges by the fact that 
\begin{align*}
    \left\vert \frac{b_{n+1}}{b_n} \right\vert =  \frac{n+P}{n+1} \left(  \frac{2n+P}{2 \left( n+1 \right) +P} \right)^P \left( \frac{1- \Phi \left( \frac{2 \left(n+1 \right)+P}{\sqrt{x}} \right)}{1- \Phi \left( \frac{2 n+P}{\sqrt{x}} \right)} \right)^{ \frac{1}{4}} \to 0, \quad \text{as} \quad n \to \infty,
\end{align*}
implying the convergence of the series $\sum_{n=0}^\infty \int_0^x \left\vert f_n \left( y \right) \right\vert \, dy$ for finite $x$ as well.

Thus, for fixed $ 0 \leq x < \infty$ and $P \in \left( 0, 1 \right) \cup \mathbb{N}$, we have $\sum_{n=0}^\infty \int_0^x \left\vert f_n \left( y \right) \right\vert \, dy < \infty$. Then, applying a corollary of the Dominated Convergence Theorem (Rudin \cite[Theorem $1.38$]{rudin1987real}) yields
\begin{align*}
    \sum_{n=0}^\infty \int_0^x  f_n \left( y \right)  \, dy =  \int_0^x \sum_{n=0}^\infty  f_n \left( y \right)  \, dy,
\end{align*}
where the integration and summation can be interchanged. This leads to the following convergent series expansion for the distribution function $F_{S^P}$: for any $x < \infty$,
\begin{align*}
 F_{S^P} \left( x \right) 
  & = \frac{1}{\sqrt{2 \pi}} \frac{2^P}{\Gamma \left(P\right)}   \sum_{n=0}^\infty \int_0^x  f_n \left( y \right)  \, dy \\
   & = \frac{1}{\sqrt{2 \pi}} \frac{2^P}{\Gamma \left(P\right)} \sum_{n=0}^\infty \frac{\Gamma \left( n+P \right)}{\Gamma \left( n+1 \right)} \int_0^x y^{-\frac{1}{2}\left( P+2 \right)} \exp{\left( -\frac{ \left( 2n+P \right)^2}{4y} \right)} D_{P+1} \left( \frac{2n+P}{\sqrt{y}} \right) \, dy  \\
    & =  \frac{1}{\sqrt{2 \pi}} \frac{2^{P+1}}{\Gamma \left( P \right)} \sum_{n=0}^\infty \frac{\Gamma \left( n+P \right) }{\Gamma \left( n+1 \right)} \left( 2n+P \right)^{-P} G \left( \frac{2 n+P}{\sqrt{x}} \right),
\end{align*}
where the function $G \left( y \right)$ is defined as
\begin{align}
    G \left( y \right) = \int_y^{+\infty} z^{P-1} \exp{\left( - \frac{1}{4} z^2 \right)} D_{P+1} \left( z \right) \, dz. \label{eq:G_function}
\end{align}

To evaluate the function $G$, we first follow the methods mentioned earlier to calculate the parabolic cylinder functions $D_{P+1}$ and hence its integrand. We may replace $D_{P+1} \left( z \right)$ by its convergent power series on the entire interval of integration to derive the corresponding series expansion for $G$. However, the power series converges too slowly to be of practical use for large $z$. Instead, we split the interval of integration $\left[ y, + \infty \right)$ into two small elements, say $\left[ y, y^* \right)$ and $\left[ y^*, +\infty \right)$ for some sufficiently large $y^* \geq y$, where we apply different representations for $D_{P+1} \left( z \right)$ depending on the value of $z$. Then, we have
\begin{align*}
    G \left( y \right) =  G_1 \left( y, y^* \right) + G_2 \left( y^* \right),
\end{align*}
where 
\begin{align*}
    G_1 \left( y, y^* \right) & \coloneqq  \int_y^{y*} z^{P-1} \exp{\left( - \frac{1}{4} z^2 \right)} D_{P+1} \left( z \right) \, dz, \\
    G_2 \left( y^* \right) & \coloneqq \int_{y^*}^{+\infty} z^{P-1} \exp{\left( - \frac{1}{4} z^2 \right)} D_{P+1} \left( z \right) \, dz.
\end{align*}
On $\left[ y^*, +\infty \right)$, the asymptotic expansion \cref{eq:Cylinder_asymp} is a convenient way to compute $D_{P+1} \left( z \right)$ and hence $G_2$. On $\left[y, y^* \right)$, the power series \cref{eq:Cylinder_power} will be useful for $G_1$. Next, we consider the integral on the two sub-intervals case by case.

On $\left[ y^*, +\infty \right)$, we approximate the parabolic cylinder function $D_{p+1} \left( z \right)$ by its asymptotic series \cref{eq:Cylinder_asymp} on the entire interval under consideration. The series is then multiplied by $z^{P-1} \exp{\left( -z^2 /4 \right)}$ and integrated term by term to generate a series approximation for the integral $G_2$. To confirm that the resulting series is the correct asymptotic expansion for the integral with large $y^*$, we note that
\begin{align*}
z^{P-1} \exp{\left( -\frac{1}{4} z^2 \right)} D_{P+1} \left( z \right) \sim z^{2 P}   \exp{\left( -\frac{1}{2} z^2 \right)} \sum_{k=0}^\infty \hat{a}_k z^{-2k}, \quad \text{as} \quad z \to +\infty,
\end{align*}
where $\hat{a}_k = \left( -1 \right)^k 2^{-k} \left( - \left( P+1 \right)\right)_{2k}  /k! $ for $k \geq 0$. By the definition for asymptotic expansions, we have for each $K$, 
\begin{align*}
    z^{P-1} \exp{\left( -\frac{1}{4} z^2 \right)} D_{P+1} \left( z \right)  - z^{2 P}  \exp{\left( -\frac{1}{2} z^2 \right)}  \sum_{k=0}^K \hat{a}_k z^{-2k} = \mathrm{o} \left( z^{2 P}  \exp{\left( -\frac{1}{2} z^2 \right)}  z^{-2 K} \right),
\end{align*}
as $z \to + \infty$, meaning that for any $\epsilon > 0$, there exists a $z_0 > 0$ such that for $z > z_0$,
\begin{align*}
\left\vert  z^{P-1} \exp{\left( -\frac{1}{4} z^2 \right)} D_{P+1} \left( z \right)  - z^{2 P}  \exp{\left( -\frac{1}{2} z^2 \right)}  \sum_{k=0}^K \hat{a}_k z^{-2k} \right\vert \leq \epsilon \left\vert  z^{2 P}  \exp{\left( -\frac{1}{2} z^2 \right)}  z^{-2 K} \right\vert.
\end{align*}
Then, properties of integration yield for any $y^* > z_0$,
\begin{align*}
\lefteqn{\left\vert G_2 \left( y^* \right) - \sum_{k=0}^K \hat{a}_k \int_{y^*}^{+\infty}  z^{2 P}  \exp{\left( -\frac{1}{2} z^2 \right)}  z^{-2 k} \, dz \right\vert } \\
   \leq & \int_{y^*}^{+\infty} \left\vert    z^{P-1} \exp{\left( -\frac{1}{4} z^2 \right)} D_{P+1} \left( z \right) -  z^{2 P}  \exp{\left( -\frac{1}{2} z^2 \right)} \sum_{k=0}^K \hat{a}_k z^{-2 k}           \right\vert          \, dz \\
   \leq & \epsilon \int_{y^*}^{+\infty} z^{2P} \exp{\left( -\frac{1}{2} z^2 \right)}  z^{-2 K}    \, dz.
\end{align*}
Hence, we have the following asymptotic relation: as $y^* \to +\infty$,
\begin{align*}
 G_2 \left( {y^*} \right) - \sum_{k=0}^K \hat{a}_k \int_{y^*}^{+\infty}  z^{2 P}  \exp{\left( -\frac{1}{2} z^2 \right)}  z^{-2 k} \, dz = \mathrm{o} \left(  \int_{y^*}^{+\infty} z^{2P} \exp{\left( -\frac{1}{2} z^2 \right)}  z^{-2 K}    \, dz \right),
\end{align*}
which further gives the asymptotic expansion
\begin{align*}
    G_2 \left( {y^*} \right) \sim  \sum_{k=0}^{\infty} \hat{a}_k \int_{y^*}^{+\infty}  z^{2 P}  \exp{\left( -\frac{1}{2} z^2 \right)}  z^{-2 k} \, dz, \quad \text{as} \quad {y^*} \to +\infty.
\end{align*}
Introducing a new variable $\zeta = z^2/2$, we find that
\begin{align*}
     G_2 \left( {y^*} \right) 
      & \sim \sum_{k=0}^{\infty} \hat{a}_k 2^{P-k-\frac{1}{2}}  \int_{\frac{\left( y^* \right)^2}{2}}^{+\infty} \zeta^{P-k-\frac{1}{2}} \exp{ \left( - \zeta \right)} \, d \zeta \\
     & \sim  \sum_{k=0}^{\infty} \hat{a}_k 2^{P-k-\frac{1}{2}} \Gamma \left( P-k+\frac{1}{2}, \frac{\left( y^* \right)^2}{2} \right)
\end{align*}
in the limit $y^* \to +\infty$, where $\Gamma\left( s,z \right)$ is the upper incomplete gamma function. Replacing $ \hat{a}_k$ by the explicit form given above generates the stated asymptotic expansion for $G_2$.

On $ \left[y , y^* \right)$, approximating the parabolic cylinder function $D_{P+1} \left( z \right)$ using the power series \cref{eq:Cylinder_power}, we can write
\begin{align*}
      G_1 \left( y, y^* \right)
    = & \int_{y}^{y^*} z^{P-1} \exp{\left(-\frac{1}{4} z^2 \right)}  \left( \sum_{k=0}^\infty \hat{d}_k \left( P \right) z^k \right) \, dz \\
    = & \sum_{k=0}^\infty \hat{d}_k \left( P \right)  \int_{y}^{y^*} z^{P+k-1} \exp{\left( -\frac{1}{4} z^2\right)} \, dz \\
     = & \sum_{k=0}^\infty \hat{d}_k \left( P \right) 2^{P+k-1} \int_{ \frac{y^2}{4}}^{\frac{{\left( y^* \right)}^2}{4}}    \zeta^{\frac{ P+k }{2} -1} \exp{\left( - \zeta \right)} \, d \zeta \\
     = & \sum_{k=0}^\infty \hat{d}_k \left( P \right) 2^{P+k-1} \left(   \Gamma \left(  \frac{ P+k }{2} , \frac{y^2}{4}  \right) - \Gamma \left( \frac{ P+k }{2} , \frac{\left( {y^*} \right)^2}{4}  \right) \right),
\end{align*}
where the interchange of integration and summation in the second step follows from the fact that the power series is uniformly convergent over the interval of integration and a change of variable $\zeta = z^2/4$ is applied for the third step. Notice that the above series is convergent for any $0 < y \leq y^{*}$.


\section{Proof of \cref{thm:leading_CDF_S^P}}
\label{sec:appendix_proof_3.7}
Recall from \cref{thm:Series_Exp_CDF_S^P} that the distribution function $F_{S^P}$ for $P \in \left(0, 1 \right) \cup \mathbb{N} $ takes the form 
\begin{align}
    F_{S^P} \left( x \right) =   \frac{1}{\sqrt{2 \pi}} \frac{2^{P+1}}{\Gamma \left( P \right)} \sum_{n=0}^\infty \frac{\Gamma \left( n+P \right)}{\Gamma \left( n+1 \right)} \left( 2n+P \right)^{-P} G \left( \frac{2 n+P}{\sqrt{x}} \right) \label{eq:series_F_S^P_small2}
\end{align}
for any $0 \leq x < \infty$, where the function $G$ has the asymptotic approximation
\begin{align*}
G \left( y \right) \sim \sum_{k=0}^\infty \left( -1 \right)^k
        \frac{ \left(- \left( P+1 \right) \right)_{2k} }{  k!}  
           2^{P-2k-\frac{1}{2}}  \Gamma \left( P-k+\frac{1}{2}, \frac{ y ^2}{2} \right), \quad \text{as} \quad y \to +\infty. 
\end{align*}
Then it follows from the definition for asymptotic expansions that
\begin{align*}
    G \left( y \right) = 2^{P-\frac{1}{2}} \Gamma \left( P + \frac{1}{2}, \frac{y^2}{2} \right) + \mathrm{o} \left(\Gamma \left( P + \frac{1}{2}, \frac{y^2}{2} \right) \right), \quad \text{as} \quad y \to +\infty.
\end{align*}
Further, by the asymptotic expansion for the incomplete gamma function (Abramowitz and Stegun \cite[formula $(6.5.32)$]{abramowitz1964handbook})
\begin{align*}
 \Gamma \left( s,z \right) \sim z^{s-1} \exp{ \left( -z \right)} \sum_{k=0}^\infty \frac{\Gamma \left( s \right)}{\Gamma \left( s-k \right)} z^{-k}, \quad \text{as} \quad z \to + \infty,   
\end{align*}
we have 
\begin{align*}
    \Gamma \left( P + \frac{1}{2}, \frac{y^2}{2} \right) = \left( \frac{y^2}{2} \right)^{P-\frac{1}{2}} \exp{\left( - \frac{y^2}{2} \right)} + \mathrm{o} \left( y^{2P-1} \exp{\left( - \frac{y^2}{2} \right)} \right), \quad \text{as} \quad y \to + \infty. 
\end{align*}
The above analysis yields 
\begin{align*}
    G \left( y \right) = y^{2P-1} \exp{\left( - \frac{y^2}{2} \right)} + \mathrm{o} \left( y^{2P-1} \exp{\left( - \frac{y^2}{2} \right)} \right), \quad \text{as} \quad y \to + \infty.
\end{align*}
Hence, we can write
\begin{align*}
    G \left( \frac{2n+P}{\sqrt{x}} \right) = & \left(2n+P \right)^{2P-1} x^{\frac{1}{2} -P} \exp{\left( - \frac{\left(2n+P \right)^2}{2x} \right)}  \\
    & + \mathrm{o} \left( \left(2n+P \right)^{2P-1} x^{\frac{1}{2} -P} \exp{\left( - \frac{\left(2n+P \right)^2}{2x} \right)} \right)
\end{align*}
in the limit $x \to 0^+$. The observation 
\begin{align*}
     \left(2n+P \right)^{2P-1} \frac{\exp{\left( - \frac{\left(2n+P \right)^2}{2x} \right)}}{\exp{\left( - \frac{P^2}{2x} \right)}} =  \left(2n+P \right)^{2P-1} \exp{\left( - \frac{4n^2 + 4nP}{2x} \right)} \to 0,  \quad \text{as} \quad x \to 0^+
\end{align*}
for any $n \geq 1$ establishes
\begin{align*}
     G \left( \frac{2n+P}{\sqrt{x}} \right)  = \mathrm{o} \left( G \left(  \frac{P}{ \sqrt{x}} \right) \right),  \quad \text{as} \quad x \to 0^+.
\end{align*}
Therefore, \cref{eq:series_F_S^P_small2} becomes
\begin{align*}
  F_{S^P} \left( x \right) & =   \frac{1}{\sqrt{2 \pi}} \frac{2^{P+1}}{\Gamma \left( P \right)}  \frac{\Gamma \left( P \right)}{\Gamma \left( 1 \right)} P ^{-P} G \left( \frac{P}{\sqrt{x}} \right) + \mathrm{o} \left( G \left( \frac{P}{\sqrt{x}} \right) \right) \\
  & = \frac{1}{\sqrt{2 \pi}} 2^{P+1} P^{P-1} x^{\frac{1}{2} - P} \exp{\left( - \frac{P^2}{2x} \right)} + \mathrm{o} \left( x^{\frac{1}{2} - P} \exp{\left( - \frac{P^2}{2x} \right)} \right) \\
  & \sim \frac{1}{\sqrt{2 \pi}} 2^{P+1} P^{P-1} x^{\frac{1}{2} - P} \exp{\left( - \frac{P^2}{2x} \right)}, \quad \text{as} \quad x \to 0^+.
\end{align*}

\section{Uniqueness and Existence of Saddle Points} 
\label{sec:appendix_unique_existence}
In this section, we show the existence and uniqueness of saddle points, which is used in the proof of \cref{thm:Asymp_Exp_pdf_Z^P} in \cref{sec:appendix_proof_3.2}.
\begin{lemma} 
\label{lem:existence_uniqueness}
For $\beta > - {1}/{3}  $, there exists a unique solution to $\rho^{\prime} \left( z; \beta \right) =0$ in the domain $\left\{ z \in \mathbb{C}:  \mathrm{Im} \left( z \right) < \pi^2  \right\}$, where
\begin{align*}
\rho \left( z; \beta \right) = \log{ \left( \frac{\sqrt{z \mathrm{i}}}{\sinh{\sqrt{z \mathrm{i}}}} \right)} + z \mathrm{i} \left( \frac{1}{6} + \frac{1}{2} \beta \right).
\end{align*}
Further, this unique solution lies on the imaginary axis and is at zero when $\beta =0$. 
\end{lemma}

\begin{proof}
Direct differentiation of $\rho\left( z; \beta \right)$ gives
\begin{align*}
   \rho^{\prime} \left( z; \beta \right) & = \frac{\sinh{\sqrt{z \mathrm{i}}}}{ \sqrt{z \mathrm{i}} }   \frac{ \sinh{ \sqrt{z \mathrm{i}}} -  \sqrt{z \mathrm{i}} \cosh{ \sqrt{z \mathrm{i}}}   }{ \sinh^2{ \sqrt{z \mathrm{i}} }   } \frac{\mathrm{i}}{2  \sqrt{z \mathrm{i}}} +  \mathrm{i} \left( \frac{1}{6} + \frac{1}{2} \beta \right) \\
   & = \frac{\mathrm{i}}{2} \left( \frac{ 1 - \sqrt{z \mathrm{i}} \coth{ \sqrt{z \mathrm{i}} }   }{ z \mathrm{i}} + \left(\frac{1}{3} + \beta \right)  \right).
\end{align*}
For convenience set $\hat{\beta} \coloneqq \beta + 1/3$ and 
\begin{align*}
    \theta\left( z \right) := \frac{1- \sqrt{z\mathrm{i}} \coth{\sqrt{z\mathrm{i}}}}{z\mathrm{i}}.
\end{align*}
Then $\rho^{\prime} \left( z; \beta \right)$ becomes
\begin{align*}
    \rho^{\prime} \left( z; \beta \right) = \frac{\mathrm{i}}{2} \left( \theta \left( z \right) + \hat{\beta} \right).
\end{align*}

First, we note $\theta\left( z \right)$ is well-defined for all $z \in \mathbb{C}$ except at a countable set of poles which we discuss now. From Weierstrass Product Theorem in complex variable theory we have the classical identity 
\begin{align*}
z \coth z = 1 + 2 z^2 \sum_{n=1}^{\infty} \frac{1}{z^2 + n^2 \pi^2}.
\end{align*}
Hence we know 
\begin{align}
    \theta \left( z \right) = -2 \sum_{n=1}^{\infty} \frac{1}{z \mathrm{i} + n^2 \pi^2}. \label{eq:func_theta}
\end{align}
Then we observe the only singularities of $\theta \left( z \right)$ and hence of $\rho^{\prime} \left( z; \beta \right) $ are at 
$z = n^2 \pi^2 \mathrm{i}$ for $n \geq 1$.

Second, if we set $z = u + \mathrm{i } v$ for $u, v \in \mathbb{R}$, equation (\ref{eq:func_theta}) can be written as 
\begin{align*}
    \theta \left( z \right) = -2 \sum_{n=1}^\infty \frac{1}{ u \mathrm{i} - v + n^2 \pi^2} =-2 \sum_{n=1}^\infty \frac{ n^2 \pi^2 -v - \mathrm{i} u }{  \left( n^2 \pi^2 -v \right)^2 + u^2 },
\end{align*}
the imaginary part of which is
\begin{align*}
    \mathrm{Im} \theta \left( z \right) = 2 u \sum_{n=1}^\infty \frac{1}{\left( n^2 \pi^2 -v \right)^2 + u^2}.
\end{align*}
At a zero of $\rho^{\prime} \left( z; \beta \right)$ we must have $\theta \left( z \right) = - \hat{\beta}$, where $\hat{\beta} \in \mathbb{R}$. In other words, $\theta \left( z \right)$ must be real, i.e. its imaginary part must be zero. Then from the above equation, we see this occurs if and only if $u = 0$, i.e. $z$ is on the imaginary axis. Hence all zeros of $\rho^{\prime} \left( z; \beta \right)$ lie on the imaginary axis.   

Third, now let us focus entirely on the imaginary axis and set $z = \mathrm{i} v$ for $v \in \mathbb{R}$. Then $\theta \left( z \right)$ becomes 
\begin{align*}
    \theta \left( \mathrm{i} v \right) = -2 \sum_{n=1}^\infty \frac{1}{  n^2 \pi^2 -v }.
\end{align*}
Since term by term $\left( n^2 \pi^2 - v \right)^{-1}$ decreases monotonically in magnitude as $v$ moves in the direction from $\pi^2$ to $-\infty$, so $\theta \left( \mathrm{i}v \right)$ increases monotonically as $v$ moves in the same direction. Further, each term $\left( n^2 \pi^2 - v \right)^{-1} > 0$ for $v < \pi^2$. Hence $\theta \left( \mathrm{i}v \right) < 0$ when $v < \pi^2$. If we set $v = - \nu$ with $\nu > 0$ then from its definition, we observe
\begin{align*}
    \theta \left( \mathrm{i} v \right) =\frac{1}{\nu} -  \frac{ \coth \sqrt{\nu}}{\sqrt{ \nu}} \to 0, \quad \text{as} \quad \nu \to +\infty
\end{align*}
since $\coth{\sqrt{\nu}} \to 1$ as $\nu \to +\infty$. This implies $\theta \left( \mathrm{i} v \right) \to 0$ as $v \to - \infty$. At $v=0 $ by using that $ \sum_{n=1}^\infty n^{-2} = \pi^2/6$, we see
\begin{align*}
\theta \left(  0 \right) = -\frac{1}{3}. 
\end{align*}
Thus $\theta \left( \mathrm{i} v \right)$ decreases monotonically from zero as $v \to -\infty$ to $-1/3$ at $v= 0$ and continues to decrease monotonically to $- \infty$ as $v \to \pi^2$.

Hence for $\beta > -1/3 $, i.e. $\hat{\beta} > 0$, there is one and only one solution to $\theta \left( z \right) = - \hat{\beta} $, i.e. to $\rho^{\prime} \left( z; \beta \right)  = 0$ in the region $\left\{ z \in \mathbb{C}:  \mathrm{Im} \left( z \right) < \pi^2  \right\}$. This unique solution must lie on the imaginary axis and is at zero when $\beta =0$. 
\end{proof}

\begin{remark}
We note that the condition $\beta > -1/3$ is satisfied in the domain where the density $f_{Z^P}$ is strictly positive, and hence this condition does not constitute a restriction.  
\end{remark}

\begin{remark}
\label{rem:beta_restriction}
We have shown the imaginary part $v = v \left( \beta \right)$ of the solution $z_0 = z_0 \left( \beta \right)$ to $\rho^{\prime} \left( z; \beta \right)  = 0$ is a monotonically increasing function of $\beta$. To ensure $v$ is within the radius of convergence of the Taylor series expansion of $\rho$ and $\rho^{\prime}$, we must require $v > - \pi^2$, and hence we require $- \left( \beta + 1/3 \right) = \theta \left( \mathrm{i} v \right) <\theta \left(- \mathrm{i} \pi^2 \right)  = \left( 1- \pi \coth{\pi} \right) / \pi^2 $, i.e. $\beta > -1/3 - \left( 1- \pi \coth{\pi} \right) / \pi^2$.
\end{remark}

\section*{Acknowledgements}
We would like to thank the referees for their useful comments and suggestions which helped to improve the original manuscript.

\bibliographystyle{siamplain}

\newpage
\section{Supplementary Materials}
In the tables below, we quote the Chebyshev coefficients of the approximations to the inverse distribution functions for the (standardised) sum ($Z^P$) $S^P$ . Note that the $u$ denotes the right boundary point of each regime.

\FloatBarrier
\begin{table}[H]
{\footnotesize
 \captionsetup{position=top} 
 \caption{Chebyshev coefficients $c_n$ for $P=1$.}
   \begin{center}
   \FloatBarrier
   \subfloat{ 
\begin{tabular}{cccc} \hline \hline
     \multirow{2}{*}{$n$} & \multirow{2}{*}{left}  & \multirow{2}{*}{central} &   \multirow{2}{*}{middle}\\ \\
          \hline
0  & 1.870164486816790e-01  & 4.979491420716220e-01  & 9.481879998153620e-01  \\
1  & 7.543713026654420e-02  & 8.479024376573340e-02  & 1.331260262296900e-01  \\
2  & -8.902496689813970e-04 & 5.515061682001420e-03  & -2.432806629044970e-04 \\
3  & 5.081110592347190e-04  & 1.234615565280710e-03  & 7.536358417854170e-05  \\
4  & -6.749946441432770e-05 & 7.826573846201580e-05  & -1.832873207509700e-05 \\
5  & 1.869595919328090e-05  & 3.190249135857450e-05  & 3.735655017085910e-06  \\
\hline \hline     
    \end{tabular}}
    \FloatBarrier
    \end{center}
}
\end{table}

    \hfill

\begin{table}[tbhp]     \ContinuedFloat
{\footnotesize
 \captionsetup{position=top} 
 \caption{(cont.) Chebyshev coefficients $c_n$ for $P=1$.}
   \begin{center}    
    \FloatBarrier
    \subfloat{ 
     \begin{tabular}{ccc} \hline \hline
       \multirow{2}{*}{$n$} & \multirow{2}{*}{right tail} & \multirow{2}{*}{right tail} \\ \\
           \hline
0  & 2.147319302658020e+00  & 7.301291887179110e+00  \\
1  & 4.666085467273610e-01  & 2.110381075693420e+00  \\
2  & -5.745397085771070e-06 & -3.401385019474830e-12 \\
3  & 3.926281274449900e-06  & 3.135451774873670e-12  \\
4  & -2.335817034793870e-06 & -2.797671851935420e-12 \\
5  & 1.221893569112290e-06  & 2.415971255636180e-12  \\
\hline \hline       
      \end{tabular}} 
      \FloatBarrier
      
    \end{center}
}
\end{table}
\FloatBarrier

\begin{table}[tbhp]
{\footnotesize
 \captionsetup{position=top} 
 \caption{Chebyshev coefficients $c_n$ for $P=10$.}
   \begin{center}
   \subfloat{ 
    \begin{tabular}{ccc} \hline \hline
     \multirow{2}{*}{$n$} & \multirow{2}{*}{left tail}  & \multirow{2}{*}{left central}  \\ \\
          \hline
0                  & -4.186236205146250e+00 & -1.966613010380010e+00 \\
1                  & 1.605190906754510e+00  & -1.483681112879300e-02 \\
2                  & 1.059708264495660e-02  & -5.446775372602540e-03 \\
3                  & -2.569138906164220e-02 & 2.105452172418740e-06  \\
4                  & 2.657630896552880e-03  & -3.250251439573260e-05 \\
5                  & 5.038422707739470e-04  & 6.250462841383680e-07  \\
\hline \hline     
    \end{tabular}}

    \subfloat{ 
     \begin{tabular}{ccc} \hline \hline
       \multirow{2}{*}{$n$} & \multirow{2}{*}{right central} & \multirow{2}{*}{right tail} \\ \\
           \hline
0                  & 2.105468260211710e+00  & 1.013961870007640e+01  \\
1                  & 5.526660297932240e-02  & 5.881348912778960e+00  \\
2                  & 6.562580263096990e-03  & 1.613214220004240e+00  \\
3                  & 4.554851569107730e-04  & 3.583662354665820e-01  \\
4                  & 4.959544883969950e-05  & 6.562610713857870e-02  \\
5                  & 4.554114153141270e-06  & 1.013276567637800e-02  \\
\hline \hline       
      \end{tabular}} 
   \end{center}
}
\end{table}

\begin{table}[tbhp]
{\footnotesize
 \captionsetup{position=top} 
 \caption{Chebyshev coefficients $c_n$ for $P=50$.}
   \begin{center}
   \subfloat{ 
    \begin{tabular}{ccc} \hline \hline
     \multirow{2}{*}{$n$} & \multirow{2}{*}{left tail}  & \multirow{2}{*}{left central} \\ \\
          \hline
0  & -5.288942152307260e+00 & -2.000965923235480e+00  \\
1  & 2.354558187240890e+00  & 5.027272329447840e-03   \\
2  & -2.148300929212860e-01 & -5.388241960386400e-03   \\
3  & 1.167801704830350e-03  & 1.321210600565510e-04    \\
4  & 4.271918579401850e-03  & -3.328258515654870e-05   \\
5  & -7.367600678980360e-04 & 1.691117475054500e-06    \\
6  & 3.190352612884990e-05  & -2.758671231347400e-07   \\
7  & 1.049570860070370e-05  & 2.011261792238410e-08    \\
8  & -2.661096411008500e-06 & -2.690233651951720e-09   \\
9  & 2.849363340335400e-07  & 2.380979895569440e-10    \\
10 & 6.003686941040730e-09  & -2.896342891338410e-11   \\
11 & -9.314474427911970e-09 & 2.851594174388340e-12    \\
12 & 2.078536143889170e-09  & -3.279640552952750e-13   \\
13 & -2.199997312217100e-10 &                          \\
14 & -1.383145073419240e-11 &                          \\
15 & 1.142880443224370e-11  &                          \\
16 & -2.619066923423380e-12 &                          \\
17 & 3.160327088396940e-13  &                          \\
                   & $u $                     &                        \\
                   & 3.197003565947650e-01  & 5.169883337643990e-01  \\
                   & $k_1$                   & $k_1$                   \\
                   & -6.266377357307220e-01 & -4.044263478442090e+00 \\
                   & $k_2$                   & $k_2$                   \\
                   & 1.082311211678360e+00  & 9.999999999999990e-01  \\ 
     \hline \hline     
    \end{tabular}}

    \subfloat{ 
     \begin{tabular}{ccc} \hline \hline
       \multirow{2}{*}{$n$} & \multirow{2}{*}{right central} & \multirow{2}{*}{right tail} \\ \\
           \hline
0  & 2.062558090116340e+00  & 8.002513934544340e+00  \\
1  & 3.614884716178920e-02  & 4.310172977793670e+00  \\
2  & 5.905537215326120e-03  & 9.680044235587960e-01  \\
3  & 3.302424143714490e-04  & 1.854270427762370e-01  \\
4  & 4.120186343509380e-05  & 3.016976801851500e-02  \\
5  & 3.385543742585720e-06  & 4.333887776280000e-03  \\
6  & 3.843164000781120e-07  & 5.620910828054070e-04  \\
7  & 3.717029485793840e-08  & 6.549643858862510e-05  \\
8  & 4.112363385448630e-09  & 6.704595648415990e-06  \\
9  & 4.289267466426690e-10  & 5.921631388292800e-07  \\
10 & 4.741128697087930e-11  & 4.529653831824030e-08  \\
11 & 5.137505874838590e-12  & 3.042863599628060e-09  \\
12 & 5.647548074791590e-13  & 1.746343337679330e-10  \\
13 &                        & 9.109476838228590e-12  \\
14 &                        & 8.873841036286460e-13  \\
15 &                        &                        \\
16 &                        &                        \\
17 &                        &                        \\
                   &                        &                        \\
                   & 7.027660432349340e-01  &                        \\
                   & $k_1$                   & $k_1$                   \\
                   & 4.294834741352050e+00  & 6.327817827711070e-01  \\
                   & $k_2$                   & $k_2$                   \\
                   & -1.000000000000000e+00 & -1.122310983565170e+00  \\
       \hline \hline       
      \end{tabular}} 
   \end{center}
}
\end{table}

\begin{table}[tbhp]
{\footnotesize
 \captionsetup{position=top} 
 \caption{Chebyshev coefficients $c_n$ for $P=5000$.}
   \begin{center}
   \subfloat{ 
    \begin{tabular}{ccc} \hline \hline
     \multirow{2}{*}{$n$} & \multirow{2}{*}{left tail}  & \multirow{2}{*}{left middle} \\ \\
          \hline
0  & -9.007526636815000e+00 & -2.320006146407520e+00 \\
1  & 2.127535755063470e+00  & 1.266023709561750e+00  \\
2  & -1.873229942604270e-01 & -1.152857342660610e-01 \\
3  & 1.220927893353500e-02  & 9.871525791146420e-03  \\
4  & -5.143519783570370e-04 & -6.330801524297140e-04 \\
5  & 1.053201736042580e-05  & 2.721250039684760e-05  \\
6  & 2.597870360642370e-07  & -5.511437046510560e-07 \\
7  & -3.977373581792380e-08 & -9.780799116408280e-09 \\
8  & 2.731783965620670e-09  & 7.487832322276510e-10  \\
9  & -1.221138438089840e-10 & -6.846285108906070e-11 \\
10 & 3.067277829880080e-12  & 1.419621555779360e-11  \\
11 & -1.657314975397310e-13 & -1.036284107015480e-12 \\
12 &                        & -3.034958826524900e-14 \\
13                 &                        &                        \\
14                 &                        &                        \\
15                 &                        &                        \\
16                 &                        &                        \\
17                 &                        &                        \\
18                 &                        &                        \\
19                 &                        &                        \\
                   & $u$                      &                        \\
                   & 5.000000000000000e-03  & 5.016991626360740e-01  \\
                   & $k_1$                   & $k_1$                   \\
                   & -1.210717168704150e+00 & -9.809650284289560e-01 \\
                   & $k_2$                   & $k_2$                   \\
                   & 3.018736842909130e+00  & 6.356505843676840e-01  \\          
          
     \hline \hline     
    \end{tabular}}

    \subfloat{ 
     \begin{tabular}{ccc} \hline \hline
       \multirow{2}{*}{$n$} & \multirow{2}{*}{right middle} & \multirow{2}{*}{right tail} \\ \\
           \hline
0  & 2.350506283752280e+00  & 9.391788080584070e+00  \\
1  & 1.288612184352350e+00  & 2.304276718789270e+00  \\
2  & 1.238550415075650e-01  & 2.244838261893660e-01  \\
3  & 1.126367631985240e-02  & 1.723724750804200e-02  \\
4  & 8.113776712522790e-04  & 1.029865654753350e-03  \\
5  & 4.574695148394010e-05  & 5.284626922206710e-05  \\
6  & 2.146343317877790e-06  & 2.632282130705390e-06  \\
7  & 1.059607255591610e-07  & 1.297109856324030e-07  \\
8  & 6.522023447796370e-09  & 6.005448267256560e-09  \\
9  & 3.433932966389800e-10  & 3.273100583196690e-10  \\
10 & 7.415030905551360e-12  & 5.306924125563730e-11  \\
11 & -5.333234673760680e-14 & 2.645356635940220e-11  \\
12 &                        & 1.613648499072340e-11  \\
13 &                        & 9.700093322305410e-12  \\
14 &                        & 5.663793575999130e-12  \\
15 &                        & 3.213742114800220e-12  \\
16 &                        & 1.771549706162470e-12  \\
17 &                        & 9.437229274020540e-13  \\
18 &                        & 4.731397672671720e-13  \\
19 &                        & 1.949505078193680e-13  \\
                   &                        &                        \\
                   & 9.950000000000000e-01  &                        \\
                   & $k_1$                   & $k_1$                   \\
                   & 9.857057212990940e-01  & 1.210983500602370e+00  \\
                   & $k_2$                   & $k_2$                   \\
                   & -6.435551648966840e-01 & -3.019180921864380e+00  \\           
           
       \hline \hline       
      \end{tabular}} 
   \end{center}
}
\end{table}

\begin{table}[tbhp]
{\footnotesize
 \captionsetup{position=top} 
 \caption{Chebyshev coefficients $c_n$ for $P=10^4$.}
   \begin{center}
   \subfloat{ 
    \begin{tabular}{ccc} \hline \hline
     \multirow{2}{*}{$n$} & \multirow{2}{*}{left tail}  & \multirow{2}{*}{left middle} \\ \\
          \hline
0  & -8.763208308698120e+00 & -2.117708885805710e+00 \\
1  & 2.272856177544270e+00  & 1.148679380728190e+00  \\
2  & -2.195485611789300e-01 & -9.719454925386660e-02 \\
3  & 1.600603323734860e-02  & 7.829228309183660e-03  \\
4  & -7.913864835076260e-04 & -4.795725327622320e-04 \\
5  & 2.398092020649720e-05  & 2.045360216276630e-05  \\
6  & -1.924215235351760e-07 & -5.013317907576720e-07 \\
7  & -3.279659242674000e-08 & 3.503726442772780e-09  \\
8  & 3.433751594624110e-09  & -1.253562317413470e-10 \\
9  & -1.944135457019680e-10 & -1.311511811147320e-12 \\
10 & 5.718823252641270e-12  & 5.140621840511810e-12  \\
11 & -3.006387310778820e-13 & -4.635014571234970e-13 \\
12                 &                        &                        \\
13                 &                        &                        \\
14                 &                        &                        \\
15                 &                        &                        \\
16                 &                        &                        \\
17                 &                        &                        \\
                   & $u$                      &                        \\
                   & 1.000000000000000e-02  & 5.012014905499710e-01  \\
                   & $k_1$                   & $k_1$                   \\
                   & -1.115994743243930e+00 & -1.054206693922630e+00 \\
                   & $k_2$                   & $k_2$                   \\
                   & 2.704324434390850e+00  & 6.099629843489590e-01  \\          
          
     \hline \hline     
    \end{tabular}}

    \subfloat{ 
     \begin{tabular}{ccc} \hline \hline
       \multirow{2}{*}{$n$} & \multirow{2}{*}{right middle} & \multirow{2}{*}{right tail} \\ \\
           \hline
0  & 2.134535481408690e+00  & 9.022503475215680e+00  \\
1  & 1.161232186820710e+00  & 2.401317040413510e+00  \\
2  & 1.019807405937260e-01  & 2.487767815506420e-01  \\
3  & 8.555012569847670e-03  & 2.028231950132690e-02  \\
4  & 5.664908014626390e-04  & 1.266220882300120e-03  \\
5  & 2.890763831163810e-05  & 6.623334821975800e-05  \\
6  & 1.181852005087120e-06  & 3.326141775983180e-06  \\
7  & 4.960639757165790e-08  & 1.667239146689460e-07  \\
8  & 2.822677760359780e-09  & 7.705454473071610e-09  \\
9  & 1.409816724283650e-10  & 3.768066578229160e-10  \\
10 & 1.871694274577610e-12  & 3.560180300524010e-11  \\
11 & -1.571409536591980e-13 & 1.210420538657680e-11  \\
12 &                        & 7.304190736146010e-12  \\
13 &                        & 4.532327470602990e-12  \\
14 &                        & 2.727718438162450e-12  \\
15 &                        & 1.578085621557310e-12  \\
16 &                        & 8.497824077877610e-13  \\
17 &                        & 3.683334903311600e-13  \\
                   &                        &                        \\
                   & 9.900000000000000e-01  &                        \\
                   & $k_1$                   & $k_1$                   \\
                   & 1.058073632931550e+00  & 1.115994743243930e+00  \\
                   & $k_2$                   & $k_2$                   \\
                   & -6.158684948176150e-01 & -2.704324434390850e+00 \\           
           
       \hline \hline       
      \end{tabular}} 
   \end{center}
}
\end{table}

\begin{table}[tbhp]
{\footnotesize
 \captionsetup{position=top} 
 \caption{Chebyshev coefficients $c_n$ for $P=10^5$.}
   \begin{center}
   \subfloat{ 
    \begin{tabular}{ccc} \hline \hline
     \multirow{2}{*}{$n$} & \multirow{2}{*}{left tail}  & \multirow{2}{*}{right tail} \\ \\
          \hline
0  & -5.706926051399320e+00 & 5.760300411129710e+00  \\
1  & 3.395553188579180e+00  & 3.436321449713240e+00  \\
2  & -6.300054999277720e-01 & 6.483188840749760e-01  \\
3  & 9.838818354304640e-02  & 1.035993426134920e-01  \\
4  & -1.136750936655030e-02 & 1.252833208873680e-02  \\
5  & 9.513883577249630e-04  & 1.161417162277510e-03  \\
6  & -6.015836221231600e-05 & 9.191380802217640e-05  \\
7  & 3.618126012554300e-06  & 7.737590464616790e-06  \\
8  & -1.697979831324770e-07 & 6.400595433239310e-07  \\
9  & -1.897253631189280e-08 & 2.904203939446630e-08  \\
10 & 3.326071860640010e-09  & 1.050861883541800e-09  \\
11 & 2.442263784022650e-10  & 6.008685289953490e-10  \\
12 & -3.833832227759890e-11 & 6.651321295963480e-11  \\
13 & -1.297154428240060e-11 & -1.031199486940060e-11 \\
14 & 1.412983301319550e-12  & -9.752972166933010e-13 \\
15 & 4.089569727968860e-13  & 5.588332180241710e-13  \\
                   & $u$                      &                        \\
                   & 5.003799449933160e-01  &                        \\
                   & $k_1$                   & $k_1$                   \\
                   & -5.424593574874610e-01 & 5.427821460191420e-01  \\
                   & $k_2$                   & $k_2$                   \\
                   & 8.005868439521350e-01  & -8.016582768171050e-01 \\          
          
     \hline \hline     
    \end{tabular}}
    \end{center}
}
\end{table}

\begin{table}[tbhp]
{\footnotesize
 \captionsetup{position=top} 
 \caption{Chebyshev coefficients $c_n$ for $P=10^6$.}
   \begin{center}
   \subfloat{ 
    \begin{tabular}{ccc} \hline \hline
     \multirow{2}{*}{$n$} & \multirow{2}{*}{left tail}  & \multirow{2}{*}{right tail} \\ \\
          \hline
0  & -5.946145460337920e+00 & 5.964470382846980e+00    \\
1  & 3.556004911984500e+00  & 3.570041770711440e+00    \\
2  & -6.801894131410160e-01 & 6.865586194977590e-01    \\
3  & 1.092510150797910e-01  & 1.110985275994380e-01    \\
4  & -1.307725205702110e-02 & 1.349685679309470e-02    \\
5  & 1.160328523663860e-03  & 1.237761453545620e-03    \\
6  & -8.240107479627320e-05 & 9.434859113712600e-05    \\
7  & 5.926003936265940e-06  & 7.508635778142890e-06    \\
8  & -3.827985039853400e-07 & 5.673276390309780e-07    \\
9  & -4.939610258605130e-09 & 1.429761666233160e-08    \\
10 & 2.191431812357590e-09  & -4.002728561697650e-10   \\
11 & 4.611932833470640e-10  & 6.107231293159660e-10    \\
12 & -4.959733545874770e-11 & 6.167427914388110e-11    \\
13 & -1.701165301087770e-11 & -1.590293060974530e-11   \\
14 & 1.492500677577930e-12  & -1.334659663440940e-12   \\
15 & 5.860252785293660e-13  & 6.358368604978550e-13    \\
                   & $u$                      &                        \\
                   & 5.001201491665880e-01  &                        \\
                   & $k_1$                   & $k_1$                   \\
                   & -5.310898090226540e-01 & 5.311876095528130e-01  \\
                   & $k_2$                   & $k_2$                   \\
                   & 8.051645964724630e-01  & -8.054970186947690e-01 \\          
          
     \hline \hline     
    \end{tabular}}
    \end{center}
}
\end{table}

\begin{table}[tbhp]
{\footnotesize
 \captionsetup{position=top} 
 \caption{Chebyshev coefficients $c_n$ for $h=2$.}
   \begin{center}
   \subfloat{ 
    \begin{tabular}{ccc} \hline \hline
     \multirow{2}{*}{$n$} & \multirow{2}{*}{left}  & \multirow{2}{*}{left} \\ \\
          \hline
0  & 3.746732186396810e-01  & 1.246285546643250e+00  \\
1  & 1.195997564921540e-01  & 3.358246308532060e-01  \\
2  & -8.070601474351940e-03 & 2.386732322681120e-02  \\
3  & 1.339603001568510e-03  & 1.294898023275990e-02  \\
4  & -2.506080830464610e-04 & 2.637423242240120e-03  \\
5  & 5.637743633696060e-05  & 1.091486949977880e-03  \\
\hline \hline     
    \end{tabular}}
    \end{center}
}
\end{table}

\begin{table}[tbhp] \ContinuedFloat
{\footnotesize
 \captionsetup{position=top} 
 \caption{(cont.) Chebyshev coefficients $c_n$ for $h=2$. }
   \begin{center}
   \subfloat{ 
    \begin{tabular}{ccc} \hline \hline
     \multirow{2}{*}{$n$} & \multirow{2}{*}{middle}  & \multirow{2}{*}{right tail} \\ \\
          \hline
0  & 3.527966142386780e+00  & 1.003328309377590e+01  \\
1  & 7.472142334253430e-01  & 2.496706503355630e+00  \\
2  & -1.333644417363060e-02 & -1.589949448670880e-02 \\
3  & 2.622018509548740e-03  & 3.113590178926820e-03  \\
4  & -6.027577883155430e-04 & -6.963073964368400e-04 \\
5  & 1.515242268297150e-04  & 1.678769292503920e-04  \\
\hline \hline     
    \end{tabular}}
    \end{center}
}
\end{table}

\begin{table}[tbhp]
{\footnotesize
 \captionsetup{position=top} 
 \caption{Chebyshev coefficients $c_n$ for $h=1/5$.}
   \begin{center}
   \subfloat{ 
    \begin{tabular}{cccc} \hline \hline
     \multirow{2}{*}{$n$} & \multirow{2}{*}{left}  & \multirow{2}{*}{central} &   \multirow{2}{*}{middle}\\ \\
          \hline
0                  & 1.551440631887020e-02  & 5.965851817537560e-02    & 2.365922039769380e-01   \\
          1                  & 7.154056831584140e-03  & 1.701537131814850e-02    & 7.509528018047980e-02   \\
          2                  & 1.120793580526900e-04  & 2.597262497204220e-03    & 6.749060031298230e-03   \\
          3                  & -2.767589127786680e-05 & 4.709183006272860e-04    & -9.756034732487190e-05  \\
          4                  & 5.310776527089690e-06  & 7.145518311775060e-05    & -4.546713797013440e-05  \\
          5                  & -1.242484981459680e-06 & 1.330260630541950e-05    & 2.265934249688400e-06   \\
\hline \hline     
    \end{tabular}}
    \end{center}
}
\end{table}

\begin{table}[tbhp] \ContinuedFloat
{\footnotesize
 \captionsetup{position=top} 
 \caption{(cont.) Chebyshev coefficients $c_n$ for $h=1/5$.}
   \begin{center}    
    \subfloat{ 
     \begin{tabular}{ccc} \hline \hline
       \multirow{2}{*}{$n$} & \multirow{2}{*}{right tail} & \multirow{2}{*}{right tail} \\ \\
           \hline
0                  & 9.766531067045780e-01       & 5.548556192934720e+00       \\
1                  & 3.018335097399800e-01       & 2.006312378503930e+00       \\
2                  & 1.137764493492600e-02       & 2.401553462574660e-02       \\
3                  & -1.796551902064940e-03      & -5.827238250950370e-03      \\
4                  & 2.961890615175630e-04       & 1.564108444981830e-03       \\
5                  & -3.906140490786650e-05      & -4.425250341488140e-04      \\
\hline \hline       
      \end{tabular}} 
   \end{center}
}
\end{table}

\begin{table}[tbhp]
{\footnotesize
 \captionsetup{position=top} 
 \caption{Chebyshev coefficients $c_n$ for $h=1/10$.}
   \begin{center}
   \subfloat{ 
    \begin{tabular}{cccc} \hline \hline
     \multirow{2}{*}{$n$} & \multirow{2}{*}{left}  & \multirow{2}{*}{central} &   \multirow{2}{*}{middle}\\ \\
          \hline
 0                  & 3.130510092436490e-03  & 2.226333345761950e-02    & 1.988982349173310e-01   \\
1                  & 1.417002017122940e-03  & 1.019151935447450e-02    & 8.748062377822970e-02   \\
2                  & 2.360674641499860e-05  & 2.673461828675060e-03    & 1.322178512127460e-02   \\
3                  & -6.764554780577100e-06 & 7.541005867012190e-04    & 1.783172989082840e-05   \\
4                  & 1.073524111115770e-06  & 1.825724014495570e-04    & -1.742061549930610e-04  \\
5                  & -2.165240827636890e-07 & 5.070406265568330e-05    & 1.493474602053700e-06   \\
\hline \hline     
    \end{tabular}}
    \end{center}
}
\end{table}

\begin{table}[tbhp] \ContinuedFloat
{\footnotesize
 \captionsetup{position=top} 
 \caption{(cont.) Chebyshev coefficients $c_n$ for $h=1/10$.}
   \begin{center}    
    \subfloat{ 
     \begin{tabular}{ccc} \hline \hline
       \multirow{2}{*}{$n$} & \multirow{2}{*}{right tail} & \multirow{2}{*}{right tail} \\ \\
           \hline
  0                  & 1.167088563196570e+00       & 5.557949689428240e+00       \\
1                  & 4.028821977072450e-01       & 1.804403214930680e+00       \\
2                  & 1.589726864878320e-02       & 2.001076313509200e-02       \\
3                  & -2.776444180529080e-03      & -4.157908839745580e-03      \\
4                  & 5.440877740378380e-04       & 9.533639401005070e-04       \\
5                  & -1.067963570631740e-04      & -2.298471294547880e-04      \\
\hline \hline       
      \end{tabular}} 
   \end{center}
}
\end{table}

\begin{table}[tbhp]
{\footnotesize
 \captionsetup{position=top} 
 \caption{Chebyshev coefficients $c_n$ for $h=1/20$.}
   \begin{center}
   \subfloat{ 
    \begin{tabular}{cccc} \hline \hline
     \multirow{2}{*}{$n$} & \multirow{2}{*}{left}  & \multirow{2}{*}{central} &   \multirow{2}{*}{middle}\\ \\
          \hline
0                  & 8.320108386176220e-04  & 1.028289222803410e-02    & 1.796716788945920e-01   \\
1                  & 3.797922844843720e-04  & 6.051387649842750e-03    & 9.215513210293030e-02   \\
2                  & 6.119152516315670e-06  & 2.386515701734710e-03    & 1.798443115218630e-02   \\
3                  & -2.216499535843730e-06 & 9.149383182711970e-04    & 3.536164396647120e-04   \\
4                  & 3.549796569970470e-07  & 3.260804312989750e-04    & -3.290230257519100e-04  \\
5                  & -6.909777679284210e-08 & 1.174302244391960e-04    & -9.190331115977200e-06  \\
\hline \hline     
    \end{tabular}}
    \end{center}
}
\end{table}

\begin{table}[tbhp] \ContinuedFloat
{\footnotesize
 \captionsetup{position=top} 
 \caption{(cont.) Chebyshev coefficients $c_n$ for $h=1/20$.}
   \begin{center}    
    \subfloat{ 
     \begin{tabular}{ccc} \hline \hline
       \multirow{2}{*}{$n$} & \multirow{2}{*}{right tail} & \multirow{2}{*}{right tail} \\ \\
           \hline
0                  & 9.761057883039990e-01       & 5.196073467821080e+00       \\
1                  & 3.017284864244540e-01       & 1.831141920349900e+00       \\
2                  & 1.166946937938660e-02       & 2.543065037748800e-02       \\
3                  & -1.687052247930830e-03      & -5.741796704905520e-03      \\
4                  & 2.745955374164150e-04       & 1.425454072455170e-03       \\
5                  & -4.255962891388790e-05      & -3.709730129757980e-04      \\
\hline \hline       
      \end{tabular}} 
   \end{center}
}
\end{table}

\begin{table}[tbhp]
{\footnotesize
 \captionsetup{position=top} 
 \caption{Chebyshev coefficients $c_n$ for $h=1/50$.}
   \begin{center}
   \subfloat{ 
    \begin{tabular}{cccc} \hline \hline
     \multirow{2}{*}{$n$} & \multirow{2}{*}{left}  & \multirow{2}{*}{central} &   \multirow{2}{*}{middle}\\ \\
          \hline
0                  & 3.342404855864130e-04  & 5.287504606339480e-03    & 1.679802218836370e-01   \\
1                  & 1.582119045977740e-04  & 3.609167042053250e-03    & 9.425159027828390e-02   \\
2                  & -4.034459826869400e-06 & 1.936145076792520e-03    & 2.146641276528250e-02   \\
3                  & -1.781674777546940e-06 & 9.633501098384940e-04    & 7.761771414171220e-04   \\
4                  & 6.111894067922420e-07  & 4.586062099584780e-04    & -4.715700923865780e-04  \\
5                  & -1.616413643230570e-07 & 2.128565367871240e-04    & -2.911833910065400e-05  \\
\hline \hline     
    \end{tabular}}
    \end{center}
}
\end{table}

\begin{table}[tbhp] \ContinuedFloat
{\footnotesize
 \captionsetup{position=top} 
 \caption{(cont.) Chebyshev coefficients $c_n$ for $h=1/50$.}
   \begin{center}    
    \subfloat{ 
     \begin{tabular}{ccc} \hline \hline
       \multirow{2}{*}{$n$} & \multirow{2}{*}{right tail} & \multirow{2}{*}{right tail} \\ \\
           \hline
  0                  & 9.761146326247330e-01       & 5.440971360286990e+00       \\
1                  & 3.016928648060250e-01       & 1.956911680269630e+00       \\
2                  & 1.167231718705550e-02       & 2.774260925258900e-02       \\
3                  & -1.651276452041110e-03      & -6.436890817427110e-03      \\
4                  & 2.666875570608340e-04       & 1.641299885337880e-03       \\
5                  & -4.247191460497490e-05      & -4.384480040071970e-04      \\
\hline \hline       
      \end{tabular}} 
   \end{center}
}
\end{table}

\begin{table}[tbhp]
{\footnotesize
 \captionsetup{position=top} 
 \caption{Chebyshev coefficients $c_n$ for $h=1/100$.}
   \begin{center}
   \subfloat{ 
    \begin{tabular}{cccc} \hline \hline
     \multirow{2}{*}{$n$} & \multirow{2}{*}{left}  & \multirow{2}{*}{central} &   \multirow{2}{*}{middle}\\ \\
          \hline
0                  & 1.064576523939530e-04  & 7.757528496943130e-04    & 6.583911279090610e-02   \\
1                  & 4.967483886634590e-05  & 3.963182148041290e-04    & 4.630386572794740e-02   \\
2                  & -2.574907330657650e-06 & 1.452820177102820e-04    & 1.757544141274230e-02   \\
3                  & -5.303715403034560e-07 & 4.855954814287430e-05    & 3.315284363718880e-03   \\
4                  & 2.663871929167150e-07  & 1.531863558980870e-05    & 1.968319378136490e-05   \\
5                  & -8.401958707782960e-08 & 4.676444520070540e-06    & -1.223428005218250e-04  \\
\hline \hline     
    \end{tabular}}
    \end{center}
}
\end{table}

\begin{table}[tbhp] \ContinuedFloat
{\footnotesize
 \captionsetup{position=top} 
 \caption{(cont.) Chebyshev coefficients $c_n$ for $h=1/100$.}
   \begin{center}    
    \subfloat{ 
     \begin{tabular}{ccc} \hline \hline
       \multirow{2}{*}{$n$} & \multirow{2}{*}{right tail} & \multirow{2}{*}{right tail} \\ \\
           \hline
0                  & 5.721713263001810e-01       & 5.012707101453440e+00       \\
1                  & 2.017820038440820e-01       & 2.062096651831810e+00       \\
2                  & 1.377375013917080e-02       & 3.999684930491530e-02       \\
3                  & -1.806370775567440e-03      & -1.094915283125730e-02      \\
4                  & 1.523593047177300e-04       & 3.280050342478840e-03       \\
5                  & 2.295752543436930e-05       & -1.026359592164560e-03      \\
\hline \hline       
      \end{tabular}} 
   \end{center}
}
\end{table}

\begin{table}[tbhp]
{\footnotesize
 \captionsetup{position=top} 
 \caption{Chebyshev coefficients $c_n$ for $h=1/200$.}
   \begin{center}
   \subfloat{ 
    \begin{tabular}{cccc} \hline \hline
     \multirow{2}{*}{$n$} & \multirow{2}{*}{left}  & \multirow{2}{*}{central} &   \multirow{2}{*}{middle}\\ \\
          \hline
 0                  & 3.518991274756410e-05  & 8.392520820142320e-04    & 3.802327783487690e-02   \\
1                  & 1.590990829653510e-05  & 6.431136826076690e-04    & 2.293406644160830e-02   \\
2                  & -1.491507257739220e-06 & 4.085596936397020e-04    & 6.950796412254640e-03   \\
3                  & -9.066540515692110e-08 & 2.399296017984750e-04    & 1.080933113927390e-03   \\
4                  & 1.118185077651380e-07  & 1.344982201802300e-04    & 3.968276530040410e-05   \\
5                  & -4.689629675507760e-08 & 7.316918708694770e-05    & -1.510396050891140e-05  \\
\hline \hline     
    \end{tabular}}
    \end{center}
}
\end{table}

\begin{table}[tbhp] \ContinuedFloat
{\footnotesize
 \captionsetup{position=top} 
 \caption{(cont.) Chebyshev coefficients $c_n$ for $h=1/200$.}
   \begin{center}    
    \subfloat{ 
     \begin{tabular}{ccc} \hline \hline
       \multirow{2}{*}{$n$} & \multirow{2}{*}{right tail} & \multirow{2}{*}{right tail} \\ \\
           \hline
 0                  & 4.996811103033990e-01       & 4.914258356400550e+00       \\
1                  & 2.283166265702920e-01       & 2.011738127244230e+00       \\
2                  & 2.514423513031150e-02       & 3.937356800757610e-02       \\
3                  & -3.474198991631110e-03      & -1.065266526129480e-02      \\
4                  & 5.077220214696990e-05       & 3.152291195579000e-03       \\
5                  & 1.613616422231280e-04       & -9.738895553647930e-04      \\
\hline \hline       
      \end{tabular}} 
   \end{center}
}
\end{table}

\begin{table}[tbhp]
{\footnotesize
 \captionsetup{position=top} 
 \caption{Chebyshev coefficients $c_n$ for $h=1/500$.}
   \begin{center}
   \subfloat{ 
    \begin{tabular}{cccc} \hline \hline
     \multirow{2}{*}{$n$} & \multirow{2}{*}{left}  & \multirow{2}{*}{central} &   \multirow{2}{*}{middle}\\ \\
          \hline
  0                  & 5.483749876010110e-06  & 6.081052138201630e-05    & 5.473769059145160e-02   \\
1                  & 2.486444309865320e-06  & 3.802911084933040e-05    & 4.351538755149600e-02   \\
2                  & -2.246254082499010e-07 & 1.823368632895480e-05    & 2.198780054496340e-02   \\
3                  & -1.578587968151560e-08 & 7.946990081244110e-06    & 6.662984075703530e-03   \\
4                  & 1.732923293003170e-08  & 3.273356396625710e-06    & 7.907173467155090e-04   \\
5                  & -7.081399642607440e-09 & 1.301100356483330e-06    & -2.243067264619820e-04  \\
\hline \hline     
    \end{tabular}}
    \end{center}
}
\end{table}

\begin{table}[tbhp] \ContinuedFloat
{\footnotesize
 \captionsetup{position=top} 
 \caption{(cont.) Chebyshev coefficients $c_n$ for $h=1/500$.}
   \begin{center}    
    \subfloat{ 
     \begin{tabular}{ccc} \hline \hline
       \multirow{2}{*}{$n$} & \multirow{2}{*}{right tail} & \multirow{2}{*}{right tail} \\ \\
           \hline
  0                  & 5.722628305041320e-01       & 4.619812351486580e+00       \\
1                  & 2.017761232436850e-01       & 1.860589481022910e+00       \\
2                  & 1.372549559855640e-02       & 3.705233566867140e-02       \\
3                  & -1.799486944230340e-03      & -9.679912077752170e-03      \\
4                  & 1.548179312896890e-04       & 2.763008630922080e-03       \\
5                  & 2.188831960478410e-05       & -8.227338295193510e-04      \\
\hline \hline       
      \end{tabular}} 
   \end{center}
}
\end{table}

\begin{table}[tbhp]
{\footnotesize
 \captionsetup{position=top} 
 \caption{Chebyshev coefficients $c_n$ for $h=1/1000$.}
   \begin{center}
   \subfloat{ 
    \begin{tabular}{cccc} \hline \hline
     \multirow{2}{*}{$n$} & \multirow{2}{*}{left}  & \multirow{2}{*}{central} &   \multirow{2}{*}{middle}\\ \\
          \hline
  0                  & 1.702183541035500e-06  & 4.683227357260460e-05    & 5.536218616018370e-02   \\
1                  & 7.420902874786450e-07  & 3.703626675908400e-05    & 4.374721584420760e-02   \\
2                  & -9.865066726948370e-08 & 2.459788457587220e-05    & 2.175988835715660e-02   \\
3                  & 2.116960181208090e-09  & 1.509706997971100e-05    & 6.418561137830710e-03   \\
4                  & 5.497757664443220e-09  & 8.833813587467830e-06    & 7.077802624591750e-04   \\
5                  & -3.075679409634450e-09 & 5.007785871509330e-06    & -2.234688941525440e-04  \\
\hline \hline     
    \end{tabular}}
    \end{center}
}
\end{table}

\begin{table}[tbhp] \ContinuedFloat
{\footnotesize
 \captionsetup{position=top} 
 \caption{(cont.) Chebyshev coefficients $c_n$ for $h=1/1000$.}
   \begin{center}    
    \subfloat{ 
     \begin{tabular}{ccc} \hline \hline
       \multirow{2}{*}{$n$} & \multirow{2}{*}{right tail} & \multirow{2}{*}{right tail} \\ \\
           \hline
 0                  & 5.722747148899260e-01       & 4.619756015894430e+00       \\
1                  & 2.017753450395620e-01       & 1.860594275953970e+00       \\
2                  & 1.371924357276870e-02       & 3.708005345183760e-02       \\
3                  & -1.798581149904210e-03      & -9.684906797069420e-03      \\
4                  & 1.551213775007570e-04       & 2.763643547482780e-03       \\
5                  & 2.175254385332560e-05       & -8.226383894790440e-04      \\
\hline \hline       
      \end{tabular}} 
   \end{center}
}
\end{table}

\begin{table}[tbhp]
{\footnotesize
 \captionsetup{position=top} 
 \caption{Chebyshev coefficients $c_n$ for $h=1/2000$.}
   \begin{center}
   \subfloat{ 
    \begin{tabular}{cccc} \hline \hline
     \multirow{2}{*}{$n$} & \multirow{2}{*}{left}  & \multirow{2}{*}{central} &   \multirow{2}{*}{middle}\\ \\
          \hline
0                  & 8.340782562108640e-08  & 2.468455017350810e-06    & 2.280294691861460e-02   \\
1                  & 3.790348335769390e-08  & 1.808380820210540e-06    & 1.944720833702610e-02   \\
2                  & 6.055403556177730e-10  & 9.977786192716350e-07    & 1.206597650481480e-02   \\
3                  & -2.457590587255690e-10 & 5.095262839987460e-07    & 5.365044701332920e-03   \\
4                  & 3.852931864752020e-11  & 2.442450448287350e-07    & 1.591201684585550e-03   \\
5                  & -7.060509039767990e-12 & 1.138166487478510e-07    & 2.189680324488170e-04   \\
\hline \hline     
    \end{tabular}}
    \end{center}
}
\end{table}

\begin{table}[tbhp] \ContinuedFloat
{\footnotesize
 \captionsetup{position=top} 
 \caption{(cont.) Chebyshev coefficients $c_n$ for $h=1/2000$.}
   \begin{center}    
    \subfloat{ 
     \begin{tabular}{ccc} \hline \hline
       \multirow{2}{*}{$n$} & \multirow{2}{*}{right tail} & \multirow{2}{*}{right tail} \\ \\
           \hline
0                  & 4.997785332154780e-01       & 4.423666819273100e+00       \\
1                  & 2.283167712864420e-01       & 1.759820399395020e+00       \\
2                  & 2.509022440364520e-02       & 3.541998526121500e-02       \\
3                  & -3.473302792894470e-03      & -9.020823900527130e-03      \\
4                  & 5.642398807896150e-05       & 2.508301403041810e-03       \\
5                  & 1.601477839497730e-04       & -7.271867741631950e-04      \\
\hline \hline       
      \end{tabular}} 
   \end{center}
}
\end{table}



\end{document}